\documentclass{lmcs}
\usepackage[utf8]{inputenc}
\pdfoutput=1

\usepackage{lastpage}
\lmcsdoi{19}{4}{18}
\lmcsheading{}{\pageref{LastPage}}{}{}%
{Mar.~11,~2022}{Dec.~01,~2023}{}

\pdfoutput=1


\usepackage{tikz}
\usetikzlibrary{arrows,shapes,automata,calc,decorations.pathmorphing,decorations.pathreplacing}
\tikzset{decoration={snake,amplitude=.8mm,segment length=3mm,
		post length=1mm,pre length=0mm}}
\usepackage{thm-restate}

\usepackage{xspace}
\usepackage{amsfonts,amssymb}
\usepackage{mathtools}

\let\oldExists\exists%
\renewcommand{\exists}{\oldExists\,}
\let\oldForall\forall%
\renewcommand{\forall}{\oldForall\,}
\newcommand{\qedEx}{\hfill$\lhd$} 

\newcommand{\IN}{\mathbb{N}}

\newcommand{\IR}{\mathbb{R}}
\newcommand*{\intervaloo}[1]{\mathopen(#1\mathclose)} 
\newcommand*{\card}[1]{\ensuremath{\lvert#1\rvert}}

\newcommand{\Dist}{\ensuremath{\mathsf{Dist}}}
\newcommand{\probSpace}{\ensuremath{\Omega}}
\newcommand*{\Cyl}[1]{\ensuremath{\mathsf{Cyl}(#1)}}
\newcommand{\oalg}{\ensuremath{\mathcal{F}}}
\newcommand{\dist}{\ensuremath{\mu}}
\newcommand{\distSet}{\ensuremath{\Lambda}}

\newcommand*{\player}[1]{\ensuremath{\mathcal{P}_{#1}}}
\newcommand{\Pone}{\ensuremath{\player{1}}}
\newcommand{\Ptwo}{\ensuremath{\player{2}}}

\newcommand{\clr}{\ensuremath{c}}
\newcommand{\colors}{\ensuremath{C}}
\newcommand{\st}{\ensuremath{s}}
\newcommand{\states}{\ensuremath{S}}
\newcommand{\initStates}{\ensuremath{\states_{\mathsf{init}}}}
\newcommand*{\histOut}[1]{\ensuremath{\mathsf{out}(#1)}}
\newcommand{\act}{\ensuremath{A}}
\newcommand{\action}{\ensuremath{a}}
\newcommand*{\actions}[1]{\ensuremath{\act(#1)}}
\newcommand{\transProb}{\ensuremath{\delta}}
\newcommand{\colSolo}{\ensuremath{\mathsf{col}}}
\newcommand{\col}[1]{\ensuremath{\colSolo(#1)}}
\newcommand{\arena}{\ensuremath{\mathcal{A}}}
\newcommand{\colHatSolo}{\ensuremath{\mathsf{\widehat{col}}}}
\newcommand*{\colHat}[1]{\ensuremath{\colHatSolo(#1)}}
\newcommand*{\supp}[1]{\ensuremath{\mathsf{Supp}(#1)}}
\newcommand{\arenaFull}{\ensuremath{(\states_1, \states_2, \act, \transProb, \colSolo)}}
\newcommand{\arenaFullIndex}[1]{\ensuremath{(\states_1^{#1}, \states_2^{#1}, \act^{#1}, \transProb^{#1}, \colSolo^{#1})}}
\newcommand{\initArena}{\ensuremath{(\arena, \initStates)}}

\newcommand{\hist}{\ensuremath{\rho}}
\newcommand{\usualHist}{\ensuremath{\st_0\action_1\st_1\ldots\action_n\st_n}}
\newcommand{\play}{\ensuremath{\pi}}

\newcommand{\Plays}{\ensuremath{\mathsf{Plays}}}
\newcommand{\Hists}{\ensuremath{\mathsf{Hists}}}

\newcommand{\word}{\ensuremath{w}}
\newcommand{\memState}{\ensuremath{m}}
\newcommand{\memStates}{\ensuremath{M}}
\newcommand{\memUpd}{\alpha_{\mathsf{upd}}}
\newcommand{\memNxt}{\alpha_{\mathsf{nxt}}}
\newcommand{\memUpdHat}{\ensuremath{\widehat{\memUpd}}}
\newcommand{\memInit}{\ensuremath{\memState_{\mathsf{init}}}}
\newcommand{\mealy}{\ensuremath{\Gamma}}
\newcommand{\memSkel}{\ensuremath{\mathcal{M}}}
\newcommand{\memSkelFull}{\ensuremath{(\memStates, \memInit, \memUpd)}}
\newcommand*{\memSkelFullIndex}[1]{\ensuremath{(\memStates^{#1}, \memInit^{#1}, \memUpd^{#1})}}
\newcommand*{\memLang}[2]{\ensuremath{L_{#1,#2}}}
\newcommand*{\prodAS}[2]{\ensuremath{#1 \ltimes #2}}
\newcommand{\memProduct}{\otimes}
\newcommand{\memSkelTriv}{\ensuremath{\memSkel_{\mathsf{triv}}}}
\newcommand{\coverFunction}{\ensuremath{\phi}}

\newcommand{\pref}{\ensuremath{\sqsubseteq}} 
\newcommand{\strictPref}{\ensuremath{\sqsubset}}
\newcommand*{\inverse}[1]{{#1}^{-1}}

\newcommand{\invPref}{\ensuremath{\inverse{\pref}}}

\newcommand{\wc}{\ensuremath{W}} 
\newcommand{\pf}{\ensuremath{f}} 

\newcommand{\game}{\ensuremath{\mathcal{G}}}
\newcommand{\initGameFull}{\ensuremath{(\arena, \initStates, \pref)}}

\newcommand{\strat}{\ensuremath{\sigma}}
\newcommand{\stratBis}{\ensuremath{\tau}}
\newcommand*{\stratsType}[3]{\ensuremath{\Sigma_{#1}^{#2}(#3)}}
\newcommand{\FM}{\ensuremath{\mathsf{FM}}}
\newcommand{\pure}{\ensuremath{\mathsf{P}}}
\newcommand{\rand}{\ensuremath{\mathsf{G}}}
\newcommand{\PFM}{\ensuremath{\pure\FM}}
\newcommand{\RFM}{\ensuremath{\rand\FM}}
\newcommand{\allStrats}{\ensuremath{\{\PFM, \pure, \RFM, \rand\}}}
\newcommand{\UCol}{\mathsf{UCol}}
\newcommand{\DCol}{\mathsf{DCol}}
\newcommand*{\prSolo}[3]{\ensuremath{\mathsf{P}_{#1, #2}^{#3}}}
\newcommand*{\pr}[4]{\ensuremath{\prSolo{#1}{#2}{#3}[#4]}}
\newcommand{\prcOperator}{\ensuremath{\mathsf{Pc}}}
\newcommand*{\prcSolo}[3]{\ensuremath{\prcOperator_{#1, #2}^{#3}}}
\newcommand*{\prcSoloLocalFix}[2]{\ensuremath{\prcOperator_{#1}^{#2}}}
\newcommand*{\prc}[4]{\ensuremath{\prcSolo{#1}{#2}{#3}[#4]}}
\newcommand{\expectation}{\ensuremath{\mathsf{E}}}

\newcommand{\stratType}{\ensuremath{\mathsf{X}}}
\newcommand*{\cl}[3]{\ensuremath{[#1]_{#2}^{#3}}}
\newcommand{\stoch}{\ensuremath{\mathsf{S}}}
\newcommand{\deter}{\ensuremath{\mathsf{D}}}
\newcommand*{\PoneArenas}[1]{\ensuremath{\arenaFam^{#1}_{\Pone}}}
\newcommand{\PoneDetArenas}{\ensuremath{\PoneArenas{\deter}}}
\newcommand{\PoneStochArenas}{\ensuremath{\PoneArenas{\stoch}}}
\newcommand{\arenaType}{\ensuremath{\mathsf{Y}}}
\newcommand{\PoneA}{\ensuremath{\PoneArenas{\arenaType}}}
\newcommand{\arenaJoin}{\ensuremath{\sqcup}}
\newcommand{\bigArenaJoin}{\ensuremath{\bigsqcup}}
\newcommand{\joinState}{\ensuremath{t}}
\newcommand*{\arenaPref}[3]{\ensuremath{#1_{{#2}\leadsto{#3}}}}
\newcommand{\blank}{\ensuremath{{\,\cdot}}}

\newcommand{\MP}{\ensuremath{\mathsf{MP}}}
\newcommand{\simPar}{\ensuremath{\wc_\mathsf{wp}}}
\newcommand{\Mmax}{\ensuremath{\memSkel_\mathsf{max}}} 

\newcommand{\pathFun}{\ensuremath{\mathcal{H}}}
\newcommand{\stratFun}{\ensuremath{f}}
\newcommand*{\bijPaths}[1]{\ensuremath{\phi^{#1}}}
\newcommand*{\bijPlays}[1]{\ensuremath{\bijPaths{#1}_{\infty}}}
\newcommand*{\bijSplit}[2]{\ensuremath{\Phi_{#1}^{#2}}}
\newcommand*{\size}[1]{\ensuremath{n_{#1}}}
\newcommand{\arenaFam}{\ensuremath{\mathfrak{A}}}
\newcommand*{\shiftStrat}[2]{\ensuremath{#1[#2]}}
\newcommand*{\shiftPref}[1]{\ensuremath{\pref_{[#1]}}}
\newcommand{\bijSt}{\ensuremath{\psi_\states}}
\newcommand*{\bijAct}[1]{\ensuremath{\psi_\act^{#1}}}
\newcommand{\arenaSplitFixedStrat}{\ensuremath{\splitArena{\arena}{t}_{\strat_2}}}
\newcommand{\disjUnion}{\ensuremath{\uplus}}
\newcommand\restr[2]{{
		\left.\kern-\nulldelimiterspace
		#1\right|_{#2}}}
\newcommand*\ud{\mathop{}\!\mathrm{d}}
\newcommand{\discFac}{\ensuremath{\lambda}}
\newcommand{\disc}{\ensuremath{\mathsf{Disc}_\discFac}}
\renewcommand{\complement}[1]{\ensuremath{{#1}^{\mathsf{c}}}}
\newcommand{\F}{\ensuremath{\lozenge}}
\newcommand{\GG}{\ensuremath{\square}}
\newcommand{\Amon}{\ensuremath{\arena_\mathsf{mon}}}
\newcommand{\Asel}{\ensuremath{\arena_\mathsf{sel}}}

\newcommand*{\splitArena}[2]{\ensuremath{\mathsf{Split}(#1, #2)}}
\newcommand*{\splitInitStates}[1]{\ensuremath{\initStates^{#1}}}

\makeatletter
\tikzset{circle split part fill/.style  args={#1,#2}{%
 alias=tmp@name,
  postaction={%
    insert path={
     \pgfextra{%
     \pgfpointdiff{\pgfpointanchor{\pgf@node@name}{center}}%
                  {\pgfpointanchor{\pgf@node@name}{east}}%
     \pgfmathsetmacro\insiderad{\pgf@x}
      \fill[#1] (\pgf@node@name.base) ([xshift=-\pgflinewidth]\pgf@node@name.east) arc
                          (0:180:\insiderad-\pgflinewidth)--cycle;
      \fill[#2] (\pgf@node@name.base) ([xshift=\pgflinewidth]\pgf@node@name.west)  arc
                           (180:360:\insiderad-\pgflinewidth)--cycle;            
         }}}}}
 \makeatother

\tikzstyle{rond}=[draw,circle,minimum height=7mm]
\tikzstyle{oval}=[draw,ellipse,minimum height=7mm]
\tikzstyle{carre}=[draw,minimum width=6mm,minimum height=6mm]

\makeatletter
\def\listof#1{\expandafter\@listof#1+\@end}
\def\@listof#1+#2\@end{\def\@tempa{#1}\ifx\@tempa\@empty\else
    \langle#1\rangle\fi
  \def\@tempa{#2}\ifx\@tempa\@empty\else,\@listof#2\@end\fi}
\def\stackof#1{\begin{array}{@{}>{\scriptstyle}c@{}}\expandafter\@stackof#1+\@end}
\def\@stackof#1+#2\@end{\def\@tempa{#1}\ifx\@tempa\@empty\else
    \langle #1\rangle \fi
  \def\@tempa{#2}\ifx\@tempa\@empty\end{array}\else\\[-1mm]\@stackof#2\@end\fi}
\makeatother

\begin{document}

\title[Arena-Independent Finite-Memory Determinacy in Stochastic Games]{Arena-Independent Finite-Memory Determinacy in Stochastic Games\rsuper*}
\titlecomment{{\lsuper*}Research supported by the Fonds de la Recherche Scientifique -- FNRS under Grants n$^\circ$ F.4520.18 (ManySynth) and n$^\circ$ T.0188.23 (PDR ControlleRS), by ENS Paris-Saclay visiting professorship (M.~Randour, 2019), and by the ANR Project MAVeriQ (ANR-20-CE25-0012).
Mickael Randour is an F.R.S.-FNRS Research Associate and Pierre Vandenhove is an F.R.S.-FNRS Research Fellow.}

\author[P.~Bouyer]{Patricia~Bouyer\lmcsorcid{0000-0002-2823-0911}}[a]
\address{Université Paris-Saclay, CNRS, ENS Paris-Saclay, Laboratoire Méthodes Formelles, 91190, Gif-sur-Yvette, France}

\author[Y.~Oualhadj]{Youssouf~Oualhadj\lmcsorcid{0000-0003-0200-4032}}[b]
\address{Univ Paris Est Creteil, LACL, F-94010 Creteil, France}

\author[M.~Randour]{Mickael~Randour\lmcsorcid{0000-0001-8777-2385}}[c]
\address{F.R.S.-FNRS \& UMONS -- Universit\'e de Mons, Belgium}

\author[P.~Vandenhove]{Pierre~Vandenhove\lmcsorcid{0000-0001-5834-1068}}[a,c]

\begin{abstract}
	\noindent We study stochastic zero-sum games on graphs, which are prevalent tools to model decision-making in presence of an antagonistic opponent in a random environment.
	In this setting, an important question is the one of \emph{strategy complexity}: what kinds of strategies are sufficient or required to play optimally (e.g., randomization or memory requirements)?
	Our contributions further the understanding of \emph{arena-independent finite-memory (AIFM) determinacy}, i.e., the study of objectives for which memory is needed, but in a way that only depends on limited parameters of the game graphs.
	First, we show that objectives for which pure AIFM strategies suffice to play optimally also admit pure AIFM \emph{subgame perfect} strategies.
	Second, we show that we can reduce the study of objectives for which pure AIFM strategies suffice in two-player stochastic games to the easier study of one-player stochastic games (i.e., Markov decision processes).
	Third, we \emph{characterize} the sufficiency of AIFM strategies through two intuitive properties of objectives.
	This work extends a line of research started on deterministic games to stochastic ones.
\end{abstract}

\maketitle

\section{Introduction}\label{sec:intro}
\emph{Controller synthesis} consists, given a system, an environment, and a specification, in automatically generating a controller of the system that guarantees the specification in the environment.
This task is often studied through a game-theoretic lens: the system is a game, the controller is a player, the uncontrollable environment is its adversary, and the specification is a game objective~\cite{rECCS}.
A \emph{game on graph} consists of a directed graph, called an \emph{arena}, partitioned into two kinds of vertices: some of them are controlled by the system (called \emph{player~$1$}) and the others by the environment (called \emph{player~$2$}).
Player~$1$ is given a \emph{game objective} (corresponding to the specification) and must devise a \emph{strategy} (corresponding to the controller) to accomplish the objective or optimize an outcome.
The strategy can be seen as a function that dictates the decisions to make in order to react to every possible chain of events.
In case of uncertainty in the system or the environment, probability distributions are often used to model transitions in the game graph, giving rise to the \emph{stochastic game} model.
We study here \emph{stochastic turn-based zero-sum games on graphs}~\cite{Con92}, also called perfect-information stochastic games.
We also discuss the case of \emph{deterministic games}, which can be seen as a subcase of stochastic games in which only Dirac distributions are used in transitions.

\paragraph{Strategy complexity.}
A common question underlying all game objectives is the one of \emph{strategy complexity}: how \emph{complex} must optimal strategies be, and how \emph{simple} can optimal strategies be?
For each distinct game objective, multiple directions can be investigated, such as the need for randomization~\cite{CDGH10} (must optimal strategies make stochastic choices?), the need for memory~\cite{GZ05,GZ09,BLORV22} (how much information about the past must optimal strategies remember?), or what trade-offs exist between randomization and memory~\cite{CAH04,Hor09,CRR14,MPR20}.
With respect to memory requirements, three cases are typically distinguished: \emph{memoryless-determined objectives}, for which \emph{memoryless} strategies suffice to play optimally; \emph{finite-memory-determined objectives}, for which finite-memory strategies suffice (memory is then usually encoded as a deterministic finite automaton); and objectives for which infinite memory is required.
High memory requirements (such as exponential memory and obviously infinite memory) are a major drawback when it comes to implementing controllers; hence specific approaches are often developed to look for \textit{simple} strategies (e.g.,~\cite{DBLP:conf/tacas/DelgrangeKQR20}).

Many classical game objectives (reachability~\cite{Con92}, B\"uchi and parity~\cite{CJH04}, energy~\cite{BBE10}, discounted sum~\cite{Sha53}\ldots) are memoryless-determined, both in deterministic and stochastic arenas.
Nowadays, multiple general results allow for a more manageable proof for most of these objectives: we mention~\cite{Kop06,BFMM11,AR17} for sufficient conditions in deterministic games, and~\cite{Gim07,GK14} for similar conditions in one-player and two-player stochastic games.
One milestone for memoryless determinacy in deterministic games was achieved by Gimbert and Zielonka~\cite{GZ05}, who provide \emph{two characterizations} of it: the first one states two necessary and sufficient conditions (called \emph{monotony} and \emph{selectivity}) for memoryless determinacy, and the second one states that memoryless determinacy in both players' \emph{one-player games} suffices for memoryless determinacy in two-player games (we call this result the \emph{one-to-two-player lift}).
Together, these characterizations provide a theoretical and practical advance.
On the one hand, monotony and selectivity improve the high-level understanding of what conditions well-behaved objectives verify.
On the other hand, only having to consider the one-player case thanks to the one-to-two-player lift is of tremendous help in practice.
A generalization of the one-to-two-player lift to \emph{stochastic games} was shown also by Gimbert and Zielonka in an unpublished paper~\cite{GZ09} and is about memoryless strategies that are \emph{pure} (i.e., not using randomization).

\paragraph{The need for memory.}
Recent research tends to study increasingly complex settings --- such as combinations of qualitative/quantitative objectives or of behavioral models --- for which finite or infinite memory is often required; see examples in deterministic games~\cite{CD12b,DBLP:journals/iandc/VelnerC0HRR15,DBLP:journals/iandc/BruyereFRR17,DBLP:journals/acta/BouyerMRLL18,BHRR19}, Markov decision processes --- i.e., one-player stochastic games~\cite{DBLP:conf/vmcai/RandourRS15,RRS17,DBLP:journals/lmcs/ChatterjeeKK17,DBLP:conf/icalp/BerthonRR17,BDOR20}, or stochastic games~\cite{CFKSW13,CD16,DBLP:conf/concur/ChatterjeeP19,CKWW20,MSTW21}.
Motivated by the growing memory requirements of these endeavors, research about strategy complexity often turns toward \emph{finite-memory determinacy}.
Proving finite-memory determinacy is sometimes difficult (already in deterministic games, e.g.,~\cite{BHMRZ17}), and as opposed to memoryless strategies, there are few widely applicable results.
We mention~\cite{LPR18}, which provides sufficient conditions for finite-memory determinacy in Boolean combinations of finite-memory-determined objectives in deterministic games.
Results for multi-player non-zero-sum games are also available~\cite{LP18}.

\paragraph{Arena-independent finite-memory.}
An interesting middle ground between the well-understood memoryless determinacy and the more puzzling finite-memory determinacy was proposed for deterministic games in~\cite{BLORV22}: an objective is said to admit \emph{arena-independent finite-memory (AIFM) strategies} if a \emph{single finite memory structure} suffices to play optimally in any arena.
In practice, this memory structure may depend on parameters of the objective (for instance, largest weight, number of priorities), but not on parameters intrinsically linked to the arena (e.g., number of states or transitions).
AIFM strategies include as a special case memoryless strategies, since they can be implemented with a trivial memory structure with a single state.

AIFM strategies have a remarkable feature: in deterministic arenas, AIFM generalizations of \emph{both} characterizations from~\cite{GZ05} hold, including the one-to-two-player lift~\cite{BLORV22}.
From a practical point of view, it brings techniques usually linked to memoryless determinacy to many finite-memory-determined objectives.
The aim of this article is to show that this also holds true in stochastic arenas.

AIFM strategies bring therefore an interesting trade-off: they admit good structural properties that facilitate their study (whereas few such results are known about the more general finite-memory strategies), while still being applicable to many objectives (see paragraph \textbf{Applicability} below; for instance, they suffice for all \emph{$\omega$-regular objectives}, while memoryless strategies only suffice for a subclass of these).

\paragraph{Contributions.}
We provide an overview of desirable properties of objectives in which pure AIFM strategies suffice to play optimally in \emph{stochastic} games, and tools to study them.
This entails:
\begin{itemize}
	\item a proof of a specific feature of objectives for which pure AIFM strategies suffice to play optimally: for such objectives, there also exist pure AIFM \emph{subgame perfect (SP)} strategies (Theorem~\ref{prop:optAIFMimpliesSP}), which is a stronger requirement than optimality;
	\item a more general one-to-two-player lift: we show the equivalence between the existence of pure AIFM optimal strategies in two-player games for both players and the existence of pure AIFM optimal strategies in one-player games, thereby simplifying the proof of memory requirements for many objectives (Theorem~\ref{thm:1to2});
	\item two conditions generalizing monotony and selectivity in the stochastic/AIFM case; these conditions are equivalent to the existence of pure AIFM optimal strategies in \emph{one-player stochastic arenas} (Theorem~\ref{thm:monSel}) for objectives that can be encoded as \emph{real payoff functions}.
\end{itemize}
In practice, this last theorem can be used to prove memory requirements in one-player arenas, and then the second theorem can be used to lift these to the two-player case.

These results reinforce both sides on the frontier between AIFM strategies and general finite-memory strategies:
on the one hand, objectives for which pure AIFM strategies suffice indeed share interesting properties with objectives for which pure memoryless strategies suffice, rendering their analysis easier, even in the stochastic case;
on the other hand, our novel result about SP strategies does not hold for (arena-\emph{dependent}) finite-memory strategies, and therefore further distinguishes the AIFM case from the finite-memory case.

The one-to-two-player lift for pure AIFM strategies in stochastic games is not surprising, as it holds for pure memoryless strategies in stochastic games~\cite{GZ09}, and for AIFM strategies in deterministic games~\cite{BLORV22}.
Moreover, although the monotony/selectivity characterization is definitely inspired from the deterministic case~\cite{GZ05,BLORV22}, it had not been formulated for stochastic games, even in the pure memoryless case --- its proof involves new technical difficulties to which our improved understanding of subgame perfect strategies brings insight.

All our results are about the optimality of \emph{pure} AIFM strategies in various settings: they can be applied in an independent way for deterministic games and for stochastic games, and they can also consider optimality under restriction to different classes of strategies (allowing or not the use of randomization and infinite memory).

The proof technique for the one-to-two-player lift shares a similar outline in~\cite{GZ05,GZ09,BLORV22} and in this paper: it relies on an induction on the number of edges in arenas to show the existence of memoryless optimal strategies.
This edge-induction technique is frequently used in comparable ways in other works about memoryless determinacy~\cite{Kop06,Gim07,GK14,CD16}.
In the AIFM case, the extra challenge consists of applying such an induction to the right set of arenas in order for a result about memoryless strategies to imply something about AIFM strategies.
Work in~\cite{BLORV22} paved the way to neatly overcome this technical hindrance and we were able to factorize the main argument in Lemma~\ref{prop:optProdIsCov}.

Although obtaining only results about \emph{pure} strategies can be seen as a limitation, we show in Section~\ref{sec:noLiftRandomized} an example illustrating that the one-to-two-player lift does not hold if we allow for unconstrained randomization in the strategies.

\paragraph{Applicability.}
Let us discuss objectives that admit, or not, pure AIFM optimal strategies in stochastic arenas.
\begin{itemize}
	\item Objectives for which AIFM optimal strategies exist include the aforementioned memoryless-determined objectives~\cite{Con92,CJH04,Sha53,BBE10}, as explained earlier.
	Such objectives could already be studied through the lens of a one-to-two-player lift~\cite{GZ09}, but our two other main results also apply to these.
	\item Pure AIFM optimal strategies exist in lexicographic reachability-safety games~\cite[Theorem~4]{CKWW20}: the memory depends only on the number of targets to visit or avoid, but not on parameters of the arena (number of states or transitions).
	\item Muller objectives whose probability must be maximized~\cite{Cha12} also admit pure AIFM optimal strategies: the number of memory states depends only on the colors and on the Muller condition.
	\item In general, every $\omega$-regular objective admits pure AIFM optimal strategies, as it can be seen as a parity objective (for which pure memoryless strategies suffice) after taking the product of the game graph with a deterministic parity automaton accepting the objective~\cite{Mos84,CH12}.
	This parity automaton can be taken as an arena-independent memory structure.
	It is therefore possible to use our results to investigate precise memory bounds in stochastic games for multiple $\omega$-regular objectives which have been studied in deterministic games or in one-player stochastic games: generalized parity games~\cite{CHP07}, lower- and upper-bounded energy games~\cite{BFLMS08}, some window objectives~\cite{BHR16,BDOR20}, weak parity games~\cite{Tho08} (this last example is detailed in Section~\ref{sec:app:wp}).
	\item There are objectives for which finite-memory strategies suffice for some player, but with an underlying memory structure depending on parameters of the arena (an example is provided by the \emph{Gain} objective in~\cite[Theorem~6]{MSTW21}).
	Many objectives also require infinite memory, such as generalized mean-payoff games~\cite{CD16} (both in deterministic and stochastic games) and energy-parity games (only in stochastic games~\cite{CD12b,MSTW17}).
	Our characterizations provide a more complete understanding of why AIFM strategies do not suffice.
\end{itemize}

\paragraph{Deterministic and stochastic games.}
There are natural ways to extend classical objectives for deterministic games to a stochastic context: typically, for qualitative objectives, a natural stochastic extension is to maximize the probability to win.
Still, in general, memory requirements may increase when switching to the stochastic context.
To show that understanding the deterministic case is insufficient to understand the stochastic case, we outline three situations displaying different behaviors.
\begin{itemize}
	\item As mentioned above, for many classical objectives, memoryless strategies suffice both in deterministic and in stochastic games.
	\item AIFM strategies may suffice both for deterministic and stochastic games, but with a difference in the size of the required memory structure.
	One such example is provided by the weak parity objective~\cite{Tho08}, for which memoryless strategies suffice in deterministic games, but which requires memory in stochastic games (this was already noticed in~\cite[Section~4.4]{GZ09}).
	Yet, it is possible to show that pure AIFM strategies suffice in stochastic games using the results from our paper.
	This shows that to go from the deterministic to the stochastic case, a ``constant'' increase in memory may be necessary and sufficient.
	\item There are also objectives for which memoryless strategies suffice in deterministic games, but even AIFM strategies do not suffice in stochastic games.
	One such example consists in maximizing the probability to obtain a non-negative discounted sum (which is different from maximizing the expected value of the discounted sum, for which memoryless strategies suffice, as is shown in~\cite{Sha53}).
\end{itemize}
Formal proofs for these last two examples are provided in Section~\ref{sec:app}.
These three situations further highlight the significance of establishing results about memory requirements in stochastic games, even for objectives whose deterministic version is well-understood.

\paragraph{Outline.}
We introduce our framework and notations in Section~\ref{sec:prelim}.
We discuss AIFM strategies and tools to relate them to memoryless strategies in Section~\ref{sec:memory}, which allows us to prove our result about subgame perfect strategies.
The one-to-two-player lift is presented in Section~\ref{sec:1to2}, followed by the one-player characterization in Section~\ref{sec:1p}.
We provide illustrative applications of our results in Section~\ref{sec:app}.

This paper is a full version of a preceding conference version~\cite{BORV21Conf}.
This version supplements the conference version with extra examples and remarks, and contains all the detailed proofs of the statements.

\section{Preliminaries}\label{sec:prelim}

Let $\colors$ be a non-empty set of \emph{colors}.
There are no further constraints on set $\colors$; in particular, $\colors$ is allowed to be infinite ($\colors$ can for instance be $\IN$ or $\IR$).
For $B$ a set, we write $B^*$ for the set of finite sequences of elements of $B$ and $B^\omega$ for the set of infinite sequences of elements of $B$.

\paragraph{Probabilities.}
For a measurable space $(\probSpace, \oalg)$ (resp.\ a finite set $\probSpace$), we write $\Dist(\probSpace, \oalg)$ (resp.\ $\Dist(\probSpace)$) for the set of \emph{probability distributions on $(\probSpace, \oalg)$} (resp.\ \emph{on $\probSpace$}).
For $\probSpace$ a finite set and $\dist\in\Dist(\probSpace)$, we write $\supp{\dist} = \{\omega\in\probSpace \mid \dist(\omega) > 0\}$ for the \emph{support of $\dist$}.

\paragraph{Arenas.}
We consider stochastic games played by two players, called $\Pone$ (for player~$1$) and $\Ptwo$ (for player~$2$), who play in a turn-based fashion on \emph{arenas}.

\begin{defi}[Arena]
	A (two-player stochastic turn-based) \emph{arena} is a tuple $\arena = \arenaFull$, where:
	\begin{itemize}
	\item $\states_1$ and $\states_2$ are two disjoint finite sets of \emph{states}, respectively controlled by $\Pone$ and $\Ptwo$ --- we denote $\states = \states_1\uplus\states_2$ for the union of all states;
	\item $\act$ is a finite set of \emph{actions};
	\item $\transProb \colon \states \times \act \to \Dist(\states)$ is a partial function called \emph{probabilistic transition function};
	\item $\colSolo \colon \states \times \act \to \colors$ is a partial function called \emph{coloring function}.
	\end{itemize}
	For a state $\st\in\states$, we write $\actions{\st}$ for the set of actions that are \emph{available in $\st$}, that is, the set of actions for which $\transProb(\st, \action)$ is defined.
	For $\st\in\states$, function $\colSolo$ must be defined for all pairs $(\st, \action)$ such that $\action$ is available in $\st$.
	We require that for all $\st\in\states$, $\actions{\st}\neq\emptyset$.
\end{defi}
The last condition ensures that there is at least one available action in every state (i.e., arenas are \emph{non-blocking}).
For $\st, \st' \in\states$ and $\action\in\actions{\st}$, we usually denote $\transProb(\st, \action, \st')$ instead of $\transProb(\st, \action)(\st')$ for the probability to reach $\st'$ in one step by playing $\action$ in $\st$, and we write $(\st, \action, \st') \in \transProb$ if and only if $\transProb(\st, \action, \st') > 0$.
An interesting subclass of (stochastic) arenas is the class of \emph{deterministic} arenas: an arena $\arena = \arenaFull$ is \emph{deterministic} if for all $\st\in\states$, $\action\in\actions{\st}$, $\card{\supp{\transProb(\st, \action)}} = 1$.

Let $\arena = \arenaFull$ be an arena.
A \emph{play of $\arena$} is an infinite sequence of states and actions $\st_0\action_1\st_1\action_2\st_2\ldots \in (\states\act)^\omega$ such that for all $i \ge 0$, $(\st_i, \action_{i+1}, \st_{i+1}) \in \transProb$.
The set of all plays starting in a state $\st\in\states$ is denoted $\Plays(\arena, \st)$.
A prefix of a play is an element in $\states(\act\states)^*$ and is called a \emph{history}; the set of all histories starting in a state $\st\in\states$ is denoted $\Hists(\arena, \st)$.
For $\states'\subseteq\states$, we write $\Plays(\arena, \states')$ (resp.\ $\Hists(\arena, \states')$) for the unions of $\Plays(\arena, \st)$ (resp.\ $\Hists(\arena, \st)$) over all states $\st\in\states'$.
For $\hist = \usualHist$ a history, we write $\histOut{\hist}$ for $\st_n$.
For $i\in\{1, 2\}$, we write $\Hists_i(\arena, \st)$ and $\Hists_i(\arena, \states')$ for the corresponding histories $\hist$ such that $\histOut{\hist} \in \states_i$.
For $\st\in\states$ (resp.\ $\states'\subseteq \states$) and $\st'\in\states$, we write $\Hists(\arena, \st, \st')$ (resp.\ $\Hists(\arena, \states', \st')$) for the histories in $\Hists(\arena, \st)$ (resp.\ $\Hists(\arena, \states')$) such that $\histOut{\hist} = \st'$.

We extend $\colSolo$ to histories and plays with $\colHatSolo$: for a history $\hist = \usualHist$, we write $\colHat{\rho}$ for the finite sequence $\col{\st_0, \action_1}\ldots\col{\st_{n-1}, \action_n}\in\colors^*$; for $\play = \st_0\action_1\st_1\action_2\st_2\ldots$ a play, we write $\colHat{\play}$ for the infinite sequence $\col{\st_0, \action_1}\col{\st_{1}, \action_2}\ldots\in\colors^\omega$.

A \emph{one-player arena of $\player{i}$} is an arena $\arena = \arenaFull$ such that for all $\st\in \states_{3-i}$, $\card{\actions{\st}} = 1$.
A one-player arena in our context corresponds to the notion of \emph{Markov decision process (MDP)} often found in the literature~\cite{Put94,BK08}.

For technical reasons that will be further justified later, we will usually work on arenas where the set of initial states is explicitly specified.
\begin{defi}[Initialized arena]
	An \emph{initialized arena} is a pair $(\arena, \initStates)$ such that $\arena$ is an arena and $\initStates$ is a non-empty subset of the states of $\arena$, called the set of \emph{initial states}.
	We assume w.l.o.g.\ that all states of $\arena$ are reachable from $\initStates$ following transitions with positive probabilities in the probabilistic transition function of $\arena$.
\end{defi}

If an initialized arena has only one initial state $\st\in\states$, we write $(\arena, \st)$ for $(\arena, \{\st\})$.

We will often compare initialized arenas even if they are not formally defined on the same state space by using a natural definition of \emph{isomorphism}: we say that two initialized arenas $(\arenaFull, \initStates)$ and $((\states_1', \states_2', \act', \transProb', \colSolo'), \initStates')$ are \emph{isomorphic} if there exist a bijection $\bijSt\colon \states \to \states'$ and for all $\st\in\states$, a bijection $\bijAct{\st}\colon\actions{\st}\to\act'(\bijSt(\st))$ such that $\bijSt(\states_1) = \states_1'$, $\bijSt(\states_2) = \states_2'$, $\bijSt(\initStates) = \initStates'$, and for all $\st_1,\st_2\in\states$, $\action\in\act$, we have $\transProb(\st_1, \action)(\st_2) = \transProb'(\bijSt(\st_1), \bijAct{\st_1}(\action))(\bijSt(\st_2))$ and $\col{\st_1, \action} = \colSolo'(\bijSt(\st_1), \bijAct{\st_1}(\action))$.

We will consider sets (which we call \emph{classes}) of initialized arenas, which are usually denoted by the letter $\arenaFam$.
Although our results often apply to more fine-grained classes of arenas, typical classes that we will consider consist of all one-player or two-player, deterministic or stochastic initialized arenas.
We use \emph{initialized} arenas throughout the paper for technical reasons, but all of our results can be converted to results using only the more classical notion of arena.

\paragraph{Memory.}
To play in games, players use strategies, which can sometimes be efficiently implemented with \emph{finite memory}.
We define a classical notion of memory based on complete deterministic automata on \emph{colors}.
The goal of using colors instead of states/actions for transitions of the memory is to allow to define memory structures independently of arenas, so that they can be used in all arenas.

\begin{defi}[Memory skeleton]
	A \emph{memory skeleton} is a tuple $\memSkel = \memSkelFull$ where $\memStates$ is a set of \emph{memory states}, $\memInit \in \memStates$ is an \emph{initial state} and $\memUpd\colon \memStates\times\colors\to\memStates$ is an \emph{update function}.
	We add the following constraint: for all finite sets of colors $B\subseteq \colors$, the number of states reachable from $\memInit$ with transitions provided by $\restr{\memUpd}{\memStates\times B}$ is finite (where $\restr{\memUpd}{\memStates\times B}$ is the restriction of the domain of $\memUpd$ to $\memStates\times B$).
\end{defi}

We slightly relax the usual finiteness constraint for the state space by simply requiring that whenever restricted to finitely many colors, the state space of the skeleton is finite.
Memory skeletons with a finite state space are all encompassed by this definition, but this also allows some memory skeletons with infinitely many states.
For example, if $\colors = \IN$, the tuple $(\IN, 0, (\memState, n) \mapsto \max\{\memState, n\})$, which remembers the greatest color seen, is a valid memory skeleton: for any finite $B \subseteq \colors$, we only need to use memory states up to $\max B$.
However, the tuple $(\IN, 0, (\memState, n) \mapsto m + n)$ remembering the current sum of all colors seen is \emph{not} a memory skeleton, as infinitely many states are reachable from $0$, even if only $B = \{1\}$ can be used.

We denote $\memUpdHat\colon \memStates\times \colors^*\to \memStates$ for the natural extension of $\memUpd$ to finite sequences of colors.

It will often be useful to use two memory skeletons in parallel, which is equivalent to using their \emph{product}.
\begin{defi}[Product of skeletons]
	Let $\memSkel^1 = \memSkelFullIndex{1}$, $\memSkel^2 = \memSkelFullIndex{2}$ be two memory skeletons.
	We define their \emph{product $\memSkel^1 \memProduct \memSkel^2$} as the memory skeleton $(\memStates, \memInit, \memUpd)$ obtained as follows: $\memStates = \memStates^1 \times \memStates^2$, $\memInit = (\memInit^1, \memInit^2)$, and, for all $\memState^1 \in \memStates^1$, $\memState^2 \in \memStates^2$, $\clr \in \colors$, $\memUpd((\memState^1, \memState^2), \clr) = (\memUpd^1(\memState^1, \clr), \memUpd^2(\memState^2, \clr))$.
\end{defi}
The update function of the product simply updates both skeletons in parallel.

\begin{defi}[Product initialized arenas]
	Let $(\arena = \arenaFull, \initStates)$ be an initialized arena and $\memSkel = \memSkelFull$ be a memory skeleton.
	We define the \emph{product initialized arena $\prodAS{(\arena, \initStates)}{\memSkel}$} as the initialized arena $((\states_1', \states_2', \act, \transProb', \colSolo'), \initStates')$ where:
	\begin{itemize}
		\item $\initStates' = \initStates \times \{\memInit\}$,
		\item $\transProb' \colon (\states\times\memStates)\times \act \to \Dist(\states\times\memStates)$ is such that for all $(\st, \memState)\in\states\times\memStates$ and $\action\in\act$, $\transProb'((\st, \memState), \action)$ is defined if and only if $\transProb(\st, \action)$ is defined, in which case $\transProb'((\st, \memState), \action, (\st', \memState'))$ is equal to $\transProb(\st, \action, \st')$ if $\memUpd(\memState, \col{\st, \action}) = \memState'$, and is $0$ otherwise --- this implies that $\actions{(\st, \memState)} = \actions{\st}$,
		\item $\states'$ is the smallest subset of $\states\times\memStates$ such that $\initStates'\subseteq \states'$, and for all $(\st, \memState), (\st', \memState')\in\states\times\memStates$, $\action\in\act$, if $(\st, \memState)\in\states'$ and $\transProb((\st, \memState), \action, (\st', \memState'))$ is positive, then $(\st', \memState')\in\states'$; we define $\states_1' = \states' \cap (\states_1\times\memStates)$ and $\states_2' = \states' \cap (\states_2\times\memStates)$,
		\item for all $(\st, \memState)\in\states'$ and $\action\in\actions{\st}$, $\colSolo'((\st, \memState), \action) = \col{\st, \action}$.
	\end{itemize}
\end{defi}

\noindent
A product initialized arena $\prodAS{\initArena}{\memSkel}$ is an initialized arena with transitions obtained from $\arena$, with state space enriched with extra information about the current memory state, which is initialized at $\memInit$.
We only keep states that are reachable from $\initStates'$ following transitions of $\transProb'$, thereby enforcing that in initialized arenas, all states in the state space are reachable from the initial states.
Even if memory skeletons have infinitely many states or transitions, product initialized arenas are always finite, as only finitely many colors appear in an initialized arena, and only these colors appear in the product initialized arena.

\paragraph{Strategies.}
We can now define strategies, which are functions describing what each player does in response to every possible scenario.
\begin{defi}[Strategy]
	Given an initialized arena $(\arena, \initStates)$ and $i\in\{1, 2\}$, a \emph{strategy of $\player{i}$ on $(\arena, \initStates)$} is a function $\strat_i\colon \Hists_i(\arena, \initStates) \to \Dist(\act)$ such that for all $\hist\in\Hists_i(\arena, \initStates)$, $\supp{\strat_i(\hist)} \subseteq \actions{\histOut{\hist}}$.
\end{defi}

We now discuss interesting subclasses of strategies.

A strategy $\strat_i$ of $\player{i}$ on $(\arena, \initStates)$ is \emph{pure} if it does not resort to probability distributions to choose actions, that is, if for all $\hist\in\Hists_i(\arena, \initStates)$, $\card{\supp{\strat_i(\hist)}} = 1$.
If a strategy is not pure, then it is \emph{randomized}.

A strategy $\strat_i$ of $\player{i}$ on $(\arena, \initStates)$ is \emph{memoryless} if every distribution over actions it selects only depends on the current state of the arena, and not on the whole history, that is, if for all $\hist, \hist' \in \Hists_i(\arena, \initStates)$, $\histOut{\hist} = \histOut{\hist'}$ implies $\strat_i(\hist) = \strat_i(\hist')$.
A \emph{pure memoryless} strategy of $\player{i}$ can be simply specified as a function $\states_i \to \act$.

A strategy $\strat_i$ of $\player{i}$ on $(\arena, \initStates)$ is \emph{finite-memory} if it can be encoded as a \emph{Mealy machine} $\mealy = (\memSkel, \memNxt)$, with $\memSkel = \memSkelFull$ being a memory skeleton and $\memNxt\colon \states_i \times \memStates \to \Dist(\act)$ being the \emph{next-action function}, which is such that for $\st\in\states_i$, $\memState\in\memStates$, $\supp{\memNxt(\st, \memState)} \subseteq \actions{\st}$.
Strategy $\strat_i$ is encoded by $\mealy$ if for all histories $\hist\in\Hists_i(\arena, \initStates)$,
\[
	\strat_i(\hist) = \memNxt(\histOut{\hist}, \memUpdHat(\memInit, \colHat{\hist})).
\]
If $\strat_i$ can be encoded as a Mealy machine $(\memSkel, \memNxt)$, we say that $\strat_i$ is \emph{based on (memory) $\memSkel$}.
If $\strat_i$ is \emph{based on $\memSkel$} and is \emph{pure}, then the next-action function can be specified as a function $\states_i\times\memStates \to \act$.
Memoryless strategies correspond to finite-memory strategies based on the \emph{trivial memory skeleton} $\memSkelTriv = (\{\memInit\}, \memInit, (\memInit, \clr)\mapsto \memInit)$ that has a single state.

We denote by $\stratsType{i}{\PFM}{\arena, \initStates}$ (resp.\ $\stratsType{i}{\pure}{\arena, \initStates}$, $\stratsType{i}{\RFM}{\arena, \initStates}$, $\stratsType{i}{\rand}{\arena, \initStates}$) the set of pure finite-memory (resp.\ pure, finite-memory, general) strategies of $\player{i}$ on $(\arena, \initStates)$ (where the adjective \emph{general} encompasses all strategies, including those with randomization).
A \emph{type of strategies} is an element $\stratType\in\allStrats$ corresponding to these subsets.

\begin{rem}\label{rmk:stochSkel}
	Observe that our definition of finite-memory strategies with randomization only allows for randomization in the output of the function $\memNxt$.
	In general, to induce distributions as arbitrary as possible, it may be useful to allow for randomization in the skeletons (in the initial memory state and/or in the update function $\memUpd$)~\cite{MR22}.
	Here, one of our upcoming notions (\emph{coverability}) is not compatible with randomization in the memory skeleton (see Remark~\ref{rmk:stochCover}).
	In order not to make the notations heavier, we therefore prefer to only consider fully deterministic memory skeletons.
	\qedEx
\end{rem}

\paragraph{Outcomes.}
Let $(\arena = \arenaFull, \initStates)$ be an initialized arena.
For $\hist \in \Hists(\arena, \initStates)$, we denote
\[
	\Cyl{\hist} = \{\play \in \Plays(\arena, \initStates) \mid \hist \text{ is a prefix of } \play\}
\]
for the \emph{cylinder of $\hist$}, that is, the set of plays (which are infinite) starting with $\hist$.
We denote by $\oalg_{(\arena, \initStates)}$ the smallest $\sigma$-algebra generated by all the cylinders of histories in $\Hists(\arena, \initStates)$.
Hence, $(\Plays(\arena, \initStates), \oalg_{(\arena, \initStates)})$ is a measurable space.

When both players have decided on a strategy and an initial state has been chosen, the generated object is a (finite or countably infinite) Markov chain, which induces a probability distribution on the plays.
More precisely, for strategies $\strat_1$ of $\Pone$ and $\strat_2$ of $\Ptwo$ on $(\arena, \initStates)$ and $\st\in\initStates$, we denote $\prSolo{\arena}{\st}{\strat_1, \strat_2}$ for the probability distribution on $(\Plays(\arena, \st), \oalg_{(\arena, \st)})$ induced by $\strat_1$ and $\strat_2$, starting from state $\st$.
This distribution is defined on the set of cylinders as follows: for $\hist = \usualHist\in\Hists(\arena, \st)$ (so $\st = \st_0$),
\[
	\pr{\arena}{\st}{\strat_1, \strat_2}{\Cyl{\hist}} =
	\prod_{j = 1}^{n} \sigma_{i_j}(\st_0\ldots\action_{j-1}\st_{j-1})(\action_j) \cdot
	\transProb(\st_{j-1}, \action_{j}, \st_j),
\]
where $i_j = 1$ if $\st_{j-1} \in \states_1$, and $i_j = 2$ if $\st_{j-1} \in \states_2$.
This pre-measure can be uniquely extended to $(\Plays(\arena, \st), \oalg_{(\arena, \st)})$ by Carath\'eodory's extension theorem~\cite[Theorem~A.1.3]{Dur19}, as the class of cylinders is a semi-ring of sets that generates the whole $\sigma$-algebra.

Similarly, we define $\oalg$ to be the smallest $\sigma$-algebra on $\colors^\omega$ generated by the set of all cylinders on $\colors$.
We can extend $\colHatSolo$ to distributions over $(\Plays(\arena, \st), \oalg_{(\arena, \st)})$: for $\dist\in\Dist(\Plays(\arena, \st), \oalg_{(\arena, \st)})$, we write $\colHat{\dist}$ for the distribution $\dist\circ\colHatSolo^{-1}\in\Dist(\colors^\omega,\oalg)$.
In particular, every probability distribution $\prSolo{\arena}{\st}{\strat_1, \strat_2}$ naturally induces a probability distribution $\colHat{\prSolo{\arena}{\st}{\strat_1, \strat_2}}$ over $(\colors^\omega, \oalg)$ through the $\colHatSolo$ function, which we denote $\prcSolo{\arena}{\st}{\strat_1, \strat_2}$.

\paragraph{Preferences.}
To specify each player's objective, or preference, we use the general notion of \emph{preference relation}.
\begin{defi}[Preference relation]
	A \emph{preference relation $\pref$ (on $\colors$)} is a total preorder over $\Dist(\colors^\omega, \oalg)$.
\end{defi}
The idea is that $\Pone$ favors the distributions in $\Dist(\colors^\omega, \oalg)$ that are the largest for $\pref$, and as we are studying \emph{zero-sum} games, $\Ptwo$ favors the distributions that are the smallest for $\pref$.
Equivalently, $\Ptwo$'s goal is to obtain the largest distribution for the \emph{inverse preference relation $\invPref$}, defined as $\dist \mathrel{\invPref} \dist'$ if and only if $\dist' \pref \dist$.
For $\pref$ a preference relation and $\dist, \dist'\in\Dist(\colors^\omega, \oalg)$, we write $\dist\strictPref\dist'$ if $\dist\pref\dist'$ and $\dist'\not\pref\dist$.

Depending on the context, it might not be necessary to define a preference relation as \emph{total}: it is sufficient to order distributions that can arise as an element $\prSolo{\arena}{\st}{\strat_1, \strat_2}$ in the context.
For example, in the specific case of deterministic games in which only pure strategies are considered, all distributions that arise are always Dirac distributions on a single infinite word in $\colors^\omega$.
In this context, it is therefore sufficient to define a total preorder over all Dirac distributions (which we can then see as infinite words, giving a definition of preference relation similar to~\cite{GZ05,BLORV22}).
Another reasonable constraint comes from the fact that we mostly consider finite arenas, in which only finitely many colors appear: distributions that generate infinite words with infinitely many colors should therefore not be considered.

\begin{exa}\label{ex:objExamples}
	We give three examples corresponding to three different ways to encode preference relations.
	First, a preference relation can be induced by an event $\wc\in\oalg$ called a \emph{winning condition}, which consists of infinite sequences of colors.
	The objective of $\Pone$ is to maximize the probability that the event $\wc$ happens.
	An event $\wc$ naturally induces a preference relation $\pref_\wc$ such that for $\dist,\dist'\in\Dist(\colors^\omega, \oalg)$, $\dist\pref_\wc\dist'$ if and only if $\dist(\wc) \le \dist'(\wc)$.
	For $\colors = \IN$, we give the example of the \emph{weak parity} winning condition $\simPar$~\cite{Tho08}, defined as
	\[
	\simPar = \{\clr_1\clr_2\ldots\in\colors^\omega \mid
	\max_{j\ge 1} \clr_j\text{ exists and is even}\}.
	\]
	In finite arenas, the value $\max_{j\ge 1} \clr_j$ always exists, as there are only finitely many colors that appear.
	This is different from the classical parity condition, which requires the maximal color \emph{seen infinitely often} to be even, and not just the maximal color seen.
	In particular, $\simPar$ is not prefix-independent.

	A preference relation can also be induced by a Borel \emph{(real) payoff function $\pf\colon \colors^\omega \to \IR$}.
	For example, if $\colors = \IR$ and $\discFac \in \intervaloo{0,1}$, a classical payoff function~\cite{Sha53} is the \emph{discounted sum $\disc$}, defined for $\clr_1\clr_2\ldots\in\colors^\omega$ as
	\[
		\disc(\clr_1\clr_2\ldots) = \lim_n \sum_{i = 0}^n \discFac^i\cdot \clr_{i+1}.
	\]
	The goal of $\Pone$ is to maximize the expected value of $\pf$, which is defined for a probability distribution $\dist\in\Dist(\colors^\omega, \oalg)$ as $\expectation_\dist[\pf] = \int \pf \ud\dist$.
	A payoff function $\pf$ naturally induces a preference relation $\pref_\pf$: for $\dist_1, \dist_2\in\Dist(\colors^\omega, \oalg)$, $\dist_1 \pref_\pf \dist_2$ if and only if $\expectation_{\dist_1}[\pf] \le \expectation_{\dist_2}[f]$.
	Payoff functions are more general than winning conditions: for $\wc$ a winning condition, the preference relation induced by the indicator function of $\wc$, which is a payoff function, corresponds to the preference relation induced by $\wc$.

	It is also possible to specify preference relations that cannot be expressed as a payoff function.
	An example is given in~\cite{CFKSTU12}: we assume that the goal of $\Pone$ is to see color $\clr\in\colors$ with probability \emph{precisely} $\frac{1}{2}$.
	We denote the event of seeing color $\clr$ as $\F\clr \in \oalg$.
	Then for $\dist,\dist'\in\Dist(\colors^\omega, \oalg)$, $\dist\pref\dist'$ if and only if $\dist(\F\clr) \neq \frac{1}{2}$ or $\dist'(\F\clr) = \frac{1}{2}$.
	\qedEx
\end{exa}

Combining an initialized arena, describing how the players interact with each other, and a preference relation, describing both players' objectives, defines an \emph{initialized game}.
\begin{defi}[Initialized game]
	A (two-player stochastic turn-based zero-sum) \emph{initialized game} is a tuple $\game = \initGameFull$, where $\initArena$ is an initialized arena and $\pref$ is a preference relation.
\end{defi}

\paragraph{Optimality of strategies.}
Let $\game = \initGameFull$ be an initialized game and let $\stratType\in\{\PFM, \pure, \RFM, \rand\}$ be a type of strategies.
For $\st\in\initStates$, $\strat_1\in\stratsType{1}{\stratType}{\arena, \initStates}$, we define
\[
	\UCol_\pref^\stratType(\arena, \st, \strat_1) = \{\dist\in\Dist(\colors^\omega, \oalg) \mid
	\exists \strat_2\in\stratsType{2}{\stratType}{\arena, \st}, \prcSolo{\arena}{\st}{\strat_1, \strat_2} \pref \dist\}.
\]
The set $\UCol_\pref^\stratType(\arena, \st, \strat_1)$ corresponds to all the distributions that are at least as good for $\Pone$ (w.r.t.\ $\pref$) as a distribution that $\Ptwo$ can induce by playing a strategy $\strat_2$ of type $\stratType$ against $\strat_1$; this set is upward-closed w.r.t.\ $\pref$.
We can define a similar operator for strategies of $\Ptwo$: for $\st\in\initStates$, $\strat_2\in\stratsType{2}{\stratType}{\arena, \initStates}$,
\[
	\DCol_\pref^\stratType(\arena, \st, \strat_2) = \{\dist\in\Dist(\colors^\omega, \oalg) \mid
	\exists \strat_1\in\stratsType{1}{\stratType}{\arena, \st}, \dist \pref \prcSolo{\arena}{\st}{\strat_1, \strat_2}\}.
\]

For $\strat_1, \strat_1'\in\stratsType{1}{\stratType}{\arena, \initStates}$, we say that $\strat_1$ is \emph{at least as good as $\strat_1'$ from $\st\in\initStates$ under $\stratType$ strategies} if
\[
	\UCol_\pref^\stratType(\arena, \st, \strat_1) \subseteq \UCol_\pref^\stratType(\arena, \st, \strat_1').
\]
This inclusion means that the best replies of $\Ptwo$ against $\strat_1'$ yield an outcome that is at least as bad for $\Pone$ (w.r.t.\ $\pref$) as the best replies of $\Ptwo$ against $\strat_1$.

Symmetrically for $\strat_2, \strat_2'\in\stratsType{2}{\stratType}{\arena, \initStates}$, we say that $\strat_2$ is \emph{at least as good as $\strat_2'$ from $\st\in\initStates$ under $\stratType$ strategies} if
\[
	\DCol_\pref^\stratType(\arena, \st, \strat_2) \subseteq \DCol_\pref^\stratType(\arena, \st, \strat_2').
\]

\begin{defi}[Optimal strategy]
	Let $\game = \initGameFull$ be an initialized game and $\stratType\in\{\PFM, \pure, \RFM, \rand\}$ be a type of strategies.
	A strategy $\strat_i\in\stratsType{i}{\stratType}{\arena, \initStates}$ is \emph{$\stratType$-optimal in $\game$} if it is at least as good under $\stratType$ strategies as any other strategy in $\stratsType{i}{\stratType}{\arena, \initStates}$ from all $\st\in\initStates$.
\end{defi}

When the considered preference relation $\pref$ is clear from the context, we often talk about $\stratType$-optimality in an initialized arena $(\arena, \initStates)$ to refer to $\stratType$-optimality in the initialized game $(\arena, \initStates, \pref)$.
Notice that given two \emph{isomorphic} initialized arenas, there is an obvious bijection between the strategies on them, and the properties of the strategies (pure, memoryless, finite-memory, $\stratType$-optimal\ldots) are preserved through this bijection.

Our goal will be to understand, given a preference relation, a class of arenas, and a type of strategies, what kinds of strategies are sufficient to play optimally.
In the following definition, abbreviations \emph{AIFM} and \emph{FM} stand respectively for \emph{arena-independent finite-memory} and \emph{finite-memory}.
\begin{defi}[Sufficiency of strategies]
	Let $\pref$ be a preference relation, $\arenaFam$ be a class of initialized arenas, $\stratType\in\allStrats$ be a type of strategies, and $\memSkel$ be a memory skeleton.
	\begin{itemize}
		\item We say that \emph{pure memoryless strategies} (resp.\ \emph{pure strategies based on $\memSkel$}) \emph{suffice to play $\stratType$-optimally in $\arenaFam$ for $\Pone$} if for all $\initArena\in\arenaFam$, $\Pone$ has a pure memoryless strategy (resp.\ a pure strategy based on $\memSkel$) that is $\stratType$-optimal in $\initArena$.
		\item We say that \emph{pure AIFM strategies suffice to play $\stratType$-optimally in $\arenaFam$ for $\Pone$} if there exists a memory skeleton $\memSkel$ such that pure strategies based on $\memSkel$ suffice to play $\stratType$-optimally for all arenas in $\arenaFam$ for $\Pone$.
		Observe that $\memSkel$ may not depend on the arena.
		\item We say that \emph{pure FM strategies suffice to play $\stratType$-optimally in $\arenaFam$ for $\Pone$} if for all $\initArena\in\arenaFam$, there exists a memory skeleton $\memSkel$ such that $\Pone$ has a pure strategy based on $\memSkel$ that is $\stratType$-optimal in $\initArena$.
	\end{itemize}
\end{defi}

\noindent
If $\arenaFam$ is clear in the context (typically all initialized deterministic or stochastic arenas), we often omit it.
When no type of strategies is specified, it means that we consider optimality against all (general) strategies.

Since memoryless strategies are a specific kind of finite-memory strategies based on the same memory skeleton $\memSkelTriv$, the sufficiency of pure memoryless strategies is equivalent to the sufficiency of pure strategies based on $\memSkelTriv$, and is therefore just a specific case of the sufficiency of pure AIFM strategies.
Notice the difference between the order of quantifiers for AIFM and FM strategies: the sufficiency of pure AIFM strategies implies the sufficiency of pure FM strategies, but the opposite is false, as we show in the following example.

\begin{exa}
	Let us consider the \emph{energy parity} winning condition studied in deterministic arenas in~\cite{CD12b}.
	We do not explain this winning condition in detail, but comment on its memory requirements and how it illustrates the difference between AIFM and FM strategies.
	For this objective, in deterministic arenas, pure memoryless strategies suffice to play optimally for $\Ptwo$.%
	\footnote{Work on deterministic arenas only consider $\pure$-optimality, but as pure strategies suffice for Borel objectives~\cite{Mar75}, this implies $\rand$-optimality.} 
	On the other hand, $\Pone$ can play optimally with finite memory, but the memory needed \emph{depends on the number of states of the arena}; it is not arena-independent.
	Therefore, pure FM strategies suffice in deterministic games for $\Pone$, but not pure AIFM strategies: $\Pone$ needs to change its memory skeleton depending on the arena, and no single memory skeleton is sufficient to play optimally in all deterministic arenas (even with a fixed and finite number of colors).

	Interestingly, it is shown in~\cite{MSTW17} that the same winning condition needs infinite memory in (even one-player) stochastic arenas for $\Pone$, which shows that the sufficiency of pure FM strategies in deterministic arenas does not imply the sufficiency of pure FM strategies in stochastic arenas.

	Now let us reconsider the weak parity winning condition $\simPar$ introduced in Example~\ref{ex:objExamples}: the goal of $\Pone$ is to maximize the probability that the greatest color seen is even.
	As will be proven formally in Section~\ref{sec:app} thanks to the results of this article, to play optimally in any stochastic game, it is sufficient for both players to remember the greatest color already seen, which can be implemented by the memory skeleton $\Mmax = (\IN, 0, (\memState, n) \mapsto \max\{\memState, n\})$.
	As explained above, this memory skeleton has an infinite state space, but as there are only finitely many colors in every (finite) arena, only a finite part of the skeleton is sufficient to play optimally in a given arena.
	The size of the skeleton used for a fixed arena depends on the appearing colors, but for a fixed number of colors, it does not depend on parameters of the arena (such as its state and action spaces).
	Therefore pure AIFM strategies suffice to play optimally for both players, and more precisely pure strategies based on $\Mmax$ suffice for both players.
	\qedEx
\end{exa}

We define a second stronger notion related to optimality of strategies, which is the notion of \emph{subgame perfect strategy}: a strategy is \emph{subgame perfect} in a game if it reacts optimally to all histories consistent with the arena, even histories not consistent with the strategy itself, or histories that only a non-rational adversary would play~\cite{Osb04}.
This is a desirable property of strategies that is stronger than optimality, since a subgame perfect strategy is not only optimal from the initial position, but from any arbitrary stage (\emph{subgame}) of the game.
In particular, if an opponent plays non-optimally, an optimal strategy that is not subgame perfect does not always fully exploit the advantage that the opponent's suboptimal behavior provides, and may yield a result that is not optimal \emph{when starting in a subgame}.
To do so, we first need an extra definition.
\begin{defi}[Shifted distributions, strategies and preference relations]\label{def:shift}
	For $\word\in\colors^*$, $\dist\in\Dist(\colors^\omega, \oalg)$, we define the \emph{shifted distribution $\word\dist$} as the distribution such that for an event $E\in\oalg$, $\word\dist(E) = \dist(\{\word'\in\colors^\omega \mid \word\word'\in E\})$.
	Note that a shifted distribution evaluates events prepended with $\word$ as if they were evaluated by the original distribution without an extra $\word$ at the start; it is \emph{not} the distribution conditioned on the infinite words whose prefix is $\word$.

	Let $\initArena$ be an initialized arena, and $\strat_i\in\stratsType{i}{\rand}{\arena, \initStates}$.
	For $\hist = \usualHist\in\Hists(\arena, \initStates)$, we define the \emph{shifted strategy} $\shiftStrat{\strat_i}{\hist}\in\stratsType{i}{\rand}{\arena, \histOut{\hist}}$ which is such that, for $\hist' = \st_n\action_{n+1}\st_{n+1}\ldots\action_m\st_m\in\Hists_i(\arena, \histOut{\hist})$, $\shiftStrat{\strat_i}{\hist}(\hist') = \strat_i(\st_0\action_1\st_1\ldots\action_m\st_m)$.

	For $\pref$ a preference relation and $\word\in\colors^*$, we define the \emph{shifted preference relation} $\shiftPref{\word}$ as the preference relation such that for $\dist, \dist'\in\Dist(\colors^\omega, \oalg)$, $\dist \shiftPref{\word} \dist'$ if and only if $\word\dist \pref \word\dist'$.
\end{defi}

\begin{defi}[Subgame perfect strategy]
	Let $\game = \initGameFull$ be an initialized game and $\stratType\in\allStrats$ be a type of strategies.
	A strategy $\strat_i\in\stratsType{i}{\stratType}{\arena, \initStates}$ is \emph{$\stratType$-subgame perfect ($\stratType$-SP) in $\game$} if for all $\hist\in\Hists(\arena, \initStates)$, shifted strategy $\shiftStrat{\strat_i}{\hist}$ is $\stratType$-optimal in the initialized game $(\arena, \histOut{\hist}, \shiftPref{\colHat{\hist}})$.
\end{defi}
Strategies that are $\stratType$-SP are in particular $\stratType$-optimal; the converse is not true in general.

For technical reasons, we will use the notion of \emph{equilibrium} as a tool in proofs to show the existence of optimal strategies.
\begin{defi}[Nash, SP equilibrium]
	Let $\game = \initGameFull$ be an initialized game and $\stratType\in\allStrats$ be a type of strategies.
	A pair of strategies $(\strat_1, \strat_2) \in \stratsType{1}{\stratType}{\arena, \initStates}\times\stratsType{2}{\stratType}{\arena, \initStates}$ is an \emph{$\stratType$-Nash equilibrium ($\stratType$-NE) in $\game$} if for all $\strat_1'\in\stratsType{1}{\stratType}{\arena, \initStates}$, for all $\strat_2'\in\stratsType{2}{\stratType}{\arena, \initStates}$, for all $\st\in\initStates$,
	\[
	\prcSolo{\arena}{\st}{\strat_1', \strat_2} \pref
	\prcSolo{\arena}{\st}{\strat_1, \strat_2} \pref
	\prcSolo{\arena}{\st}{\strat_1, \strat_2'}.
	\]
	We say that $(\strat_1, \strat_2)$ is an \emph{$\stratType$-subgame perfect equilibrium ($\stratType$-SPE)} if for all $\hist\in\Hists(\arena, \initStates)$, the pair of strategies $(\shiftStrat{\strat_1}{\hist}, \shiftStrat{\strat_2}{\hist})$ is an $\stratType$-Nash equilibrium in $(\arena, \histOut{\hist}, \shiftPref{\colHat{\hist}})$.
\end{defi}
A pair of strategies $(\strat_1, \strat_2)$ is thus an $\stratType$-NE if no player has any interest in unilaterally deviating from its strategy (using strategies of type $\stratType$), as the induced probability distribution on the colors (or equivalently, the induced Markov chain) would not be better for this player than not changing the strategy.
In the zero-sum context, if $(\strat_1, \strat_2)$ is an $\stratType$-NE (resp.\ $\stratType$-SPE), then both $\strat_1$ and $\strat_2$ are $\stratType$-optimal (resp.\ SP).

We say that a pair of strategies $(\strat_1, \strat_2)$ is pure (resp.\ randomized, memoryless, based on $\memSkel$) if both $\strat_1$ and $\strat_2$ are pure (resp.\ randomized, memoryless, based on $\memSkel$).
Thanks to the fact that we consider zero-sum games, we can use the following handy result about Nash equilibria.

\begin{lem}\label{lem:mixNE}
	Let $\game = \initGameFull$ be a game and $\stratType\in\allStrats$ be a type of strategies.
	Let $(\strat_1^a, \strat_2^a), (\strat_1^b, \strat_2^b)\in\stratsType{1}{\stratType}{\arena, \initStates}\times\stratsType{2}{\stratType}{\arena, \initStates}$ be two $\stratType$-NE (resp.\ $\stratType$-SPE) in $\game$.
	Then $(\strat_1^a, \strat_2^b)$ is also an $\stratType$-NE (resp.\ $\stratType$-SPE) in $\game$.
\end{lem}
\begin{proof}
	A very similar proof appears in~\cite[Lemma~2.5]{BLORV22}.
	We do the proof for NE, and the result about SPE follows since its definition uses the notion of NE\@.
	We need to prove that for all $\strat_1'\in\stratsType{1}{\stratType}{\arena, \initStates}$, for all $\strat_2'\in\stratsType{2}{\stratType}{\arena, \initStates}$, for all $\st\in\initStates$,
	\begin{equation}\label{eq:mixingNEproof}
		\prcSolo{\arena}{\st}{\strat_1', \strat_2^b} \pref
		\prcSolo{\arena}{\st}{\strat_1^a, \strat_2^b} \pref
		\prcSolo{\arena}{\st}{\strat_1^a, \strat_2'}.
	\end{equation}

	Since $(\strat^a_1, \strat^a_2)$ is an $\stratType$-NE, we know that
	\[
	\prcSolo{\arena}{\st}{\strat^b_1, \strat^a_2} \pref \prcSolo{\arena}{\st}{\strat^a_1, \strat^a_2} \pref \prcSolo{\arena}{\st}{\strat^a_1, \strat^b_2},
	\]
	instantiating $\strat'_1$ and $\strat'_2$ as $\strat_1^b$ and $\strat^b_2$ in the definition of $\stratType$-NE\@.
	Similarly, since $(\strat^b_1, \strat^b_2)$ is an $\stratType$-NE, we know that
	\[
	\prcSolo{\arena}{\st}{\strat^a_1, \strat^b_2} \pref \prcSolo{\arena}{\st}{\strat^b_1, \strat^b_2} \pref \prcSolo{\arena}{\st}{\strat^b_1, \strat^a_2},
	\]
	instantiating $\strat'_1$ and $\strat'_2$ as $\strat_1^a$ and $\strat^a_2$ in the definition of $\stratType$-NE\@.

	One can see from the last two lines that all six probability distributions over sequences of colors are equivalent w.r.t.\ $\pref$ as the inequalities form a cycle.
	Now, let $\strat_1'\in\stratsType{1}{\stratType}{\arena, \initStates}$ and $\strat_2'\in\stratsType{2}{\stratType}{\arena, \initStates}$.
	Since $(\strat^a_1, \strat^a_2)$ and $(\strat^b_1, \strat^b_2)$ are both $\stratType$-NE, and since $\prcSolo{\arena}{\st}{\strat^a_1, \strat^b_2}$ is equivalent w.r.t.\ $\pref$ to both $\prcSolo{\arena}{\st}{\strat^a_1, \strat^a_2}$ and $\prcSolo{\arena}{\st}{\strat^b_1, \strat^b_2}$, we obtain
	\[
	\prcSolo{\arena}{\st}{\strat_1', \strat_2^b} \pref \prcSolo{\arena}{\st}{\strat^b_1, \strat^b_2} \pref
	\prcSolo{\arena}{\st}{\strat_1^a, \strat_2^b} \pref \prcSolo{\arena}{\st}{\strat^a_1, \strat^a_2} \pref
	\prcSolo{\arena}{\st}{\strat_1^a, \strat_2'},
	\]
	thus~\eqref{eq:mixingNEproof} is verified.
\end{proof}

\paragraph{Operations on arenas.}
We introduce two operations on arenas that we will use multiple times through the course of this article.

For $\arena$ an arena, $\word\in\colors^*$, and $\st$ a state of $\arena$, we write
\begin{equation}\label{eq:prefExt}
\arenaPref{\arena}{\word}{\st}
\end{equation}
for the \emph{prefix-extended arena} that consists of arena $\arena$ with an extra ``chain'' of states leading up to $\st$ with the same colors as $\word$.
Formally, if $\arena = \arenaFull$, $\word = \clr_1\clr_2\ldots\clr_n$, and $\st\in\states$, we define $\arenaPref{\arena}{\word}{\st}$ as the arena $\arenaFullIndex{\prime}$ where $\states'_1 = \states_1\disjUnion\{\st_0^\word,\ldots,\st_{n-1}^\word\}$, $\states'_2 = \states_2$; $\act' = \act \disjUnion \{\action^\word\}$, $\restr{\act'}{\states} = \act$, and for $i$, $0\le i \le n-1$, $\act'(\st_i^\word) = \{\action^\word\}$; $\restr{\transProb'}{\states\times\act} = \transProb$, for $i$, $0\le i < n-1$, $\transProb'(\st_i^\word, \action^\word, \st_{i+1}^\word) = 1$, and $\transProb'(\st_{n-1}^\word, \action^\word, \st) = 1$; $\restr{\colSolo'}{\states\times\act} = \colSolo$, and for $i$, $0\le i\le n-1$, $\colSolo'(\st_i^\word, \action^\word) = \clr_{i+1}$.

If $\dist = \prcSolo{\arena}{\st}{\strat_1,\strat_2}$ is induced by strategies $\strat_1$ and $\strat_2$ on some initialized arena $(\arena, \st)$, notice that the shifted distribution $\word\dist$ equals $\prcSolo{\arenaPref{\arena}{\word}{\st}}{\st_0^\word}{\strat_1',\strat_2'}$, where $\strat_1'$ plays the only available action $\action^\word$ until it reaches $\st$, and then $\strat_1'$ and $\strat_2'$ play like $\strat_1$ and $\strat_2$, ignoring they ever saw $\word$.

For two arenas $\arena_1$ and $\arena_2$ with disjoint state spaces, if $\st_1$ and $\st_2$ are two states controlled by $\Pone$ that are respectively in $\arena_1$ and $\arena_2$ with disjoint sets of available actions, we write
\begin{equation}\label{eq:merge}
(\arena_1, \st_1)\arenaJoin(\arena_2, \st_2)
\end{equation}
for the \emph{merged arena} in which $\st_1$ and $\st_2$ are merged, and everything else is kept the same.
The merged state which comes from the merge of $\st_1$ and $\st_2$ is usually called $\joinState$.
Formally, let $\arena_1 = \arenaFullIndex{1}$, $\arena_2 = \arenaFullIndex{2}$, $\st_1\in\states_1^1$, and $\st_2\in\states_1^2$.
We assume that $\states^1\cap\states^2 = \emptyset$ and that $\actions{\st_1}\cap\actions{\st_2} = \emptyset$.
We define $(\arena_1, \st_1)\arenaJoin(\arena_2, \st_2)$ as the arena $\arenaFull$ with $\states_1 = \states_1^1 \disjUnion \states_1^2 \disjUnion \{\joinState\} \setminus \{\st_1, \st_2\}$, $\states_2 = \states_2^1 \disjUnion \states_2^2$; $\actions{\joinState} = \act^1(\st_1) \disjUnion \act^2(\st_2)$ and all the other available actions are kept the same as in the original arenas;
for $i\in\{1, 2\}$, $\transProb(\joinState, \action) = \transProb^i(\joinState, \action)$ if $\action\in\actions{\st_i}$ and all the other transitions are kept the same as in the original arenas (with transitions going to $\st_1$ or $\st_2$ being directed to $\joinState$); for $i\in\{1, 2\}$, $\col{\joinState, \action} = \colSolo^i(\joinState, \action)$ if $\action\in\actions{\st_i}$ and all the other colors are kept the same as in the original arenas.
A symmetrical definition can be written if $\st_1$ and $\st_2$ are both controlled by $\Ptwo$.

In practice, we often consider classes $\arenaFam$ of initialized arenas that are \emph{closed with respect to some operation} --- we specify the exact meaning for each operation we will use here:
\begin{itemize}
	\item for $\memSkel$ a memory skeleton, $\arenaFam$ is \emph{closed by product with $\memSkel$} if for all $\initArena\in\arenaFam$, $\prodAS{\initArena}{\memSkel}\in\arenaFam$;
	\item $\arenaFam$ is \emph{closed by prefix-extension} if for all $\initArena\in\arenaFam$, for all $\word\in\colors^*$, for all states $\st$ of $\arena$, $(\arenaPref{\arena}{\word}{\st}, \initStates \cup \{\st_0^\word\})\in\arenaFam$.
\end{itemize}

\noindent
Standard classes of arenas are all closed by these operations: we give as examples the classes of all initialized one-player deterministic arenas of $\Pone$, one-player stochastic arenas of $\Pone$, two-player deterministic arenas, and two-player stochastic arenas (corresponding respectively to the classes of $1$-player, $1\frac{1}{2}$-player, $2$-player, $2\frac{1}{2}$-player arenas often found in the literature).
Throughout the article, we state results with the exact required closure properties for generality, but most applications use such standard classes.

\section{Coverability and subgame perfect strategies}\label{sec:memory}
In this section, we establish a few key results about memory and playing optimally.
The main tool is given by Lemma~\ref{prop:optProdIsCov}, which can be used to reduce questions about the sufficiency of AIFM strategies in reasonable classes of initialized arenas to the sufficiency of memoryless strategies in a subclass.
We end the section by showing the use of Lemma~\ref{prop:optProdIsCov} in the proof of our first main result (Theorem~\ref{prop:optAIFMimpliesSP}), which shows that the sufficiency of pure AIFM strategies implies the stronger existence of pure AIFM SP strategies in well-behaved classes of initialized arenas.

First, we restate the intuitive result linking playing optimally with memory $\memSkel$ in an initialized arena and playing optimally with a memoryless strategy in its product with $\memSkel$.

\begin{restatable}{lem}{memProd}\label{lem:memProd}
	Let $\pref$ be a preference relation, $\stratType\in\allStrats$ be a type of strategies, and $\memSkel = \memSkelFull$ be a memory skeleton.
	Let $\game = \initGameFull$ be an initialized game, and let $\strat_i \in \stratsType{i}{\stratType}{\arena, \initStates}$ be a finite-memory strategy encoded by a Mealy machine $\mealy = (\memSkel, \memNxt)$.
	Then, $\strat_i$ is $\stratType$-optimal in $\game$ if and only if $\memNxt$ corresponds to a memoryless $\stratType$-optimal strategy in game $\game' = (\prodAS{(\arena, \initStates)}{\memSkel}, \pref)$.
\end{restatable}

We defer the proof of this result to Appendix~\ref{sec:memProdProof}; a proof of a very similar result can be found in~\cite[Lemma~2.4]{BLORV22}.
Lemma~\ref{lem:memProd} can be restated for SP strategies, for NE and for SPE with a similar proof.

We now define a property of initialized arenas called \emph{coverability by $\memSkel$} (for a memory skeleton $\memSkel$), which happens to characterize initialized arenas that are a product with $\memSkel$ (Lemma~\ref{lem:prodCov}).
Albeit intuitive, this is a key technical step, as the class of arenas covered by a memory skeleton is sufficiently well-behaved to support edge-induction arguments, whereas it is more difficult to perform such techniques directly on the class of product arenas: removing a single edge from a product arena makes it hard to express as a product arena, whereas it is clear that coverability is preserved.

\begin{defi}[Coverability by $\memSkel$]
	An initialized arena \emph{$(\arenaFull, \initStates)$ is covered by memory skeleton $\memSkel = \memSkelFull$} if there exists a function $\coverFunction\colon \states \to \memStates$ such that for all $\st\in\initStates$, $\coverFunction(\st) = \memInit$, and for all $(\st, \action, \st')\in\transProb$, $\memUpd(\coverFunction(\st), \col{\st, \action}) = \coverFunction(\st')$.
\end{defi}
This property means that it is possible to assign a unique memory state to each arena state such that transitions of the arena always update the memory state in a way that is consistent with the memory skeleton.
Note that isomorphism of initialized arenas preserves coverability by any memory skeleton.
Also, every initialized arena is covered by $\memSkelTriv$, which is witnessed by the constant function $\coverFunction$ associating $\memInit$ to every state.

Definitions close to our notion of \emph{coverability by $\memSkel$} were introduced for deterministic arenas in~\cite{KopThesis,BLORV22}.
If we restrict our definition to deterministic arenas, the definition of \emph{adherence with $\memSkel$} in~\cite[Definition~8.12]{KopThesis} is very similar, but does not distinguish initial states from the rest (neither in the arena nor in the memory skeleton) --- the reason is that~\cite{KopThesis} only considers \emph{prefix-independent} objectives, for which selecting the right initial memory state is not as important (see~\cite[Proposition~8.2]{KopThesis}).
Our property of $(\arena, \initStates)$ being \emph{covered by $\memSkel$} is also equivalent to $\arena$ being both \emph{prefix-covered} and \emph{cyclic-covered} by $\memSkel$ from $\initStates$, as defined in~\cite{BLORV22}.
Distinguishing both notions gives insight in~\cite{BLORV22} as they are used at different places in proofs (\emph{prefix-covered} along with \emph{monotony}, and \emph{cyclic-covered} along with \emph{selectivity}).
Here, we opt for a single concise definition, as most of our proofs do not mention \emph{monotony} and \emph{selectivity}.

\begin{rem}\label{rmk:stochCover}
	Our definition of coverability is helped by the fact that our memory skeletons are deterministic and do not allow stochastic updates (cf.\ Remark~\ref{rmk:stochSkel}).
	Allowing for stochastic updates may lead to smaller memory requirements~\cite{DBLP:journals/lmcs/ChatterjeeKK17,MR22} and would be one way to extend our results to deal with strategies that are not only pure.
	Yet, it appears difficult to extend this idea of coverability to skeletons with stochastic updates, as the same transition in the arena may lead to two memory states, which a function $\coverFunction\colon \states \to \memStates$ cannot deal with.
	\qedEx
\end{rem}

We link products and coverability: first, product initialized arenas are covered; second, covered initialized arenas are exactly the ones that are isomorphic to their own product.

\begin{lem}\label{lem:prodCov}
	Let $\memSkel$ be a memory skeleton and $\initArena$ be an initialized arena.
	The product initialized arena $\prodAS{\initArena}{\memSkel}$ is covered by $\memSkel$.
	Moreover, $\initArena$ is covered by $\memSkel$ if and only if $\initArena$ is isomorphic to $\prodAS{\initArena}{\memSkel}$.
\end{lem}

\begin{proof}
	Let $\memSkel = \memSkelFull$, $\arena = \arenaFull$.
	We show that $\prodAS{(\arena, \initStates)}{\memSkel}$ is covered by $\memSkel$.
	Let $\coverFunction\colon \states \times \memStates \to \memStates$ be the projection on $\memStates$ (that is, $\coverFunction(\st, \memState) = \memState$ for all $(\st, \memState)\in\states\times\memStates$).
	This function witnesses that $\prodAS{(\arena, \initStates)}{\memSkel}$ is covered, by definition of product initialized arena: only transitions that are consistent with the memory skeleton are allowed.

	Assume now $(\arena, \initStates)$ is covered by $\memSkel$, witnessed by function $\coverFunction$.
	We show that $(\arena, \initStates)$ is isomorphic to $\prodAS{(\arena, \initStates)}{\memSkel}$.
	In $\prodAS{(\arena, \initStates)}{\memSkel}$, it is not possible to reach two states $(\st, \memState)$, $(\st, \memState')$ with $\memState \neq \memState'$ from $\initStates\times\{\memInit\}$: otherwise, this would contradict that $(\arena, \initStates)$ is covered by $\memSkel$.
	Hence the function $\bijSt\colon \st \mapsto (\st, \coverFunction(\st))$ is a bijection between states of $(\arena, \initStates)$ and states of $\prodAS{(\arena, \initStates)}{\memSkel}$ (remember that we only keep the reachable states of the product initialized arena).
	Moreover, all the actions, transitions and colors are preserved, by definition of product initialized arena.
	Hence $(\arena, \initStates)$ is isomorphic to its own product with $\memSkel$.
	Now for the other direction, assume $\initArena$ is isomorphic to $\prodAS{\initArena}{\memSkel}$.
	Thus, $\initArena$ can be expressed as a product with $\memSkel$ and is thus covered by $\memSkel$ by the first claim.
\end{proof}

This last lemma shows in some sense an equivalence between a product and a covered initialized arena: product initialized arenas are covered (a similar result for non-initialized arenas is discussed in~\cite[Lemma~3.5]{BLORV22}) and conversely, covered initialized arenas can be written as a product.
The latter implication requires the use of \emph{initialized} arenas to be expressed in a concise way: if a game could always start from any state of an arena, taking the product with a memory skeleton would virtually always make the arena grow, and it could therefore not be isomorphic to its own product.
That is one of the main reasons we resort to \emph{initialized} arenas: we are therefore able to talk interchangeably about \emph{being a product}, which is a technical property at the core of the idea of \emph{playing with memory}, and about \emph{coverability}, which is a more intuitive, easy-to-check condition that trivially benefits from nice closure properties.

We establish two easy consequences of the previous lemmas to have a better understanding of the links between covered and product initialized arenas.
\begin{cor}\label{lem:M1M2}
	Let $\memSkel_1$ and $\memSkel_2$ be two memory skeletons.
	An initialized arena $\initArena$ is covered by $\memSkel_1$ and by $\memSkel_2$ if and only if it is covered by $\memSkel_1\memProduct\memSkel_2$.
\end{cor}
\begin{proof}
	Let $\initArena$ be covered by $\memSkel_1$ and by $\memSkel_2$.
	It is thus isomorphic to $\prodAS{(\prodAS{(\arena, \initStates)}{\memSkel_1})}{\memSkel_2}$ by applying Lemma~\ref{lem:prodCov} twice.
	Notice that $\prodAS{(\prodAS{(\arena, \initStates)}{\memSkel_1})}{\memSkel_2}$ is isomorphic to $\prodAS{(\arena, \initStates)}{(\memSkel_1\memProduct\memSkel_2)}$ (simply consider the bijection $\bijSt\colon ((\st, \memState_1), \memState_2) \mapsto (\st, (\memState_1, \memState_2))$).
	Hence, $\initArena$ is isomorphic to $\prodAS{(\arena, \initStates)}{(\memSkel_1\memProduct\memSkel_2)}$, and by using Lemma~\ref{lem:prodCov} in the other direction, is covered by $\memSkel_1\memProduct\memSkel_2$.
	Following the arguments backwards yields the other direction of the implication.
\end{proof}

The following lemma sums up our main practical use of the idea of \emph{coverability}, by proving an equivalence between optimal strategies with memory in initialized arenas and memoryless optimal strategies in covered initialized arenas, in classes of arenas with mild hypotheses.
\begin{lem}\label{prop:optProdIsCov}
	Let $\pref$ be a preference relation, $\memSkel$ be a memory skeleton, and let $\stratType\in\{\PFM, \pure, \RFM, \rand\}$ be a type of strategies.
	Let $\arenaFam$ be a class of initialized arenas closed by product with $\memSkel$.
	Then, $\Pone$ has an $\stratType$-optimal (resp.\ $\stratType$-SP) strategy based on $\memSkel$ in all initialized arenas in $\arenaFam$ if and only if $\Pone$ has a memoryless $\stratType$-optimal (resp.\ $\stratType$-SP) strategy in all initialized arenas covered by $\memSkel$ in $\arenaFam$.
	Also, there is an $\stratType$-NE (resp.\ $\stratType$-SPE) based on $\memSkel$ in all initialized arenas in $\arenaFam$ if and only if there is a memoryless $\stratType$-NE (resp.\ $\stratType$-SPE) in all initialized arenas covered by $\memSkel$ in $\arenaFam$.
\end{lem}

\begin{proof}
	We first prove that products of initialized arenas in $\arenaFam$ with $\memSkel$ correspond exactly to initialized arenas of $\arenaFam$ covered by $\memSkel$, that is,
	\begin{equation}\label{eq:prodCov}
		\{\prodAS{\initArena}{\memSkel}\mid \initArena\in\arenaFam\}
		= \{\initArena\in\arenaFam\mid \initArena\text{ is covered by } \memSkel\}.
	\end{equation}
	We start with the left-to-right inclusion.
	Let $\prodAS{\initArena}{\memSkel}$ be an initialized product arena, with $\initArena\in\arenaFam$.
	Then, $\prodAS{\initArena}{\memSkel}$ belongs to $\arenaFam$, as $\arenaFam$ is closed by product with $\memSkel$.
	Moreover, $\prodAS{\initArena}{\memSkel}$ is covered by $\memSkel$ by Lemma~\ref{lem:prodCov}.
	For the right-to-left inclusion, let $\initArena\in\arenaFam$ be covered by $\memSkel$.
	Then, it is isomorphic to $\prodAS{\initArena}{\memSkel}$ by Lemma~\ref{lem:prodCov}.
	Therefore, it can be expressed as the product of an initialized arena of $\arenaFam$ with $\memSkel$.

	We now prove the main statements of the lemma.
	Player $\Pone$ has an $\stratType$-optimal (resp.\ $\stratType$-SP) strategy based on $\memSkel$ in all initialized arenas of $\arenaFam$ if and only if $\Pone$ has a memoryless $\stratType$-optimal (resp.\ $\stratType$-SP) strategy in all products of initialized arenas in $\arenaFam$ with $\memSkel$ (by Lemma~\ref{lem:memProd}) if and only if $\Pone$ has a memoryless $\stratType$-optimal (resp.\ $\stratType$-SP) strategy in all initialized arenas covered by $\memSkel$ in $\arenaFam$ (by~\eqref{eq:prodCov}).

	Similarly, there is an $\stratType$-NE (resp.\ $\stratType$-SPE) based on $\memSkel$ in all initialized arenas in $\arenaFam$ if and only if there is a memoryless $\stratType$-NE (resp.\ $\stratType$-SPE) in all products of initialized arenas in $\arenaFam$ with $\memSkel$ (by Lemma~\ref{lem:memProd}) if and only if there is a memoryless $\stratType$-NE (resp.\ $\stratType$-SPE) in all initialized arenas covered by $\memSkel$ in $\arenaFam$ (by~\eqref{eq:prodCov}).
\end{proof}

We conclude this section by showing that when pure strategies based on the same memory skeleton $\memSkel$ are sufficient to play \emph{optimally}, then pure \emph{SP} strategies based on $\memSkel$ exist.

\begin{thm}\label{prop:optAIFMimpliesSP}
	Let $\pref$ be a preference relation, $\memSkel$ be a memory skeleton, and $\stratType\in\{\PFM, \pure, \RFM, \rand\}$ be a type of strategies.
	Let $\arenaFam$ be a class of initialized arenas closed by product with $\memSkel$ and by prefix-extension.
	If $\Pone$ has pure $\stratType$-optimal strategies based on $\memSkel$ in all initialized arenas of $\arenaFam$, then $\Pone$ has pure $\stratType$-SP strategies based on $\memSkel$ in all initialized arenas of $\arenaFam$.
	If there exist pure $\stratType$-NE based on $\memSkel$ in all initialized arenas of $\arenaFam$, then there exist pure $\stratType$-SPE based on $\memSkel$ in all initialized arenas of $\arenaFam$.
\end{thm}
\begin{proof}
	We start by proving the first claim (going from $\stratType$-optimal to $\stratType$-SP strategies).
	As $\arenaFam$ is closed by product with $\memSkel$, by Lemma~\ref{prop:optProdIsCov}, both the hypothesis and the thesis of this claim can be reformulated for pure memoryless strategies in initialized arenas covered by $\memSkel$.
	We thus prove equivalently that $\Pone$ has pure memoryless $\stratType$-SP strategies in all initialized arenas covered by $\memSkel$ in $\arenaFam$, based on the hypothesis that $\Pone$ has pure memoryless $\stratType$-optimal strategies in all initialized arenas covered by $\memSkel$ in $\arenaFam$.

	Let $(\arena_0, \initStates^0)\in\arenaFam$ be covered by $\memSkel$.
	By hypothesis, $\Pone$ has a pure memoryless $\stratType$-optimal strategy $\strat_1^0$ on $(\arena_0, \initStates^0)$.
	If this strategy is $\stratType$-SP, then we are done.
	If not, then that means that there exists $\hist_0 \in\Hists(\arena_0, \initStates^0)$, with $\st_0 = \histOut{\hist_0}$ and $\word_0 = \colHat{\hist_0}$, such that $\shiftStrat{\strat_1^0}{\hist_0}$ is not $\stratType$-optimal in $(\arena_0, \st_0, \shiftPref{\word_0})$.
	We extend arena $\arena_0$ to a new prefix-extended arena $\arena_1 = \arenaPref{(\arena_0)}{\word_0}{\st_0}$ (this notation was introduced at~\eqref{eq:prefExt}) by ``plugging'' a copy of history $\hist_0$ before $\st_0$.
	We also fix $\initStates^1 = \initStates^0 \disjUnion \{\st_0^{\word_0}\}$, where $\st_0^{\word_0}$ is the first state of the newly added chain with colors similar to $\word_0$.
	Initialized arena $(\arena_1, \initStates^1)$ is in $\arenaFam$ since $\arenaFam$ is closed by prefix-extension.
	We show that $(\arena_1, \initStates^1)$ is covered by $\memSkel$: the covering property holds from $\initStates^0$ because $(\arena_0, \initStates^0)$ was already covered by $\memSkel$ and the newly added states are not reachable from $\initStates^0$, and it holds from $\st_0^{\word_0}$ because the colors up to $\st_0$ are the same as history $\hist_0$ from $\initStates^0$.

	By hypothesis, there exists a pure memoryless $\stratType$-optimal strategy $\strat_1^1$ on $(\arena_1, \initStates^1)$.
	We argue that $\strat_1^1$ is $\stratType$-optimal in $(\arena_0, \st_0, \shiftPref{\word_0})$ (i.e., after seeing $\colHat{\hist_0}$); if it were not, then it would not be $\stratType$-optimal from $(\arena_1, \st_0^{\word_0}, \pref)$, as in both cases, the sequence of colors $\word_0$ is seen before $\st_0$ is reached, and the same (memoryless) strategy is played from $\st_0$.
	The restriction of strategy $\strat_1^1$ to $\Hists(\arena_0, \initStates^0)$ is therefore also $\stratType$-optimal in $(\arena_0, \initStates^0)$, but it is better than $\strat_1^0$ after seeing $\colHat{\hist_0}$.

	If the restriction of $\strat_1^1$ to $\Hists(\arena_0, \initStates^0)$ is $\stratType$-SP in $(\arena_0, \initStates^0)$, then we are done.
	If not, then it means that some history $\hist_1\in\Hists(\arena_0, \initStates^0)$ witnesses that $\strat_1^1$ is not $\stratType$-SP in $(\arena_0, \initStates^0)$.
	Let $\word_1 = \colHat{\hist_1}$, $\st_1 = \histOut{\hist_1}$.
	We can keep going and build an arena $\arena_2 = \arenaPref{(\arena_1)}{\word_1}{\st_1}$, with initial states $\initStates^2 = \initStates^1\disjUnion\{\st_0^{\word_1}\}$, which gives us a pure memoryless $\stratType$-optimal strategy $\strat_1^2$ on $(\arena_2, \initStates^2)$.

	We keep building initialized arenas $(\arena_i, \initStates^i)$ and pure memoryless $\stratType$-optimal strategies $\strat_1^i$ as long as the restrictions of the strategies to $\Hists(\arena_0, \initStates^0)$ are not $\stratType$-SP in $(\arena_0, \initStates^0)$.
	We argue that this iteration ends after a finite number of steps.
	The restriction of every strategy $\strat_1^i$ to $\Hists(\arena_0, \initStates^0)$ is necessarily different from the same restriction for all the previous strategies: for all $j$, $0\le j < i$, $\strat_1^i$ is better than $\strat_1^j$ after seeing history $\hist^j$, and can therefore not be equal to $\strat_1^j$.
	Moreover, there are only finitely many pure memoryless strategies on $(\arena_0, \initStates^0)$ (as this arena is finite), and there is a bijection between \emph{pure memoryless} strategies of arenas $(\arena_i, \initStates^i)$ and of arena $(\arena_0, \initStates^0)$ (as building prefix-extensions does not provide more choices for memoryless strategies).

	Combining that all strategies $\strat_1^i$ are different and the finiteness of the number of strategies shows that the iteration ends, and therefore, that, for some $i\ge 0$, the restriction of the pure memoryless strategy $\strat_1^i$ to $\Hists(\arena_0, \initStates^0)$ is $\stratType$-SP in $(\arena_0, \initStates^0)$.

	The proof to go from pure $\stratType$-NE based on $\memSkel$ to pure $\stratType$-SPE based on $\memSkel$ works in the same way, as there are also finitely many \emph{pairs} of pure memoryless strategies.
\end{proof}

The facts that we consider \emph{finite} arenas and that the hypothesis is about \emph{pure AIFM strategies} are both crucial in the previous proof, as we need the finiteness of the type of strategies considered.

This result shows a major distinction between the sufficiency of AIFM strategies and the more general sufficiency of FM strategies: if a player can always play optimally with the same memory, then SP strategies may be played with the same memory as optimal strategies --- if a player can play optimally but needs arena-\emph{dependent} finite memory, then infinite memory may still be required to obtain SP strategies.
One such example is provided in~\cite[Example~16]{LPR18} for the \emph{average-energy games with lower-bounded energy} objective in deterministic arenas: $\Pone$ can always play optimally with pure finite-memory strategies~\cite[Theorem~13]{BHMRZ17}, but infinite memory is needed for SP strategies.
As will be further explained later, we will also use Theorem~\ref{prop:optAIFMimpliesSP} to gain technical insight in the proof of the main result of Section~\ref{sec:1p}.

\section{One-to-two-player lift}\label{sec:1to2}
Our goal in this section is to obtain a practical tool to help study the memory requirements of two-player stochastic (or deterministic) games.
This tool consists in reducing the study of the sufficiency of pure AIFM strategies for both players in \emph{two}-player games to \emph{one}-player games.
We will first state our result, and the rest of the section is devoted to its proof.
This result mentions two properties of classes of arenas called being \emph{closed by subarena} and \emph{closed by split}, which we will introduce later.
In particular, it can be instantiated with $\arenaFam$ being the class of all initialized deterministic arenas or the class of all initialized stochastic arenas.

\begin{thm}[Pure AIFM one-to-two-player lift]\label{thm:1to2}
	Let $\pref$ be a preference relation, $\memSkel_1$ and $\memSkel_2$ be two memory skeletons, and $\stratType\in\allStrats$ be a type of strategies.
	Let $\arenaFam$ be a class of initialized arenas that is closed by subarena, by split, and by product with $\memSkel_1$ and $\memSkel_2$.
	Assume that
	\begin{itemize}
		\item in all initialized one-player arenas of $\Pone$ in $\arenaFam$, $\Pone$ can play $\stratType$-optimally with a pure strategy based on memory $\memSkel_1$;
		\item in all initialized one-player arenas of $\Ptwo$ in $\arenaFam$, $\Ptwo$ can play $\stratType$-optimally with a pure strategy based on memory $\memSkel_2$.
	\end{itemize}
	Then all initialized two-player arenas in $\arenaFam$ admit a pure $\stratType$-NE based on memory $\memSkel_1 \memProduct \memSkel_2$.
	If $\arenaFam$ is moreover closed by prefix-extension, then all initialized two-player arenas in $\arenaFam$ admit a pure $\stratType$-SPE based on memory $\memSkel_1 \memProduct \memSkel_2$.
\end{thm}

The practical usage of this result can be summed up as follows: to determine whether pure AIFM strategies are sufficient for both players in stochastic (resp.\ deterministic) arenas to play $\stratType$-optimally, it is sufficient to prove it for stochastic (resp.\ deterministic) one-player arenas.
Studying memory requirements of one-player arenas is significantly easier than studying memory requirements of two-player arenas, as a one-player arena can be seen as a graph (in the deterministic case) or an MDP (in the stochastic case).
Still, we will bring more tools to study memory requirements of one-player arenas in Section~\ref{sec:1p}.

Our proof technique for Theorem~\ref{thm:1to2} is able to deal in a uniform manner with stochastic arenas and with deterministic arenas, under different types of strategies.
It borrows ideas from~\cite{GZ09} and from~\cite{BLORV22} and extends them both: it extends~\cite{GZ09} by generalizing to a wider type of strategies (AIFM instead of memoryless) and it extends~\cite{BLORV22} by extending the class of arenas and preference relations considered (stochastic instead of deterministic).
Thanks to Theorem~\ref{prop:optAIFMimpliesSP}, we also go further in our understanding of the optimal strategies: we are able to obtain the existence of $\stratType$-SPE with almost the same constraints, instead of the seemingly weaker existence of $\stratType$-NE\@.

\begin{rem}
	As discussed previously, Theorem~\ref{thm:1to2}, which deals with AIFM strategies, was known for memoryless strategies, i.e., with $\memSkel_1 = \memSkel_2 = \memSkelTriv$~\cite[Theorem~9]{GZ09}.
	Observe that by using Lemma~\ref{lem:memProd}, we can reduce Theorem~\ref{thm:1to2} to a result dealing only with \emph{memoryless} strategies, both in the hypothesis and the conclusion.
	This observation is insufficient to immediately derive the result about AIFM strategies from the result about memoryless strategies, as the hypothesis is then about memoryless strategies in the class of \emph{product arenas}.
	Without additional technical changes (discussed below) to~\cite[Theorem~9]{GZ09}, we cannot directly apply it to this class of product arenas.
	\qedEx
\end{rem}

\subsection{Proof scheme}
In order to prove Theorem~\ref{thm:1to2}, we first establish a similar result about memoryless strategies.
We carry out all our intermediate proofs with the concept of Nash equilibrium, and we will strengthen it to subgame perfect equilibria at the end, thanks to Theorem~\ref{prop:optAIFMimpliesSP}.

\begin{lem}[Memoryless one-to-two-player lift]\label{thm:GZ16Gen}
	Let $\pref$ be a preference relation and $\stratType\in\{\PFM, \pure, \RFM, \rand\}$ be a type of strategies.
	Let $\arenaFam$ be a class of initialized arenas that is closed by subarena and by split.
	If both players have pure memoryless $\stratType$-optimal strategies in the initialized one-player arenas in~$\arenaFam$, then all initialized arenas in $\arenaFam$ admit a pure memoryless $\stratType$-NE\@.
\end{lem}

This result and its proof are very similar to~\cite[Theorem~9]{GZ09}.
It applies in a generic way to various classes of arenas (mostly, deterministic arenas or stochastic arenas).
This is an advantage compared to the proofs of~\cite{GZ05,BLORV22}, that were both strongly coupled with the monotony and selectivity notions, that are (in the form stated in these papers) only suited to deal with deterministic games.
An important addition to~\cite[Theorem~9]{GZ09} is that we consider here \emph{initialized} arenas: the strategies do not have to be optimal from all states, but only from the specified initial states.
We explain why, albeit small, this is an important addition to obtain our result.

Lemma~\ref{thm:GZ16Gen} can immediately be instantiated with $\arenaFam$ being the class of all deterministic or stochastic arenas to obtain an interesting result about pure memoryless strategies.
Our goal will be to instantiate it, for some fixed memory skeleton $\memSkel$, with the class of \emph{initialized arenas covered by $\memSkel$}, so that we can later obtain results about pure AIFM strategies (through Lemma~\ref{prop:optProdIsCov}) instead of only using it for pure memoryless strategies, in a similar spirit to~\cite{BLORV22}.
As we will see, the class of \emph{initialized} arenas covered by $\memSkel$ happens to be closed by subarena and by split.

This extension to pure AIFM strategies is one precise step where the notion of \emph{initialized arenas} finds its use: in covered (or product) arenas, we are only interested in optimality from arena states associated to memory state $\memInit$, and not from all states.

Without making Lemma~\ref{thm:GZ16Gen} about \emph{initialized} arenas (i.e., as in~\cite[Theorem~9]{GZ09}), a natural candidate to extend it to AIFM strategies would be to consider the class of all ``product arenas'' (with no distinction of initial states) in its statement, through their connection with memoryless strategies from Lemma~\ref{lem:memProd}.
However, such product arenas are not closed by subarena: if we remove a transition from one of them, it is not possible in general to realize it as a product with a smaller arena (unlike what happens if we distinguish initial states, as in Lemma~\ref{lem:prodCov}).
Therefore, the inductive proof technique cannot be performed directly on such product arenas.
On the other hand, covered arenas, even with no distinguished initial states, are closed by subarena, but all of them are not obtained by product arenas, so Lemma~\ref{lem:memProd} cannot be used straight away.

Without restating and reproving~\cite[Theorem~9]{GZ09} with an extra quantification on the initial states, it seems difficult to extend it straightforwardly to a result about pure AIFM strategies.

\subsection{Proving Lemma~\ref{thm:GZ16Gen}}
We recall the definitions of subarena and split from~\cite{GZ09}, extending them in a natural way to initialized arenas.

\begin{defi}[Initialized subarena]
	Let $(\arena = \arenaFull, \initStates)$ be an initialized arena.
	An \emph{initialized subarena of $(\arena, \initStates)$} is an initialized arena $((\states_1, \states_2, \act', \transProb', \colSolo'), \initStates)$ such that $\act' \subseteq \act$, $\transProb'\subseteq \transProb$ (that is, some states might lose a few available actions), and $\colSolo'$ is the restriction of $\colSolo$ to the pairs $(\st, \action)$ such that $\st\in\states$ and $\action\in\act'(\st)$.
\end{defi}

An initialized subarena keeps the same state space and initial states as the original arena, but with fewer available actions.
Remember that we assume that arenas are non-blocking, hence at least one available action should be kept in each state of the initialized subarena.
We say that a class $\arenaFam$ of initialized arenas is \emph{closed by subarena} if for all $\initArena\in\arenaFam$, if $(\arena', \initStates)$ is a subarena of $\initArena$, then $(\arena', \initStates)\in\arenaFam$.

We now define the notion of \emph{split on $t\in\states$} of an arena: the main idea is that for some state $t$ of the arena, the state space of the arena is augmented in such a way that players remember what was the last action played when leaving state $t$.
In practice, for each action $\action$ available in $t$, we first make a copy $\arena_\action$ of the arena where only $\action$ is available to play in $t$, and we rename all states $\st \mapsto \st^\action$.
Then, we merge all states $t^\action$ of the arenas $(\arena_\action)_{\action\in\act(t)}$, and we rename the resulting state $t$ after the merge.
All the other states stay as they were in the arenas $\arena_\action$.
Every action $\action$ available in $t$ therefore leads to a copy of the arena in which states are labeled by $\action$.
Before introducing the formal definition, we provide an example of a split in Figure~\ref{fig:split}.
There are two actions $a$ and $b$ available in $t$, and we make a copy of the other states for each action.
This way, when the game is for instance in $s^a$, we know that the last action that was chosen in $t$ was $a$ (which is not necessarily the case when in $s$ in the original arena).
The probabilities, colors and initial states are preserved in each copy of the initialized arena.

\begin{figure}[tbh]
	\centering
	\begin{minipage}{0.45\columnwidth}
	\centering
	\begin{tikzpicture}[every node/.style={font=\small,inner sep=1pt}]
		\draw (0,0) node[rond] (t) {$t$};
		\draw ($(t)+(1,0.5)$) node[rond,fill=black,minimum size=3pt] (ta) {};
		\draw ($(t)+(2.5,1)$) node[rond] (r) {$r$};
		\draw ($(t)+(2.5,-1)$) node[carre] (s) {$s$};
		\draw ($(t)-(0.75,0)$) edge[-latex'] (t);
		\draw ($(s)+(0.75,0)$) edge[-latex'] (s);

		\draw (t) edge[-latex'] node[above=2pt,xshift=-1pt] {$a$} (ta);
		\draw (ta) edge[-latex'] node[above=2pt,xshift=-1pt] {$\frac{1}{2}$} (r);
		\draw (ta) edge[-latex'] node[above=2pt,xshift=1pt] {$\frac{1}{2}$} (s);
		\draw (t) edge[-latex'] node[above=2pt] {$b$} (s);

		\draw (r) edge[-latex',out=30,in=-30,distance=15pt] (r);
		\draw (s) edge[-latex'] node[above=2pt] {} (r);
		\draw (s) edge[-latex',bend left=20] node[above=2pt] {} (t);
	\end{tikzpicture}
	\end{minipage}
	\begin{minipage}{0.45\columnwidth}
	\centering
	\begin{tikzpicture}[every node/.style={font=\small,inner sep=1pt}]
		\draw (0,0) node[rond] (t) {$t$};
		\draw ($(t)+(1,0.5)$) node[rond,fill=black,minimum size=3pt] (ta) {};
		\draw ($(t)+(2.5,1)$) node[rond] (ra) {$r^a$};
		\draw ($(t)+(2.5,-1)$) node[carre] (sa) {$s^a$};
		\draw ($(t)+(-2.5,1)$) node[rond] (rb) {$r^b$};
		\draw ($(t)+(-2.5,-1)$) node[carre] (sb) {$s^b$};
		\draw ($(t)-(0.75,0)$) edge[-latex'] (t);
		\draw ($(sa)+(0.75,0)$) edge[-latex'] (sa);
		\draw ($(sb)-(0.75,0)$) edge[-latex'] (sb);

		\draw (t) edge[-latex'] node[above=2pt,xshift=-1pt] {$a$} (ta);
		\draw (ta) edge[-latex'] node[above=2pt,xshift=-1pt] {$\frac{1}{2}$} (ra);
		\draw (ta) edge[-latex'] node[above=2pt,xshift=1pt] {$\frac{1}{2}$} (sa);
		\draw (t) edge[-latex'] node[above=2pt] {$b$} (sb);

		\draw (ra) edge[-latex',out=30,in=-30,distance=15pt] (ra);
		\draw (rb) edge[-latex',out=150,in=-150,distance=15pt] (rb);
		\draw (sa) edge[-latex'] node[above=2pt] {} (ra);
		\draw (sa) edge[-latex',bend left=20] node[above=2pt] {} (t);
		\draw (sb) edge[-latex'] node[above=2pt] {} (rb);
		\draw (sb) edge[-latex',bend left=-20] node[above=2pt] {} (t);
	\end{tikzpicture}
	\end{minipage}
	\caption{Initialized arena with $\actions{t} = \{a,b\}$ (omitting colors) (left) and its split on $t$ (right).
	States controlled by $\Pone$ (resp.\ $\Ptwo$) are depicted by circles (resp.\ squares).
	The dot after playing action $a$ represents a stochastic transition, with probability $\frac{1}{2}$ to go to $r$ and $\frac{1}{2}$ to go to $\st$.}%
	\label{fig:split}
\end{figure}

\begin{defi}[Split]\label{def:split}
	Let $(\arena = \arenaFull, \initStates)$ be an initialized arena, and $t\in\states_1$.
	For $\action\in\actions{t}$, we denote $(\arena_\action, \initStates^\action)$ as the initialized subarena of $(\arena, \initStates)$ in which only action $\action$ is available in $t$ and in which all states are renamed $\st\mapsto\st^\action$.
	The \emph{split on $t$ of $(\arena, \initStates)$} is the initialized arena $(\splitArena{\arena}{t}, \splitInitStates{t})$ where
	\[
		\splitArena{\arena}{t}
		= \bigArenaJoin_{\action\in\actions{t}} (\arena_a, t^a),
	\]
	with merged state called $t$, and $\splitInitStates{t} = \bigcup_{\action\in\actions{t}} \{\st^a\mid \st\in\initStates\}$ (with $t^\action = t$ for all $\action\in\actions{t}$).
\end{defi}

This could be defined symmetrically for a state $t\in\states_2$.
The merge operator $\arenaJoin$ was introduced at~\eqref{eq:merge} on page~\pageref{eq:merge}.
When considering a split arena on a state $t$, we use the convention that $t^\action = t$ for any action $\action$ available in $t$.
Moreover, for $\action$ an action available in $t$, and $\states'\subseteq \states$, we write $(\states')^\action$ for $\{\st^\action \mid \st\in\states'\}$.
A class $\arenaFam$ of initialized arenas is \emph{closed by split} if for all $\initArena\in\arenaFam$, for all states $t$ of $\arena$, $(\splitArena{\arena}{t}, \splitInitStates{t})\in\arenaFam$.

It is possible to formulate a few intuitive results linking plays and strategies of an initialized arena and its split.
These results are provided in a very similar context in~\cite{GZ09}; we recall three results precisely in Appendix~\ref{sec:split} and sketch their statements here.
\begin{itemize}
	\item \underline{Lemma~\ref{lem:bijSplit}}.
Let $(\arena = \arenaFull, \initStates)$ be an initialized arena with a state $t\in\states_1$.
For all $\action\in\actions{t}$, it is possible to build a natural bijection between strategies in $\stratsType{i}{\rand}{\arena, \initStates}$ and strategies on the split with restricted initial states $\stratsType{i}{\rand}{\splitArena{\arena}{t}, \initStates^\action}$.
Intuitively, the available actions are the same at every step and the split does not offer any more possibilities (besides having more initial states --- that is why we must restrict the initial states to have a bijection).
More memory might be needed to play the corresponding strategy in $\initArena$ than in its split (since the information of the last action played in $t$ is not explicitly given in $\arena$), but finite-memory strategies stay finite-memory in both directions.
This bijection also preserves the ``pure'' feature of the strategy and preserves optimality and NE\@.

\item \underline{Lemma~\ref{lem:splitProjection}}. If a pair of strategies $(\strat_1, \strat_2)$ is an $\stratType$-NE in $(\splitArena{\arena}{t}, \initStates^\action)$ and $\strat_1$ is pure memoryless with $\strat_1(t) = \action$, only the part $\states^\action$ of the split is ever reached during the play, which corresponds to the state space of $\arena$.
It is therefore possible to transform this $\stratType$-NE into an $\stratType$-NE $(\strat_1', \strat_2')$ in $\initArena$ such that $\strat_1'$ is still pure and memoryless.
Note that extending this lemma to strategies that are not only pure seems difficult: if the memoryless choice of $\strat_1$ is not deterministic in $t$, then multiple parts of the split arena may still be reached.

\item \underline{Lemma~\ref{lem:stratSplit}}.
If a pair of strategies $(\strat_1, \strat_2)$ on $(\splitArena{\arena}{t}, \initStates^\action)$ is such that $\strat_1$ is pure memoryless with $\strat_1(t) = \action$, then for $\st^\action\in\initStates^\action$, we can show that the distribution $\prcSolo{\splitArena{\arena}{t}}{\st^\action}{\strat_1, \strat_2}$ is equal to $\prcSolo{\arena_\action}{\st^\action}{\strat_1^\action, \strat_2^\action}$, where $\strat_1^\action$ and $\strat_2^\action$ are simply restrictions of $\strat_1$ and $\strat_2$ to histories of the subarena $(\arena_\action, \initStates^\action)$ (Lemma~\ref{lem:stratSplit}).
\end{itemize}

\noindent
We are now ready to prove Lemma~\ref{thm:GZ16Gen}.
The proof is by induction on the number of \emph{choices} in arenas:
for $\arena = \arenaFull$ an arena, the \emph{number of choices in $\arena$} is defined as
\begin{equation}\label{eq:choices}
\size{\arena} = (\sum_{\st\in\states} \card{\actions{\st}}) - \card{\states}.
\end{equation}
When the number of choices in $\arena$ is $0$, it means that there is exactly one available action in each state.
\begin{proof}[Proof of Lemma~\ref{thm:GZ16Gen}]
	We proceed by induction on the number of choices in arenas.
	If an initialized arena $(\arena, \initStates)\in\arenaFam$ is such that $\size{\arena} = 0$, then both players only have a single available strategy which is both pure and memoryless.
	Hence this pair of strategies correspond to a pure memoryless $\stratType$-NE, which proves the base case.
	Now let $n > 0$: we assume that the result holds for every initialized arena $\initArena\in\arenaFam$ with $\size{\arena} < n$, and let $(\arena, \initStates)\in\arenaFam$ be an initialized arena such that $\size{\arena} = n$.

	If $\Pone$ has no choice (that is, $\card{\actions{\st}} = 1$ for all $\st\in\states_1$), then $(\arena, \initStates)$ is an initialized one-player arena of $\Ptwo$, and $\Ptwo$ has a pure memoryless $\stratType$-optimal strategy by hypothesis from the statement of the theorem (and as $\Pone$ has only one possible strategy which happens to be pure and memoryless, we have a pure memoryless $\stratType$-NE).
	We now focus on the case where $\Pone$ has at least one choice:
	let $t \in \states$ be such that $\card{\actions{t}} \ge 2$.
	For $\action\in\actions{t}$, let $(\arena_\action, \initStates^\action)$ be the initialized subarena of $(\arena, \initStates)$ with only action $\action$ available in $t$, and all states renamed $\st\mapsto\st^\action$.
	This implies that $(\arena_\action, \initStates^\action)$ is in $\arenaFam$, as $\arenaFam$ is closed by subarena.
	Notice that $\size{\arena_\action} < \size{\arena}$, as this is the same arena except that some available actions are removed in $t$.
	By induction hypothesis, there is thus a pure memoryless $\stratType$-NE $(\strat_1^\action, \strat_2^\action)$ in arena $(\arena_\action, \initStates^\action)$, for each $\action\in\actions{t}$.

	We now consider the split on $t$ of $(\arena, \initStates)$, which we denote $(\splitArena{\arena}{t}, \splitInitStates{t})$, and which belongs to $\arenaFam$ as $\arenaFam$ is closed by split.
	Consider the pure memoryless strategy $\strat_2
	= \bigcup_{\action\in\actions{t}} \strat_2^\action$
	of $\Ptwo$ defined on $(\splitArena{\arena}{t}, \splitInitStates{t})$ which, when in $\states^\action$ for some $\action\in\actions{t}$, plays the same actions as the pure memoryless strategy $\strat_2^\action$ (this prescribes a unique action to every state of $\splitArena{\arena}{t}$ controlled by $\Ptwo$; the only overlap between the subarenas is $t$, but $t$ belongs to $\Pone$).
	We cannot straightaway define a similar strategy $\strat_1$ for $\Pone$, as it would not be well-defined in $t$ --- we first have to carefully choose the action played in $t$.

	To do so, we consider the initialized one-player arena $(\arenaSplitFixedStrat, \splitInitStates{t})$ of $\Pone$ resulting from fixing the (pure memoryless) strategy $\strat_2$ of $\Ptwo$ in $\splitArena{\arena}{t}$.
	This initialized arena is an initialized subarena of $(\splitArena{\arena}{t}, \splitInitStates{t})$ (actions of $\Ptwo$ have been removed), and hence belongs to $\arenaFam$.
	By hypothesis from the statement of the theorem, $\Pone$ has a pure memoryless $\stratType$-optimal strategy $\stratBis_1$ in $(\arenaSplitFixedStrat, \splitInitStates{t})$ since it is an initialized one-player arena.
	Let $\action^*\in\actions{t}$ be the action $\stratBis_1(t)$ played by $\Pone$ in $t$.
	We use this strategy $\stratBis_1$ to define a pure memoryless strategy $\strat_1$ of $\Pone$ on $(\splitArena{\arena}{t}, \splitInitStates{t})$ that plays $\action^*$ in $t$, and for all $\action\in\actions{t}$, that behaves like $\strat_1^\action$ in $\states^\action$; formally,
	\begin{align*}
		\strat_1(t) = \tau_1(t) = \action^*, \text{ and for } \action \in \actions{t}\text{, } \st\in \states_1\setminus\{t\}\text{, }
		\strat_1(\st^\action) = \strat_1^\action(\st^\action).
	\end{align*}

	We now show that $(\strat_1, \strat_2)$ is a pure memoryless $\stratType$-NE in $(\splitArena{\arena}{t}, \initStates^{\action^*})$.
	We restrict our attention to initial states $\initStates^{\action^*}$ as these states are in the part of the arena that $\strat_1$ always goes back to, and that will help us convert $\strat_1$ into a corresponding pure memoryless strategy on $\initArena$ using Lemma~\ref{lem:splitProjection}.
	We show that for all $\st^{\action^*} \in \initStates^{\action^*}$, for all $\strat_1'\in\stratsType{1}{\stratType}{\splitArena{\arena}{t}, \initStates^{\action^*}}$, $\strat_2'\in\stratsType{2}{\stratType}{\splitArena{\arena}{t}, \initStates^{\action^*}}$,
	\[
		\prcSolo{\splitArena{\arena}{t}}{\st^{\action^*}}{\sigma_1', \sigma_2} \pref
		\prcSolo{\splitArena{\arena}{t}}{\st^{\action^*}}{\sigma_1, \sigma_2} \pref
		\prcSolo{\splitArena{\arena}{t}}{\st^{\action^*}}{\sigma_1, \sigma_2'}.
	\]
	Let $\st^{\action^*} \in \initStates^{\action^*}$.
	The right-hand side inequality is clear, as the play always stays in $\arena_{\action^*}$ by definition of $\strat_1$, and the restriction of $\strat_1$ and $\strat_2$ to $\Hists(\arena_{\action^*}, \initStates^{\action^*})$ is $(\strat_1^{\action^*}, \strat_2^{\action^*})$, which is an $\stratType$-NE in $(\arena_{\action^*}, \initStates^{\action^*})$.

	For the left-hand side inequality, let $\strat_1'\in\stratsType{1}{\stratType}{\splitArena{\arena}{t}, \initStates^{\action^*}}$ be a strategy of $\Pone$.
	We denote by $\stratBis_1^{\action^*}$ the restriction of $\stratBis_1$ to $\Hists(\arena_{\action^*}, \initStates^{\action^*})$.

	We have
	\begin{align*}
		\prcSolo{\splitArena{\arena}{t}}{\st^{\action^*}}{\sigma_1', \sigma_2}
		&\pref \prcSolo{\splitArena{\arena}{t}}{\st^{\action^*}}{\stratBis_1, \strat_2}
		&&\text{as $\stratBis_1$ is $\stratType$-optimal against $\strat_2$ in $(\splitArena{\arena}{t}, \splitInitStates{t})$,}\\
		& &&\text{and $\st^{\action^*}\in\splitInitStates{t}$} \\
		&= \prcSolo{\arena_{\action^*}}{\st^{\action^*}}{\stratBis_1^{\action^*}, \strat_2^{\action^*}}
		&&\text{by Lemma~\ref{lem:stratSplit}}\\
		&\pref \prcSolo{\arena_{\action^*}}{\st^{\action^*}}{\strat_1^{\action^*}, \strat_2^{\action^*}}
		&&\text{as $(\strat_1^{\action^*}, \strat_2^{\action^*})$ is an $\stratType$-NE in $\arena_{\action^*}$}\\
		&= \prcSolo{\splitArena{\arena}{t}}{\st^{\action^*}}{\strat_1,\strat_2}
		&&\text{by Lemma~\ref{lem:stratSplit}}.
	\end{align*}

	Hence $(\strat_1, \strat_2)$ is a pure memoryless $\stratType$-NE in $(\splitArena{\arena}{t}, \initStates^{\action^*})$.
	We can transform $(\strat_1, \strat_2)$ into an $\stratType$-NE $(\strat_1', \strat_2')$ of $(\arena, \initStates)$ by Lemma~\ref{lem:splitProjection}, with $\strat_1'$ pure memoryless, but not necessarily $\strat'_2$.%
	\footnote{Lemma~\ref{lem:splitProjection} uses strongly the assumption that strategies are pure. It is unclear how to generalize the proof to randomized strategies because of this argument.} 
	We can however perform the same proof for $\Ptwo$, and obtain a second $\stratType$-NE $(\strat_1'', \strat_2'')$ such that $\strat_2''$ is a pure memoryless strategy on $(\arena, \initStates)$.
	Then we simply mix both $\stratType$-NE (by Lemma~\ref{lem:mixNE}), and obtain that $(\strat_1', \strat_2'')$ is a pure memoryless $\stratType$-NE in~$(\arena, \initStates)$.
\end{proof}

\subsection{From memoryless to AIFM}

In this section, we show how to apply Lemma~\ref{prop:optProdIsCov} to Lemma~\ref{thm:GZ16Gen} to lift its results from the sufficiency of pure memoryless strategies to the sufficiency of pure AIFM strategies, which will imply Theorem~\ref{thm:1to2}.
\begin{lem}\label{lem:covAreClosed}
	Let $\memSkel$ be a memory skeleton, and $\arenaFam$ be a class of initialized arenas closed by subarena and by split.
	The class of all initialized arenas covered by $\memSkel$ in $\arenaFam$ is closed by subarena and by split.
\end{lem}
\begin{proof}
	Let $(\arena = \arenaFull, \initStates)\in\arenaFam$ be an initialized arena covered by $\memSkel = \memSkelFull$.
	We show that its subarenas and its splits are still covered by $\memSkel$.
	Let $\coverFunction\colon \states \to \memStates$ be the witness that $(\arena, \initStates)$ is covered by $\memSkel$.

	If we consider an initialized subarena $(\arena', \initStates)$ of $(\arena, \initStates)$, the same function $\coverFunction$ will still be a witness that $(\arena', \initStates)$ is covered by $\memSkel$, as the state space is the same, the condition to check is a universally quantified property over the transitions, and there are fewer transitions in $\arena'$ than in $\arena$.

	Let $t\in\states$, and let $(\splitArena{\arena}{t}, \splitInitStates{t})$ be the split on $t$ of $(\arena, \initStates)$.
	We define a function $\widehat\coverFunction$ such that
	\[
		\widehat\coverFunction(t) = \coverFunction(t),
		\text{ and for } \action\in\actions{t}\text{, } \st\in\states\setminus\{t\}\text{, }
		\widehat\coverFunction(\st^\action) = \coverFunction(\st).
	\]
	Function $\widehat\coverFunction$ witnesses that $(\splitArena{\arena}{t}, \splitInitStates{t})$ is covered by $\memSkel$ as every transition in $\splitArena{\arena}{t}$ corresponds to a transition in $\arena$ with the same color, and linking two states assigned to the same memory state in $\arena$ as in $\splitArena{\arena}{t}$.
\end{proof}

We are now ready to prove Theorem~\ref{thm:1to2}, the main result of this section.

\begin{proof}[Proof of Theorem~\ref{thm:1to2}]
	Note first that as $\prodAS{(\prodAS{(\arena, \initStates)}{\memSkel_1})}{\memSkel_2}$ is isomorphic to $\prodAS{(\arena, \initStates)}{(\memSkel_1 \memProduct \memSkel_2)}$, $\arenaFam$ is in particular closed by product with $\memSkel_1 \memProduct \memSkel_2$.

	Using Lemma~\ref{prop:optProdIsCov}, the hypotheses can be reformulated as follows: for $i\in\{1, 2\}$, $\player{i}$ has a pure memoryless $\stratType$-optimal strategy in all initialized one-player arenas in $\arenaFam$ that are covered by $\memSkel_i$.

	Now consider the subclass $\arenaFam' = \{(\arena, \initStates)\in\arenaFam \mid (\arena, \initStates)\text{ is covered by } \memSkel_1 \memProduct \memSkel_2\}$.
	For $i\in\{1, 2\}$, $\player{i}$ has a pure memoryless $\stratType$-optimal strategy in all its initialized one-player arenas in $\arenaFam'$ (using that if an arena is covered by $\memSkel_1\memProduct\memSkel_2$, it is in particular covered by $\memSkel_i$ by Lemma~\ref{lem:M1M2}).
	Moreover, $\arenaFam'$ is closed by subarena and by split by Lemma~\ref{lem:covAreClosed}.
	Hence by Lemma~\ref{thm:GZ16Gen}, for all initialized arenas in $\arenaFam'$, there exists a pure memoryless $\stratType$-NE\@.
	Using Lemma~\ref{prop:optProdIsCov} again in the other direction allows us to conclude that all initialized arenas in $\arenaFam$ admit a pure $\stratType$-NE based on $\memSkel_1\memProduct\memSkel_2$.

	By Theorem~\ref{prop:optAIFMimpliesSP}, using that $\arenaFam$ is closed by prefix-extension, the existence of pure $\stratType$-NE based on $\memSkel_1\memProduct\memSkel_2$ in all initialized arenas in $\arenaFam$ implies the existence of pure $\stratType$-SPE based on $\memSkel_1\memProduct\memSkel_2$ in all initialized arenas in $\arenaFam$.
\end{proof}

Theorem~\ref{thm:1to2} along with Lemmas~\ref{lem:memProd} and~\ref{lem:prodCov} actually gives a bit more information about memory requirements in individual arenas than is strictly written.
The way it is phrased shows that for an initialized arena $\initArena$, memoryless strategies always suffice to play $\stratType$-optimally in $\prodAS{\initArena}{(\memSkel_1\memProduct\memSkel_2)}$ (through Lemma~\ref{lem:memProd}).
But if $\initArena$ is already covered by $\memSkel_1\memProduct\memSkel_2$, then as $\initArena$ is isomorphic to $\prodAS{\initArena}{(\memSkel_1\memProduct\memSkel_2)}$ by Lemma~\ref{lem:prodCov}, memoryless strategies are actually sufficient directly in $\initArena$.
Memory $\memSkel_1\memProduct\memSkel_2$ is thus an upper bound on the required memory, and studying coverability may show for some initialized arenas that less memory is sufficient: for example, if it is already covered by $\memSkel_1\memProduct\memSkel_2$ (resp.\ by $\memSkel_1$, by $\memSkel_2$), memoryless strategies (resp.\ strategies based on $\memSkel_2$, on $\memSkel_1$) are sufficient.
An application of Theorem~\ref{thm:1to2} is provided in Section~\ref{sec:app}.

\subsection{Discussing the use of randomization}\label{sec:noLiftRandomized}
All our results can help prove that \emph{pure} AIFM strategies suffice to play optimally.
We show that a one-to-two-player lift similar to Theorem~\ref{thm:1to2} does not hold if we allow (unconstrained) randomization in the strategies.
We exhibit an objective and a class of arenas for which (non-necessarily pure) memoryless strategies suffice in one-player arenas for both players, but not in two-player arenas.

Let $\colors = \{-1, 1\}$.
For an infinite word $\word = \clr_1\clr_2\ldots\in\colors^\omega$, we define
\[
\MP(\word) = \liminf_{n\to\infty} \frac{1}{n} \sum_{i=1}^n \clr_i
\]
as the \emph{mean payoff of $\word$}.
We define the winning condition
\[
\wc = \{\word\in\colors^\omega \mid \MP(\word) = 0\}.
\]
As described in Example~\ref{ex:objExamples}, we consider the preference relation $\pref_\wc$ induced by $\wc$ (i.e., $\Pone$ wants to guarantee the  greatest possible probability that $\wc$ happens).

We consider for this example the class of \emph{deterministic} arenas.
We show the three following facts about $\pref_\wc$:
\begin{itemize}
	\item (non-necessarily pure) memoryless strategies suffice for $\Pone$ in its one-player deterministic arenas;
	\item pure memoryless strategies suffice for $\Ptwo$ in its one-player deterministic arenas;
	\item memoryless strategies (even with randomization) do not suffice for $\Pone$ in two-player deterministic arenas.
\end{itemize}

\noindent
For $\arena = \arenaFull$ a deterministic arena, we define a \emph{simple cycle of $\arena$} as a history $\st_0\action_1\st_1\ldots\action_n\st_n$ such that $\st_0 = \st_n$, and for all $i, j\in\IN$ with $0 \le i < j \le n-1$, $\st_i \neq \st_j$.

\paragraph{One-player arenas of $\Pone$}
Let $\arena = \arenaFull$ be a one-player deterministic arena of $\Pone$.
We can think of $\arena$ as a directed graph with labeled transitions.
We establish whether $\Pone$ can guarantee $\wc$ with probability $1$ or not depending on the mean payoffs of simple cycles in each strongly connected component (SCC) of $\arena$.

Let $\initStates^\mathsf{lose}\subseteq \states$ be the states from which in every reachable SCC of $\arena$, the simple cycles either all have a mean payoff $< 0$ or all have a mean payoff $> 0$.
We show that $\Pone$ cannot obtain $\wc$ with a positive probability.
Indeed, a play $\play\in\Plays(\arena, \initStates^\mathsf{lose})$ will eventually end up in some reachable SCC in which all simple cycles have a mean payoff $< 0$ (resp.\ $> 0$).
In such an SCC, the highest (resp.\ lowest) mean payoff that $\Pone$ can obtain is given by the simple cycle with the highest (resp.\ lowest) mean payoff (this argument follows from the memoryless determinacy of mean-payoff games~\cite{EM79}).
The mean payoff of $\colHat{\play}$, if it exists, can then only be $< 0$ (resp.\ $> 0$).
Any strategy therefore achieves $\wc$ from $\initStates^\mathsf{lose}$ with probability $0$, which is optimal.

Now, we consider the set of states $\initStates^\mathsf{win} = \states\setminus\initStates^\mathsf{lose}$ from which there is a reachable SCC containing a simple cycle $\hist_1 = \st_0^1\action_1^1\st_1^1\ldots\action_n^1\st_n^1$ with $\MP(\hist_1) < 0$ and a simple cycle $\hist_2 = \st_0^2\action_1^2\st_1^2\ldots\action_m^2\st_m^2$ with $\MP(\hist_2) \ge 0$ (our argument can easily be adapted for the case $\MP(\hist_1) \le 0$ and $\MP(\hist_2) > 0$).
Let $\states_{\hist_1} = \{\st_0^1, \ldots, \st_{n-1}^1\}$ and $\states_{\hist_2} = \{\st_0^2, \ldots, \st_{m-1}^2\}$ be the states visited respectively by $\hist_1$ and $\hist_2$.

We show that we can assume w.l.o.g.\ that $\hist_1$ and $\hist_2$ have at least one common state, i.e., that $\states_{\hist_1} \cap \states_{\hist_2} \neq \emptyset$.
If not, this means that there is a simple cycle $\hist'$ that shares a state with $\hist_1$ and $\hist_2$ (as $\hist_1$ and $\hist_2$ are in the same SCC).
If $\MP(\hist') \ge 0$, we then replace $\hist_2$ by $\hist'$; if $\MP(\hist') < 0$, we then replace $\hist_1$ by $\hist'$.

From this, we obtain that there is a pure finite-memory strategy that achieves a mean payoff of exactly $0$, simply by alternating between $\hist_1$ and $\hist_2$ at the right frequency.
We end the proof by using~\cite[Lemma~15]{CRR14}, which shows that for \emph{one-player multi mean-payoff games}, pure finite-memory winning strategies can be traded for randomized memoryless strategies with rational randomization that win with probability~$1$.
To do so, we observe that our objective can be reduced to a special case of multi mean-payoff games with two dimensions: if we replace color $1$ with $(1, -1)$ and color $-1$ with $(-1, 1)$, a winning play for $\wc$ is exactly a winning play for the objective consisting of obtaining a mean payoff $\ge 0$ along both dimensions.

To sum up, we have built a (non-pure) memoryless strategy that wins with probability~$0$ from $\initStates^\mathsf{lose}$ and with probability $1$ from $\initStates^\mathsf{win}$, which is optimal in both cases.

\paragraph{One-player arenas of $\Ptwo$}
We now adopt the point of view of $\Ptwo$.
Let $\arena = \arenaFull$ be a one-player deterministic arena of $\Ptwo$.
We establish once again whether $\Ptwo$ can achieve its goal with probability $1$ depending on the values of the mean payoffs of simple cycles.

Let $\initStates^\mathsf{win}\subseteq \states$ be the states from which there is a reachable simple cycle $\hist$ with $\MP(\colHat{\hist}) \neq 0$.
From $\initStates^\mathsf{win}$, a pure memoryless strategy ensuring a win for $\Ptwo$ consists in reaching this simple cycle and looping around it forever.

Let $\initStates^\mathsf{lose} = \states\setminus\initStates^\mathsf{win}$ be the states from which all reachable simple cycles have a mean payoff of $0$.
From such states, we show that for every strategy of $\Ptwo$, $\wc$ happens with probability $1$.
For every finite word $\word = \clr_1\ldots\clr_n\in\colHat{\Hists(\arena, \initStates^\mathsf{lose})}$ that $\Ptwo$ can generate from $\initStates^\mathsf{lose}$, we have
\[
-(\card{\states} - 1) \le \sum_{i = 1}^n \clr_i \le \card{\states} - 1.
\]
Indeed, the sum of colors of any cycle appearing in $\hist$ is $0$, and after removing cycles in $\hist$ until none are left (in any order), at most $\card{\states} - 1$ transitions remain.
Let $\word = \clr_1\clr_2\ldots \in \colHat{\Plays(\arena, \initStates^\mathsf{lose})}$ be an infinite word that can be generated by $\Ptwo$ from $\initStates^\mathsf{lose}$.
By the previous inequalities, we have
\[
0 = \lim_{n\to\infty} \frac{1}{n} \cdot (-(\card{\states} - 1)) \le \MP(\word) \le \lim_{n\to\infty} \frac{1}{n} \cdot (\card{\states} - 1) = 0.
\]
Hence, $\MP(\word)$ equals $0$.

\paragraph{Insufficiency of memoryless strategies in two-player arenas}
We consider the arena in Figure~\ref{fig:MP} (close examples were considered in~\cite[Proposition~4.9]{KopThesis} and in~\cite[Lemma~15]{CRR14}.).
In this arena, $\Pone$ has a pure strategy using two states of memory that guarantees $\wc$ with probability $1$: whenever $\Ptwo$ plays $-1$ (resp.\ $1$), $\Pone$ responds with $1$ (resp.~$-1$).
This requires two memory states and ensures that the mean color seen is $0$.
However, if~$\Pone$ uses a memoryless strategy (even with randomization), $\Ptwo$ can ensure that the resulting mean payoff is different from $0$.
Indeed, if the distribution chosen by $\Pone$ picks $-1$ (resp. $1$) at each round with probability $\ge \frac{1}{2}$, $\Ptwo$ can simply always play $-1$ (resp.~$1$), which guarantees that the resulting mean payoff is $< 0$ (resp.\ $> 0$) with probability~$1$.
Condition~$\wc$ is then satisfied with probability $0$.

\begin{figure}[tbh]
	\centering
	\begin{tikzpicture}[every node/.style={font=\small,inner sep=1pt}]
		\draw (0,0) node[carre] (s1) {$\st_1$};
		\draw ($(s1)+(3,0)$) node[rond] (s2) {$\st_2$};
		\draw ($(s1)-(0.75,0)$) edge[-latex'] (s1);
		\draw (s1) edge[-latex',bend left=15] node[above=2pt] {$-1$} (s2);
		\draw (s1) edge[-latex',bend left=50] node[above=2pt] {$1$} (s2);
		\draw (s2) edge[-latex',bend left=15] node[below=2pt] {$-1$} (s1);
		\draw (s2) edge[-latex',bend left=50] node[below=2pt] {$1$} (s1);
	\end{tikzpicture}
	\caption{$\Pone$ can obtain $\wc$ with probability $1$, but not with a memoryless strategy. All transitions are deterministic; colors are shown, but action names are omitted.}\label{fig:MP}
\end{figure}

We have shown that even though (non-pure) memoryless strategies suffice for both players to play optimally in their respective one-player arenas, memoryless strategies do not suffice for $\Pone$ in two-player arenas.
This shows that Theorem~\ref{thm:1to2} does not work as stated if we allow randomized strategies.

\begin{rem}
	Throughout this example, we allowed rational randomization in the next-action function $\memNxt$.
	Another kind of randomization, closer to the intuitive idea of ``arena-independent randomization'', would be to allow fixed randomization in the initialization and update function of the memory skeletons (cf.\ Remark~\ref{rmk:stochSkel}).
	We leave open the question of whether this other kind of randomization can lead to interesting trade-offs with respect to memory requirements for some objectives, and whether the one-to-two-player lift could then hold with such randomization.
	\qedEx
\end{rem}

\section{AIFM characterization}\label{sec:1p}
In this section, we seek to characterize the preference relations for which pure strategies based on a memory skeleton suffice to play optimally in the \emph{one-player} arenas of $\Pone$, by decomposing this property into two properties.

For the section, we fix $\pref$ a preference relation, $\stratType\in\allStrats$ a type of strategies, and $\memSkel = \memSkelFull$ a memory skeleton.
We distinguish only two classes of initialized arenas: the class $\PoneDetArenas$ of all initialized one-player deterministic arenas of $\Pone$, and the class $\PoneStochArenas$ of all initialized one-player stochastic arenas of $\Pone$.
A class of arenas will therefore be specified by a letter $\arenaType\in\{\deter, \stoch\}$, which we fix for the whole section.
Rephrasing our goal with these notations, we seek to give a better understanding of the preference relations for which pure strategies based on $\memSkel$ suffice to play $\stratType$-optimally in $\PoneA$, by \emph{characterizing} it through two intuitive conditions.
All definitions and proofs are stated from the point of view of $\Pone$.
We first introduce some more notations.

As we only work with one-player arenas in this section, we abusively write $\prSolo{\arena}{\st}{\strat_1}$ and $\prcSolo{\arena}{\st}{\strat_1}$ for the distributions on plays and colors induced by a strategy $\strat_1$ of $\Pone$ on $(\arena, \st)$, with the unique, trivial strategy for $\Ptwo$.

For $\arena$ a one-player arena of $\Pone$ and $\st$ a state of $\arena$, we write
\[
	\cl{\arena}{\st}{\stratType} = \{\prcSolo{\arena}{\st}{\strat_1} \mid \strat_1\in\stratsType{1}{\stratType}{\arena, \st}\}
\]
for the set of distributions over $(\colors^\omega, \oalg)$ induced by strategies of type $\stratType$ in $\arena$ from $\st$.

For $\memState_1, \memState_2\in\memStates$, we write $\memLang{\memState_1}{\memState_2} = \{\word\in\colors^*\mid \memUpdHat(\memState_1, \word) = \memState_2\}$ for the language of words that are read from $\memState_1$ up to $\memState_2$ in $\memSkel$.
Such a language can be specified by the deterministic automaton that is simply the memory skeleton $\memSkel$ with $\memState_1$ as the initial state and $\memState_2$ as the unique final state.

We extend the \emph{shifted distribution} notation introduced in Definition~\ref{def:shift} to sets of distributions: for $\word\in\colors^*$, for $\distSet\subseteq\Dist(\colors^\omega, \oalg)$, we write $\word\distSet$ for the set $\{\word\dist \mid \dist\in\distSet\}$.

Given $\pref$ a preference relation, we extend it to sets of distributions: for $\distSet_1, \distSet_2\subseteq \Dist(\colors^\omega, \oalg)$, we write $\distSet_1 \pref \distSet_2$ if for all $\dist_1\in\distSet_1$, there exists $\dist_2\in\distSet_2$ such that $\dist_1\pref\dist_2$; we write $\distSet_1 \strictPref \distSet_2$ if there exists $\dist_2\in\distSet_2$ such that for all $\dist_1\in\distSet_1$, $\dist_1\strictPref\dist_2$.
Notice that $\lnot(\distSet_1 \pref \distSet_2)$ is equivalent to $\distSet_2 \strictPref \distSet_1$.
If $\distSet_1$ is a singleton $\{\dist_1\}$, we write $\dist_1\pref\distSet_2$ for $\{\dist_1\}\pref \distSet_2$ (and similarly for $\distSet_2$, and similarly using $\strictPref$).
Notice that $\dist_1\pref\dist_2$ is equivalent to $\{\dist_1\} \pref \{\dist_2\}$, so this notational shortcut is sound.
For two initialized arenas $(\arena_1, \st_1)$ and $(\arena_2, \st_2)$, the inequality $\cl{\arena_1}{\st_1}{\stratType} \pref \cl{\arena_2}{\st_2}{\stratType}$ means that for every strategy of type $\stratType$ on $(\arena_1, \st_1)$, there is a strategy of type $\stratType$ on $(\arena_2, \st_2)$ that induces a distribution that is at least as good.

We can now present the two properties of preference relations at the core of our characterization.
These properties are called \emph{$\stratType$-$\arenaType$-$\memSkel$-monotony} and \emph{$\stratType$-$\arenaType$-$\memSkel$-selectivity}; they depend on a type of strategies $\stratType$, a type of arenas $\arenaType$, and a memory skeleton $\memSkel$.
The first appearance of the monotony (resp.\ selectivity) notion was in~\cite{GZ05}, which dealt with deterministic arenas under pure strategies and memoryless strategies; their monotony (resp.\ selectivity) is equivalent to our $\pure$-$\deter$-$\memSkelTriv$-monotony (resp.\ $\pure$-$\deter$-$\memSkelTriv$-selectivity).
In~\cite{BLORV22}, these definitions were generalized to deal with the sufficiency of pure strategies based on $\memSkel$ in deterministic arenas; their notion of $\memSkel$-monotony (resp.\ $\memSkel$-selectivity) is equivalent to our $\pure$-$\deter$-$\memSkel$-monotony (resp.\ $\pure$-$\deter$-$\memSkel$-selectivity).

\begin{defi}[Monotony]
	We say that $\pref$ is \emph{$\stratType$-$\arenaType$-$\memSkel$-monotone} if for all $\memState\in\memStates$, for all $(\arena_1,\st_1), (\arena_2,\st_2) \in\PoneA$, there exists $i\in\{1, 2\}$ such that
	\[
		\forall \word\in\memLang{\memInit}{\memState},\,
		\word\cl{\arena_{3 - i}}{\st_{3-i}}{\stratType} \pref \word\cl{\arena_i}{\st_i}{\stratType}.
	\]
\end{defi}

The crucial part of the definition is the order of the last two quantifiers: of course, given a $\word\in\memLang{\memInit}{\memState}$, as $\pref$ is total, it will always be the case that $\word\cl{\arena_1}{\st_1}{\stratType} \pref \word\cl{\arena_2}{\st_2}{\stratType}$ or that $\word\cl{\arena_2}{\st_2}{\stratType} \pref \word\cl{\arena_1}{\st_1}{\stratType}$.
However, we ask for something stronger: it must be the case that the set of distributions $\word\cl{\arena_i}{\st_i}{\stratType}$ is preferred to $\word\cl{\arena_{3 - i}}{\st_{3-i}}{\stratType}$ for \emph{any} word $\word\in\memLang{\memInit}{\memState}$.

\begin{exa}\label{ex:monExamples}
	We can relate the notion of monotony to the more classical notion of prefix-independence: if a payoff function $\pf\colon \colors^\omega\to\IR$ is \emph{prefix-independent},%
	\footnote{A function $\pf\colon \colors^\omega\to\IR$ is \emph{prefix-independent} if for all $\word\in\colors^*$, for all $\word'\in\colors^\omega$, $\pf(\word\word') = \pf(\word')$.} then it is also $\stratType$-$\arenaType$-$\memSkel$-monotone for any $\stratType$, $\arenaType$, and $\memSkel$. 
	This is the case of the classical \emph{parity} objective and \emph{mean-payoff} payoff function.
	Thanks to the upcoming Theorem~\ref{thm:monSel}, studying the sufficiency of pure AIFM strategies for prefix-independent payoff functions immediately reduces to studying \emph{selectivity}.

	On the other hand, the weak parity winning condition $\simPar$ (defined in Example~\ref{ex:objExamples}) is not $\rand$-$\stoch$-$\memSkelTriv$-monotone.
	We consider two arenas $(\arena_1, \st_1)$ and $(\arena_2, \st_2)$ represented in Figure~\ref{fig:monExample}.
	As $\memSkelTriv$ has a single state, the only memory state to consider is $\memState = \memInit$, and all finite words are in $\memLang{\memInit}{\memState}$.
	Observe that if $\word$ is the empty word, then $\word\cl{\arena_{2}}{\st_{2}}{\rand} \strictPref \word\cl{\arena_1}{\st_1}{\rand}$, as the latter (is a singleton set whose only distribution) wins with probability $1$, while the former only wins with probability $\frac{1}{2}$.
	On the other hand, if $\word' = 1$, then $\word'\cl{\arena_{1}}{\st_{1}}{\rand} \strictPref \word'\cl{\arena_2}{\st_2}{\rand}$, as the latter wins with probability $\frac{1}{2}$ but the former only wins with probability $0$.
	This proves that $\simPar$ is not $\rand$-$\stoch$-$\memSkelTriv$-monotone.
	\qedEx
\end{exa}

\begin{figure}[tbh]
	\centering
	\begin{tikzpicture}[every node/.style={font=\small,inner sep=1pt}]
		\draw (0,0) node[rond] (s1) {$\st_1$};
		\draw ($(s1)-(0.75,0)$) edge[-latex'] (s1);
		\draw ($(s1)+(-1,.5)$) node (A1) {$\arena_1$};
		\draw (s1) edge[-latex',out=30,in=-30,distance=0.8cm] node[right=2pt] {$0$} (s1);

		\draw ($(s1)+(4,0)$) node[rond] (s2) {$\st_2$};
		\draw ($(s2)+(1,0)$) node[rond,fill=black,minimum size=3pt] (s22) {};
		\draw ($(s22)+(0.9,0.5)$) node[rond,minimum size=12pt] (s33) {};
		\draw ($(s22)+(0.9,-0.5)$) node[rond,minimum size=12pt] (s34) {};
		\draw ($(s22)+(1.8,0)$) node[rond,minimum size=12pt] (s32) {};
		\draw ($(s2)-(0.75,0)$) edge[-latex'] (s2);
		\draw (s2) edge[-latex'] node[above=2pt] {$0$} (s22);
		\draw ($(s2)+(-1,.5)$) node (A2) {$\arena_2$};
		\draw (s22) edge[-latex'] node[above=2pt,xshift=-2pt] {$\frac{1}{2}$} (s33);
		\draw (s22) edge[-latex'] node[below=2pt,xshift=-2pt] {$\frac{1}{2}$} (s34);
		\draw (s33) edge[-latex'] node[above=2pt,xshift=2pt] {$1$} (s32);
		\draw (s34) edge[-latex'] node[below=2pt,xshift=2pt] {$2$} (s32);
		\draw (s32) edge[-latex',out=30,in=-30,distance=10pt] node[right=2pt] {$0$} (s32);
	\end{tikzpicture}
	\caption{Arenas $(\arena_1, \st_1)$ and $(\arena_2, \st_2)$ used in Example~\ref{ex:monExamples}.
		Action names are omitted; integers next to the transitions represent the colors.}%
	\label{fig:monExample}
\end{figure}

The original monotony definition~\cite{GZ05} states that when presented with a choice \emph{once} among two possible continuations, if a continuation is better than the other one after some prefix, then this continuation is also at least as good after all prefixes.
This property is not sufficient for the existence of pure memoryless optimal strategies as it does not guarantee that if the same choice presents itself multiple times in the game, the same continuation should always be chosen, as alternating between both continuations might still be beneficial in the long run --- this is dealt with by selectivity.
If memory skeleton $\memSkel$ is necessary to play optimally, then it makes sense that there might be different optimal choices depending on the current memory state and that we should only compare prefixes that reach the same memory state.
The point of taking into account a memory skeleton $\memSkel$ in our definition of $\stratType$-$\arenaType$-$\memSkel$-monotony is to distinguish classes of prefixes and to only compare prefixes that are read up to the same memory state from $\memInit$.

\begin{defi}[Selectivity]
	We say that $\pref$ is \emph{$\stratType$-$\arenaType$-$\memSkel$-selective} if for all $\memState\in\memStates$, for all $(\arena_1, \st_1)$, $(\arena_2, \st_2)\in\PoneA$ such that for $i\in\{1, 2\}$, $\colHat{\Hists(\arena_i, \st_i, \st_i)} \subseteq \memLang{\memState}{\memState}$, for all $\word\in\memLang{\memInit}{\memState}$,
	\[
		\word\cl{(\arena_1, \st_1)\arenaJoin(\arena_2, \st_2)}{\joinState}{\stratType}
		\pref
		\word\cl{\arena_1}{\st_1}{\stratType} \cup \word\cl{\arena_2}{\st_2}{\stratType}
	\]
	(where $\joinState$ comes from the merge of $\st_1$ and $\st_2$).
\end{defi}

\begin{exa}\label{ex:selExample}
	We illustrate this definition by showing that the weak parity winning condition $\simPar$ (defined in Example~\ref{ex:objExamples}) is not $\rand$-$\stoch$-$\memSkelTriv$-selective.
	We consider the initialized arenas $(\arena_1, \st_1)$ and $(\arena_2, \st_2)$ from Figure~\ref{fig:selExample}.
	Let $\memState = \memInit$ be the only state of $\memSkelTriv$; once again, observe that all finite words are in $\memLang{\memInit}{\memState}$.
	Let $\word$ be the empty word.
	Observe that there is a distribution winning with probability $\frac{3}{4}$ in $\word\cl{(\arena_1, \st_1)\arenaJoin(\arena_2, \st_2)}{\joinState}{\rand}$: first try action $a$ (which has probability $\frac{1}{2}$ of winning immediately if the absorbing state is seen), and if it fails, pick action $b$, which offers again a probability $\frac{1}{2}$ of winning.
	However, both $\word\cl{\arena_1}{\st_1}{\rand}$ and $\word\cl{\arena_2}{\st_2}{\rand}$ offer at best a probability $\frac{1}{2}$ of winning.
	This shows that
	\[
	\word\cl{\arena_1}{\st_1}{\rand} \cup \word\cl{\arena_2}{\st_2}{\rand} \strictPref
	\word\cl{(\arena_1, \st_1)\arenaJoin(\arena_2, \st_2)}{\joinState}{\rand},
	\]
	so $\simPar$ is not $\rand$-$\stoch$-$\memSkelTriv$-selective.

	On the other hand, we will show in Section~\ref{sec:app:wp} that $\simPar$ is $\pure$-$\deter$-$\memSkelTriv$-selective.
	\qedEx
\end{exa}

\begin{figure}[tbh]
	\centering
	\begin{tikzpicture}[every node/.style={font=\small,inner sep=1pt}]
		\draw (0,0) node[rond] (s1) {$\st_1$};
		\draw ($(s1)-(0.75,0)$) edge[-latex'] (s1);
		\draw ($(s1)+(-1,.5)$) node (A1) {$\arena_1$};
		\draw ($(s1)+(1.5,0)$) node[rond,fill=black,minimum size=3pt] (s11) {};
		\draw ($(s1)+(1,1)$) node[rond,minimum size=12pt] (s3) {};
		\draw ($(s11)+(1.5,0)$) node[rond,minimum size=12pt] (s12) {};
		\draw (s1) edge[-latex'] node[below=2pt] {$a\mid 0$} (s11);
		\draw (s11) edge[-latex'] node[above=2pt] {$\frac{1}{2}$} (s12);
		\draw (s11) edge[-latex'] node[right=2pt,yshift=2pt] {$\frac{1}{2}$} (s3);
		\draw (s3) edge[-latex'] node[above left=1pt] {$1$} (s1);
		\draw (s12) edge[-latex',out=30,in=-30,distance=10pt] node[right=2pt] {$0$} (s12);

		\draw ($(s12)+(3,0)$) node[rond] (s2) {$\st_2$};
		\draw ($(s2)+(-1,.5)$) node (A2) {$\arena_2$};
		\draw ($(s2)+(1.5,0)$) node[rond,fill=black,minimum size=3pt] (s22) {};
		\draw ($(s22)+(1.5,0)$) node[rond,minimum size=12pt] (s32) {};
		\draw ($(s22)+(1,1)$) node[rond,minimum size=12pt] (s4) {};
		\draw ($(s2)-(0.75,0)$) edge[-latex'] (s2);

		\draw (s2) edge[-latex'] node[below=2pt] {$b\mid 1$} (s22);
		\draw (s22) edge[-latex'] node[below=2pt] {$\frac{1}{2}$} (s32);
		\draw (s22) edge[-latex'] node[above left=1pt] {$\frac{1}{2}$} (s4);
		\draw (s4) edge[-latex'] node[right=2pt,yshift=2pt] {$2$} (s32);
		\draw (s32) edge[-latex',out=30,in=-30,distance=10pt] node[right=2pt] {$0$} (s32);
	\end{tikzpicture}
	\caption{Arenas $(\arena_1, \st_1)$ and $(\arena_2, \st_2)$ used in Example~\ref{ex:selExample}. Only actions $a$ and $b$ are named. Notation $a\mid \clr$ indicates that color $\clr$ is seen when action $a$ is played.}%
	\label{fig:selExample}
\end{figure}

Our formulation of the selectivity concept differs from the original definition~\cite{GZ05} and its AIFM counterpart~\cite{BLORV22} in order to take into account the particularities of the stochastic context, even if it can be proven that they are equivalent in the pure deterministic case.
However, the idea is still the same: the original selectivity definition states that when presented with a choice among multiple possible continuations after some prefix, if a continuation is better than the others, then as the game goes on, if the same choice presents itself again, it is sufficient to always pick the same continuation to play optimally; there is no need to alternate between continuations.
This property is not sufficient for the existence of pure memoryless optimal strategies as it does not guarantee that for all prefixes, the same initial choice is always the one we should commit to --- this is dealt with by monotony.
The point of memory skeleton $\memSkel$ in our definition is to guarantee that every time the choice presents itself, we are currently in the same memory state $\memState$.

In both definitions, the point of $\stratType$ is to distinguish whether we allow all (including randomized) strategies, or only pure strategies; the point of $\arenaType$ is to quantify over a specific set of arenas.

An interesting property is that both notions are stable by product with a memory skeleton: if $\pref$ is $\stratType$-$\arenaType$-$\memSkel$-monotone (resp.\ $\stratType$-$\arenaType$-$\memSkel$-selective), then for all memory skeletons $\memSkel'$, $\pref$ is also $\stratType$-$\arenaType$-$(\memSkel\memProduct\memSkel')$-monotone (resp.\ $\stratType$-$\arenaType$-$(\memSkel\memProduct\memSkel')$-selective).
The reason is that in each definition, we quantify universally over the class of all prefixes $\word$ that reach the same memory state $\memState$; if we consider classes that are subsets of the original classes, then the definition still holds.
This property matches the idea that playing with more memory is never detrimental.

Combined, it is intuitively reasonable that $\stratType$-$\arenaType$-$\memSkel$-monotony and $\stratType$-$\arenaType$-$\memSkel$-selectivity are equivalent to the sufficiency of pure strategies based on $\memSkel$ to play $\stratType$-optimally in $\PoneA$: monotony tells us that when a single choice has to be made given a state of the arena and a memory state, the best choice is always the same no matter what prefix has been seen, and selectivity tells us that once a good choice has been made, we can commit to it in the future of the game.
We formalize and prove this idea in Theorem~\ref{thm:monSel}.
First, we add an extra restriction on preference relations which is useful when stochasticity is involved.

\begin{defi}[Mixing is useless]
	We say that \emph{mixing is useless for $\pref$} if for all sets $I$ at most countable, for all positive reals $(\lambda_i)_{i\in I}$ such that $\sum_{i\in I} \lambda_i = 1$, for all families $(\dist_i)_{i\in I}$, $(\dist'_i)_{i\in I}$ of distributions in $\Dist(\colors^\omega, \oalg)$,
	\[
	 (\forall i\in I, \dist_i \pref \dist'_i) \implies \sum_{i\in I} \lambda_i\dist_i \pref \sum_{i\in I} \lambda_i\dist'_i.
	\]
\end{defi}
That is, if we can write a distribution as a convex combination of distributions, then it is never detrimental to improve a distribution appearing in the convex combination.

\begin{rem}
	All preference relations encoded as Borel real payoff functions (as defined in Example~\ref{ex:objExamples}) satisfy this property (it is easy to show the property for indicator functions, and we can then extend this fact to all Borel functions thanks to properties of the Lebesgue integral).
	The third preference relation from Example~\ref{ex:objExamples} (having a probability to reach $\clr\in\colors$ that is precisely $\frac{1}{2}$) does not satisfy this property: if $\dist_1(\F \clr) = 0$, $\dist_1'(\F \clr) = \frac{1}{2}$, and $\dist_2(\F \clr) = 1$, we have $\dist_1 \strictPref \dist_1'$ and $\dist_2 \pref \dist_2$, but $\frac{1}{2}\dist_1' + \frac{1}{2}\dist_2\strictPref\frac{1}{2}\dist_1 + \frac{1}{2}\dist_2$.
	In case we consider pure strategies and deterministic games, only Dirac distributions on infinite words occur as probability distributions induced by an arena and a strategy, so the requirement that mixing is useless is not needed.
	\qedEx
\end{rem}

\begin{thm}\label{thm:monSel}
	Assume that no stochasticity is involved (that is, $\stratType\in\{\pure, \PFM\}$ and $\arenaType = \deter$), or that mixing is useless for $\pref$.
	Then pure strategies based on $\memSkel$ suffice to play $\stratType$-optimally in all initialized one-player arenas in $\PoneA$ for $\Pone$ if and only if $\pref$ is $\stratType$-$\arenaType$-$\memSkel$-monotone and $\stratType$-$\arenaType$-$\memSkel$-selective.
\end{thm}

We start with the proof of the necessary condition of Theorem~\ref{thm:monSel}, which is the easiest direction.
The main idea is to build the right arenas (using the arenas occurring in the definitions of monotony and selectivity) so that we can use the hypothesis about the existence of pure $\stratType$-optimal strategies based on $\memSkel$ to immediately deduce $\stratType$-$\arenaType$-$\memSkel$-monotony and $\stratType$-$\arenaType$-$\memSkel$-selectivity.
It is not necessary that mixing is useless for $\pref$ for this direction of the equivalence.

\begin{figure}[tbh]
	\centering
	\begin{minipage}{0.48\columnwidth}
	\centering
	\begin{tikzpicture}[every node/.style={font=\small,inner sep=1pt}]
		\draw (0,0) node[rond] (qw) {$\st_0^\word$};
		\draw ($(qw)+(2,0)$) node[rond] (qw') {$\st_0^{\word'}$};
		\draw ($(qw)+(0,-2.8)$) node[oval,minimum height=15mm,minimum width=10mm] (qk1) {$\arena_1$};
		\draw ($(qw')+(0,-2.8)$) node[oval,minimum height=15mm,minimum width=10mm] (qk2) {$\arena_2$};
		\draw ($(qw)!0.5!(qk2)$) node[rond] (sj) {$\joinState$};
		\draw ($(qw)+(-0.75,0)$) edge[-latex'] (qw);
		\draw ($(qw')+(0.75,0)$) edge[-latex'] (qw');
		\draw (qw) edge[decorate,-latex'] node[above right] {$\word$} (sj);
		\draw (qw') edge[decorate,-latex'] node[above left] {$\word'$} (sj);
		\draw (sj) edge[-latex',out=-145,in=90] (qk1);
		\draw (sj) edge[-latex',out=-125,in=70] (qk1);
		\draw (sj) edge[-latex',out=-35,in=90] (qk2);
		\draw (sj) edge[-latex',out=-55,in=110] (qk2);
	\end{tikzpicture}
	\end{minipage}
	\begin{minipage}{0.48\columnwidth}
	\centering
	\begin{tikzpicture}[every node/.style={font=\small,inner sep=1pt}]
		\draw (0,0) node[rond] (qw) {$\st_0^\word$};
		\draw ($(qw)+(-1,-2.8)$) node[oval,minimum height=15mm,minimum width=10mm] (qk1) {$\arena_1$};
		\draw ($(qw)+(1,-2.8)$) node[oval,minimum height=15mm,minimum width=10mm] (qk2) {$\arena_2$};
		\draw ($(qw) + (0,-1.4)$) node[rond] (sj) {$\joinState$};
		\draw ($(qw)+(-0.75,0)$) edge[-latex'] (qw);
		\draw (qw) edge[decorate,-latex'] node[right=2pt] {$\word$} (sj);
		\draw (sj) edge[-latex',out=-140,in=40] (qk1);
		\draw (sj) edge[-latex',out=-150,in=50] (qk1);
		\draw (qk1) edge[-latex',out=80,in=-180] (sj);
		\draw (qk1) edge[-latex',out=90,in=-190] (sj);
		\draw (sj) edge[-latex',out=-40,in=140] (qk2);
		\draw (sj) edge[-latex',out=-30,in=130] (qk2);
		\draw (qk2) edge[-latex',out=100,in=0] (sj);
		\draw (qk2) edge[-latex',out=90,in=10] (sj);
	\end{tikzpicture}
	\end{minipage}
	\caption{Initialized arenas $(\Amon, \{\st_0^\word,\st_0^{\word'}\})$ (left) and $(\Asel, \st_0^\word)$ (right).}%
	\label{fig:monSel}
\end{figure}
\begin{proof}[Proof of the necessary condition of Theorem~\ref{thm:monSel}]
	We assume that pure strategies based on $\memSkel$ suffice to play $\stratType$-optimally in $\PoneA$ for $\Pone$.

	We first prove that $\pref$ is $\stratType$-$\arenaType$-$\memSkel$-monotone.
	Let $\memState\in\memStates$ and $(\arena_1,\st_1), (\arena_2,\st_2) \in\PoneA$ be initialized one-player arenas of $\Pone$.
	If for all $\word\in\memLang{\memInit}{\memState}$, both $\word\cl{\arena_1}{\st_1}{\stratType} \pref \word\cl{\arena_2}{\st_2}{\stratType}$ and $\word\cl{\arena_2}{\st_2}{\stratType} \pref \word\cl{\arena_1}{\st_1}{\stratType}$, that is, if both sets of distributions are just as good as each other, then we can take either $i = 1$ or $i = 2$ and the definition is satisfied.
	If that is not the case, this means that there exists $\word'\in\memLang{\memInit}{\memState}$ such that, w.l.o.g.,
	\begin{equation}\label{eq:mon}
		\word'\cl{\arena_1}{\st_1}{\stratType} \strictPref \word'\cl{\arena_2}{\st_2}{\stratType}.
	\end{equation}
	We take $i = 2$.
	It is left to show that \emph{for all $\word\in\memLang{\memInit}{\memState}$},
	\[
	\word\cl{\arena_1}{\st_1}{\stratType} \pref \word\cl{\arena_2}{\st_2}{\stratType}.
	\]
	Let $\word\in\memLang{\memInit}{\memState}$.
	For $j\in\{1, 2\}$, we assume w.l.o.g.\ that state $\st_j$ has no incoming transition in $\arena_j$, and therefore cannot be reached after being left.
	If it is not the case, we can create a new one-player arena $\arena_j'$ by adding a new state $\st_j'$ mimicking the outgoing transitions of $\st_j$, but without any ingoing transition, and we have $\cl{\arena_j}{\st_j}{\stratType} = \cl{\arena_j'}{\st_j'}{\stratType}$.
	We also assume w.l.o.g.\ that the state and action spaces of $\arena_1$ and $\arena_2$ are disjoint.

	We consider the arena
	\[
		\Amon = \arenaPref{(\arenaPref{((\arena_1, \st_1) \arenaJoin (\arena_2, \st_2))}{\word}{\joinState})}{\word'}{\joinState}
	\]
	where $\joinState$ is the state resulting from the merge of $\st_1$ and $\st_2$.
	We consider two initial states $\st_0^\word$ and $\st_0^{\word'}$, which are the states at the start of the ``chains'' corresponding respectively to $\word$ and $\word'$.
	Arena $\Amon$ is depicted on Figure~\ref{fig:monSel}; $\Amon$ consists of two chains reading $\word$ and $\word'$ up to state $\joinState$, and then a choice between going to $\arena_1$ or $\arena_2$, with no possibility of ever going back to $\joinState$.

	The initialized arena $(\Amon, \{\st_0^\word,\st_0^{\word'}\})$ is in $\PoneA$ as all operations used preserve the number of players and the deterministic/stochastic feature.
	By hypothesis, $\Pone$ has a pure $\stratType$-optimal strategy $\strat_1\in\stratsType{1}{\PFM}{\Amon, \{\st_0^\word,\st_0^{\word'}\}}$ encoded as a Mealy machine $(\memSkel, \memNxt)$.
	Remember that both $\word$ and $\word'$ reach state $\memState$ of the memory skeleton $\memSkel$ when read from $\memInit$.
	Therefore, no matter whether the play starts in $\st_0^\word$ or $\st_0^{\word'}$, the action played in $\joinState$ by strategy $\strat_1$ is given by $\memNxt(\joinState, \memState)$ (which cannot be a randomized choice, as $\strat_1$ is pure).
	Since $\strat_1$ is $\stratType$-optimal in $(\Amon, \st_0^{\word'})$, by~\eqref{eq:mon}, this action must necessarily be an action of $\arena_2$.
	Now since $\strat_1$ is also $\stratType$-optimal in $(\Amon, \st_0^\word)$, this means that going to $\arena_2$ after $\word$ is at least as good as going to $\arena_1$.
	In other words, we have
	\[
	\word\cl{\arena_1}{\st_1}{\stratType} \pref \word\cl{\arena_2}{\st_2}{\stratType},
	\]
	which ends the $\stratType$-$\arenaType$-$\memSkel$-monotony proof.

	We now prove that $\pref$ is $\stratType$-$\arenaType$-$\memSkel$-selective.
	Let $\memState\in\memStates$ and $(\arena_1, \st_1), (\arena_2, \st_2)\in\PoneA$ such that for $i\in\{1, 2\}$, $\Hists(\arena_i, \st_i, \st_i) \subseteq \memLang{\memState}{\memState}$.
	Let $\word\in\memLang{\memInit}{\memState}$.
	We consider the arena
	\[
		\Asel = ((\arena_1, \st_1) \arenaJoin (\arena_2, \st_2))_{\word\leadsto\joinState}
	\]
	where $\joinState$ is the state resulting from the merge of $\st_1$ and $\st_2$.
	We consider an initial state $\st_0^\word$, which is the state at the start of the ``chain'' corresponding to $\word$.
	Arena $\Asel$ is depicted on Figure~\ref{fig:monSel}; $\Asel$ consists of one chain reading $\word$ up to state $\joinState$, and then has the ability to go either to $\arena_1$ and $\arena_2$.
	Here, it is possible to visit $\joinState$ multiple times (as long as it was possible to go back to $\st_1$ in $(\arena_1, \st_1)$ or to $\st_2$ in $(\arena_2, \st_2)$).

	The initialized arena $(\Asel, \st_0^\word)$ is in $\PoneA$ as all operations used preserve the number of players and the deterministic/stochastic feature.
	By hypothesis, $\Pone$ has a pure $\stratType$-optimal strategy $\strat_1\in\stratsType{1}{\PFM}{\Amon, \st_0^\word}$ encoded as a Mealy machine $(\memSkel, \memNxt)$.
	By $\stratType$-optimality of $\strat_1$, we have that
	\begin{equation}\label{eq:sel}
		\word\cl{(\arena_1, \st_1)\arenaJoin(\arena_2, \st_2)}{\joinState}{\stratType} \pref \prcSolo{\Asel}{\st_0^\word}{\strat_1}.
	\end{equation}
	Since $\word$ is in $\memLang{\memInit}{\memState}$ and for $i\in\{1, 2\}$, $\colHat{\Hists(\arena_i, \st_i, \st_i)}$ is a subset of $\memLang{\memState}{\memState}$, we have that $\colHat{\Hists(\Asel, \st_0^\word, \joinState)}$ is a subset of $\memLang{\memInit}{\memState}$.
	Therefore, at each passage in $\joinState$, strategy $\strat_1$ plays the action given by $\memNxt(\joinState, \memState)$ (which cannot be a randomized choice, as $\strat_1$ is pure).
	Strategy $\strat_1$ thus commits to $\arena_1$ or $\arena_2$ forever, which means that
	\[
		 \prcSolo{\Asel}{\st_0^\word}{\strat_1} \in \word\cl{\arena_1}{\st_1}{\stratType} \cup \word\cl{\arena_2}{\st_2}{\stratType}.
	\]
	By combining this last fact with~\eqref{eq:sel}, we obtain that
	\[
		\word\cl{(\arena_1, \st_1)\arenaJoin(\arena_2, \st_2)}{\joinState}{\stratType}
		\pref \word\cl{\arena_1}{\st_1}{\stratType} \cup 	\word\cl{\arena_2}{\st_2}{\stratType},
	\]
	which ends the $\stratType$-$\arenaType$-$\memSkel$-selectivity proof.
\end{proof}

We sketch the proof of the sufficient condition of Theorem~\ref{thm:monSel}.
We first reduce the problem to the existence of pure memoryless strategies in initialized arenas covered by $\memSkel$, using Lemma~\ref{prop:optProdIsCov}.
We proceed with an induction on the number of choices in these arenas (as for Theorem~\ref{thm:1to2}).
The base case is again trivial (as in an arena in which all states have a single available action, there is a single strategy which is pure and memoryless).
For the induction step, we take an initialized arena $(\arena',\initStates)\in\PoneA$ covered by $\memSkel$ with at least one choice, and we pick a state $t$ with (at least) two available actions.
A memory state $\coverFunction(t)$ is associated to $t$ thanks to the coverability property.
We consider the subarenas $(\arena'_\action, \initStates)$ with a single action $\action$ available in $t$, to which we can apply the induction hypothesis and obtain a pure memoryless $\stratType$-optimal strategy $\strat_1^\action$ in each subarena.
It is left to prove that one of these strategies is also $\stratType$-optimal in $(\arena',\initStates)$ --- this is where $\stratType$-$\arenaType$-$\memSkel$-monotony and $\stratType$-$\arenaType$-$\memSkel$-selectivity come into play.

The property of $\stratType$-$\arenaType$-$\memSkel$-monotony tells us that one of these subarenas $(\arena'_{\action^*}, \initStates)$ is preferred to the others w.r.t.\ $\pref$ after reading any word in $\memLang{\memInit}{\coverFunction(t)}$.
We now want to use $\stratType$-$\arenaType$-$\memSkel$-selectivity to conclude that there is no reason to use actions different from $\action^*$ when coming back to $t$, and that $\strat_1^{\action^*}$ is therefore also $\stratType$-optimal in $(\arena', \initStates)$.
To do so, we take any strategy $\strat_1\in\stratsType{1}{\stratType}{\arena', \st}$ for $\st\in\initStates$ and we condition distribution $\prSolo{\arena'}{\st}{\strat_1}$ over all the ways it reaches (or not) $t$, which gives a convex combination of probability distributions.
We want to state that once $t$ is reached, no matter how, switching to strategy $\strat_1^{\action^*}$ is always beneficial.
For this, we would like to use $\stratType$-subgame-perfection of $\strat_1^{\action^*}$ rather than simply $\stratType$-optimality: this is why in the actual proof, our induction hypothesis is about $\stratType$-SP strategies and not $\stratType$-optimal strategies.
Luckily, Theorem~\ref{prop:optAIFMimpliesSP} indicates that requiring subgame perfection is not really stronger than what we want to prove.
We then use that mixing is useless for $\pref$ to replace all the parts that go through $t$ in the convex combination by a better distribution induced by $\strat_1^{\action^*}$ from $t$.

We need two (intuitive) technical lemmas, whose proofs can be found in Appendix~\ref{sec:proofMonSel}.
We first define a similar notion to \emph{shifted distributions} (Definition~\ref{def:shift}) for distributions on plays: for $(\arena, \st)$ an initialized one-player arena, for $\hist\in\Hists(\arena, \st)$, if $\dist\in\Dist(\Plays(\arena, \histOut{\hist}), \oalg_{(\arena, \histOut{\hist})})$ is a distribution on plays, then for $E\in\oalg_{(\arena, \st)}$ an event, we define
\[
\hist\dist(E) = \dist(\{\play\in\Plays(\arena, \histOut{\hist}) \mid \hist\play\in E\}).
\]
We have used here an abuse of notation: if $\hist=\usualHist$, for $\play = \st_n\action_{n+1}\st_{n+1}\ldots\in\Plays(\arena, \histOut{\hist})$, we write $\hist\play$ for the play $\usualHist\action_{n+1}\st_{n+1}\ldots$, with no repetition of~$\st_n$.

\begin{restatable}{lem}{colHatDist}\label{lem:colHatDist}
	Let $(\arena = \arenaFull, \st)$ be an initialized one-player arena and $\hist\in\Hists(\arena, \st)$.
	Let $\dist$ be a distribution on plays in $\Dist(\Plays(\arena, \histOut{\hist}), \oalg_{(\arena, \histOut{\hist})})$.
	We have
	\[
		\colHat{\hist\dist} = \colHat{\hist}\colHat{\dist}.
	\]
\end{restatable}
We briefly recall some notations used in this last formula.
There are two different uses of notation $\colHatSolo$: $\colHat{\hist}$ maps history $\hist$ to a sequence of colors, while $\colHat{\hist\dist}$ and $\colHat{\dist}$ have as an input a distribution in $\Dist(\Plays(\arena, \st), \oalg_{(\arena, \st)})$ and map it to a distribution in $\Dist(\colors^\omega, \oalg)$.
Notations $\hist\dist$ and $\colHat{\hist}\colHat{\dist}$ denote shifted distributions.

\begin{restatable}{lem}{coincides}\label{lem:coincides}
	Let $(\arena, \st)\in\PoneA$ be an initialized one-player arena and $\strat_1,\stratBis_1\in\stratsType{1}{\rand}{\arena, \st}$ be two strategies.
	Let $\hist = \usualHist\in\Hists(\arena, \st)$.
	We say that \emph{$\strat_1$ coincides with $\stratBis_1$ on $\hist$} if for each prefix $\hist_i = \st_0\action_1\st_1\ldots\action_i\st_i$ of $\hist$ with $0\le i < n$, $\strat_1(\hist_i) = \stratBis_1(\hist_i)$.
	If $\strat_1$ coincides with $\stratBis_1$ on $\hist$, then
	\[
	\pr{\arena}{\st}{\strat_1}{\Cyl{\hist}} = \pr{\arena}{\st}{\stratBis_1}{\Cyl{\hist}}.
	\]
	Let $t$ be a state of $\arena$.
	We write $\lnot \F t$ for the event in $\oalg_{(\arena, \st)}$ that consists of all the infinite plays that never visit $t$.
	Assume that for all $\hist=\usualHist\in\Hists(\arena, \st)$ such that for all $i$, $0\le i < n$, $\st_i \neq t$, $\strat_1$ coincides with $\stratBis_1$ on $\hist$.
	Then
	\[
		\pr{\arena}{\st}{\strat_1}{\lnot \F t} = \pr{\arena}{\st}{\stratBis_1}{\lnot \F t}\
		\text{and, if $\pr{\arena}{\st}{\strat_1}{\lnot \F t} > 0$,}\
		\pr{\arena}{\st}{\strat_1}{\blank \mid \lnot \F t} = \pr{\arena}{\st}{\stratBis_1}{\blank \mid \lnot \F t}.
	\]
\end{restatable}

We now have all the ingredients for the proof of the missing implication of Theorem~\ref{thm:monSel}.

\begin{proof}[Proof of the sufficient condition of Theorem~\ref{thm:monSel}]
	We assume now that mixing is useless for $\pref$, and that $\pref$ is $\stratType$-$\arenaType$-$\memSkel$-monotone and $\stratType$-$\arenaType$-$\memSkel$-selective.
	We prove that pure strategies based on $\memSkel$ suffice to play $\stratType$-optimally in $\PoneA$ for $\Pone$.
	Equivalently, thanks to Lemma~\ref{prop:optProdIsCov}, we show that for all initialized arenas covered by $\memSkel$ in $\PoneA$, $\Pone$ has a pure memoryless $\stratType$-optimal strategy.
	We will actually prove something stronger, which is that for all initialized one-player arenas covered by $\memSkel$ in $\PoneA$, $\Pone$ has a pure memoryless $\stratType$-SP strategy.

	Let $\initArena\in\PoneA$ be an initialized one-player arena covered by $\memSkel$.
	Our proof proceeds by induction on the number of choices $\size{\arena'}$ of subarenas $(\arena', \initStates)$ of $\initArena$.
	Our induction will prove the following property for subarenas $(\arena', \initStates)$ of $\initArena$: there exists a pure memoryless strategy $\strat_1\in\stratsType{1}{\PFM}{\arena', \initStates}$ such that for all $\hist\in\Hists(\arena, \initStates)$, $\strat_1$ is $\stratType$-optimal in the game $(\arena', \histOut{\hist}, \shiftPref{\colHat{\hist}})$.
	We call this property \emph{having a pure memoryless $\initArena$-$\stratType$-SP strategy}.
	There is a slight abuse of notation in the definition: $\strat_1$ is not necessarily well-defined from $\histOut{\hist}$, but as it is pure memoryless, we simply interpret it as a function $\states \to \act$, and for $\hist'\in\Hists(\arena', \histOut{\hist})$, we define $\strat_1(\hist') = \strat_1(\histOut{\hist'})$.
	For subarenas $(\arena', \initStates)$, having a pure memoryless $\initArena$-$\stratType$-SP strategy is stronger than having a pure memoryless $\stratType$-SP strategy, as $\Hists(\arena', \initStates)$ is a subset of $\Hists(\arena, \initStates)$.
	For arena $\initArena$, having a pure memoryless $\initArena$-$\stratType$-SP strategy is equivalent to having a pure memoryless $\stratType$-SP strategy, which is what we want to prove.
	Requiring SP strategies instead of simply optimal strategies may seem stronger than what we actually need, but by Theorem~\ref{prop:optAIFMimpliesSP}, it turns out being equivalent in this AIFM context; we use SP strategies in this case for technical reasons.

	Let $(\arena' = (\states_1, \states_2, \act', \transProb', \colSolo'), \initStates)$ be a subarena of $\initArena$.
	If $\size{\arena'} = 0$, then $\Pone$ has only one strategy which is in particular a pure memoryless $\initArena$-$\stratType$-SP strategy (notation $\size{\arena'}$ is defined at~\eqref{eq:choices}).
	Now let $n > 0$; we assume that the property is true for all arenas $(\arena', \initStates)$ such that $\size{\arena'} < n$, and we take $(\arena', \initStates)$ such that $\size{\arena'} = n$.
	Since $n > 0$, there is a state $t\in\states_1$ such that $\card{\act'(t)} \ge 2$.

	For $\action\in\act'(t)$, let $(\arena'_\action, \initStates)\in\PoneA$ be the initialized subarena of $(\arena', \initStates)$ such that only action $\action$ is available in $t$.
	Initialized arena $(\arena'_\action, \initStates)$ is covered by $\memSkel$ (Lemma~\ref{lem:covAreClosed}).
	By induction hypothesis, for all $\action\in\act'(t)$, $\Pone$ has a pure memoryless $(\arena, \initStates)$-$\stratType$-SP strategy $\strat_1^\action$ in $(\arena_\action', \initStates)$.

	Let $\memState\in\memStates$ be the memory state corresponding to $t$ in $\initArena$, that is, if $\coverFunction$ is the function witnessing that $\initArena$ is covered, $\memState = \coverFunction(t)$.
	The same function $\coverFunction$ also witnesses that all the initialized subarenas of $\initArena$ are covered by $\memSkel$.
	As $\pref$ is $\stratType$-$\arenaType$-$\memSkel$-monotone, there exists $\action^*\in\act'(t)$ such that for all $\word\in\memLang{\memInit}{\memState}$, for all $\action\in\act'(t)$,
	\begin{equation}\label{eq:monotony}
		\word\cl{\arena_\action'}{t}{\stratType} \pref \word\cl{\arena_{\action^*}'}{t}{\stratType}.
	\end{equation}
	Notice that as $(\arena, \initStates)$ is covered by $\memSkel$, $\colHat{\Hists(\arena, \initStates, t)} \subseteq \memLang{\memInit}{\memState}$.

	We now prove that the pure memoryless strategy $\strat_1^{\action^*}$ is $\initArena$-$\stratType$-SP in $(\arena', \initStates)$.
	Let $\hist\in\Hists(\arena, \initStates)$.
	We denote $\word = \colHat{\hist}$ and $\st = \histOut{\hist}$.

	Let $\strat_1$ be any strategy in $\stratsType{1}{\stratType}{\arena', \st}$.
	Our goal is to show that $\strat_1^{\action^*}$ is at least as good as $\strat_1$ in $(\arena', \st, \shiftPref{\word})$, i.e., that $\word\prcSolo{\arena'}{\st}{\strat_1} \pref \word\prcSolo{\arena'}{\st}{\strat_1^{\action^*}}$.
	We condition $\prSolo{\arena'}{\st}{\strat_1}$ over whether $t$ is visited or not (we assume that $t$ is both visited and not visited with a non-zero probability --- otherwise, one of the terms of the following sum is simply $0$).
	We denote by $\F t$ the event of visiting state $t$ and by $\Hists(\arena', \st, t!)$ the set of histories in $\Hists(\arena', \st, t)$ that visit $t$ exactly \emph{once} (at their last step).
	We have
	\begin{align*}
		\prSolo{\arena'}{\st}{\strat_1}
		&= \pr{\arena'}{\st}{\strat_1}{\lnot\F t}\cdot
		\pr{\arena'}{\st}{\strat_1}{\blank \mid \lnot\F t} +
		\pr{\arena'}{\st}{\strat_1}{\F t}\cdot
		\pr{\arena'}{\st}{\strat_1}{\blank \mid \F t} \notag\\
		&= \pr{\arena'}{\st}{\strat_1}{\lnot\F t}\cdot
		\pr{\arena'}{\st}{\strat_1}{\blank \mid \lnot\F t} +
		\sum_{\hist'\in\Hists(\arena', \st, t!)}
		\pr{\arena'}{\st}{\strat_1}{\Cyl{\hist'}}\cdot
		\pr{\arena'}{\st}{\strat_1}{\blank \mid \Cyl{\hist'}}\notag\\
		&= \pr{\arena'}{\st}{\strat_1}{\lnot\F t}\cdot
		\pr{\arena'}{\st}{\strat_1}{\blank \mid \lnot\F t} +
		\sum_{\hist'\in\Hists(\arena', \st, t!)}
		\pr{\arena'}{\st}{\strat_1}{\Cyl{\hist'}}\cdot
		\hist'\prSolo{\arena'}{t}{\shiftStrat{\strat_1}{\hist'}}.
	\end{align*}

	By applying operator $\colHatSolo$ to $\prSolo{\arena'}{\st}{\strat_1}$ and shifting the distribution with $\word$, thanks to Lemma~\ref{lem:colHatDist} and the previous equation, we have
	\begin{equation}\label{eq:condition}
		\word\prcSolo{\arena'}{\st}{\strat_1} =
		\pr{\arena'}{\st}{\strat_1}{\lnot\F t}\cdot
		\word\colHat{\pr{\arena'}{\st}{\strat_1}{\blank \mid \lnot\F t}} +
		\sum_{\substack{\hist'\in\Hists(\arena', \st, t!)\\\text{s.t.\ $\colHat{\hist'} = \word'$}}}
		\pr{\arena'}{\st}{\strat_1}{\Cyl{\hist'}}\cdot
		\word\word'\prcSolo{\arena'}{t}{\shiftStrat{\strat_1}{\hist'}}.
	\end{equation}

	For $\hist'\in\Hists(\arena', \st, t!)$, $\word' = \colHat{\hist'}$, let us focus on the distribution $\word\word'\prcSolo{\arena'}{t}{\shiftStrat{\strat_1}{\hist'}}$.
	Notice that distribution $\prcSolo{\arena'}{t}{\shiftStrat{\strat_1}{\hist'}}$ can also be induced by some strategy in $\stratsType{1}{\stratType}{\splitArena{\arena'}{t}, t}$ by Lemma~\ref{lem:bijSplit}.
	Therefore,
	\[
		\word\word'\prcSolo{\arena'}{t}{\shiftStrat{\strat_1}{\hist'}}
		\in \word\word'\cl{\splitArena{\arena'}{t}}{t}{\stratType}
		= \word\word'\cl{\bigArenaJoin_{\action\in\act'(t)} (\arena'_\action, t)}{t}{\stratType}.
	\]

	Using the hypotheses, we get
	\begin{align*}
		\word\word'\prcSolo{\arena'}{t}{\shiftStrat{\strat_1}{\hist'}}
		&\in\word\word'\cl{\bigArenaJoin_{\action\in\act'(t)} (\arena'_\action, t)}{t}{\stratType} \\
		&\pref \bigcup_{\action\in\act'(t)} \word\word'\cl{\arena'_\action}{t}{\stratType} &&\text{by $\stratType$-$\arenaType$-$\memSkel$-selectivity}\\
		&\pref \word\word'\cl{\arena'_{\action^*}}{t}{\stratType} &&\text{by~\eqref{eq:monotony}, which relied on $\stratType$-$\arenaType$-$\memSkel$-monotony} \\
		&\pref \word\word'\prcSolo{\arena'_{\action^*}}{t}{\strat^{\action^*}_1}
		&&\text{as $\strat_1^{\action^*}$ is pure memoryless $\initArena$-$\stratType$-SP in $(\arena'_{\action^*}, \initStates)$.}
	\end{align*}

	Therefore, by using this last equation in~\eqref{eq:condition}, thanks to the fact that mixing is useless for $\pref$ (or, if we consider pure strategies and deterministic arenas, that the sum contains a single term corresponding to an infinite word), we obtain
	\begin{equation}\label{eq:condition2}
		\word\prcSolo{\arena'}{\st}{\strat_1}
		\pref \pr{\arena'}{\st}{\strat_1}{\lnot\F t}\cdot
		\word\colHat{\pr{\arena'}{\st}{\strat_1}{\blank \mid \lnot\F t}} +
		\sum_{\substack{\hist'\in\Hists(\arena, \st, t!)\\\text{with}\; \colHat{\hist'} = \word'}}
		\pr{\arena'}{\st}{\strat_1}{\Cyl{\hist'}}\cdot
		\word\word'\prcSolo{\arena'_{\action^*}}{t}{\strat^{\action^*}_1}.
	\end{equation}
	We show that the right-hand side of this inequality can be written as a distribution $\word\prcSolo{\arena'_{\action^*}}{\st}{\stratBis_1}$, for a suitably chosen strategy $\stratBis_1\in\stratsType{1}{\stratType}{\arena'_{\action^*}, \st}$.

	Let $\stratBis_1\in\stratsType{1}{\stratType}{\arena'_{\action^*}, \st}$ be such that $\stratBis_1$ starts playing like $\strat_1$ and then switches to $\strat_1^{\action^*}$ as soon as $t$ is visited; formally, for $\hist''\in\Hists(\arena'_{\action^*}, \st)$,
	\[
		\stratBis_1(\hist'') = \begin{cases}
			\strat_1(\hist'') &\text{if $\hist''$ does not visit $t$} \\
			\strat_1^{\action^*}(\histOut{\hist''}) &\text{if $\hist''$ visits $t$}.
		\end{cases}
	\]
	Strategy $\stratBis_1$ only plays action $\action^*$ in $t$, and is therefore a strategy on $(\arena'_{\action^*}, \st)$.
	As $\stratBis_1$ coincides with $\strat_1$ as long as $t$ has not been visited, using Lemma~\ref{lem:coincides}, we have
	\[
		\pr{\arena'}{\st}{\strat_1}{\lnot\F t} = \pr{\arena'_{\action^*}}{\st}{\stratBis_1}{\lnot\F t},\ \
		\pr{\arena'}{\st}{\strat_1}{\blank \mid \lnot\F t} = \pr{\arena'_{\action^*}}{\st}{\stratBis_1}{\blank \mid \lnot\F t},\ \
		\pr{\arena'}{\st}{\strat_1}{\Cyl{\hist'}} = \pr{\arena'_{\action^*}}{\st}{\stratBis_1}{\Cyl{\hist'}}.
	\]
	Moreover, for all $\hist'\in\Hists(\arena', \st, t!)$,
	\[
		\prcSolo{\arena'_{\action^*}}{t}{\strat^{\action^*}_1}
		= \prcSolo{\arena'_{\action^*}}{t}{\shiftStrat{\stratBis_1}{\hist'}}
	\]
	as $t$ is immediately visited.
	We can therefore replace all terms of the right-hand side of~\eqref{eq:condition2} and obtain, using Lemma~\ref{lem:colHatDist},
	\begin{align*}
		\word\prcSolo{\arena'}{\st}{\strat_1}
		&\pref \pr{\arena'_{\action^*}}{\st}{\stratBis_1}{\lnot\F t}\cdot
		\word\colHat{\pr{\arena'_{\action^*}}{\st}{\stratBis_1}{\blank \mid \lnot\F t}} +
		\sum_{\substack{\hist'\in\Hists(\arena, \st, t!)\\\text{with}\; \colHat{\hist'} = \word'}}
		\pr{\arena'_{\action^*}}{\st}{\stratBis_1}{\Cyl{\hist'}}\cdot
		\word\word'\prcSolo{\arena'_{\action^*}}{t}{\shiftStrat{\stratBis_1}{\hist'}} \\
		&\pref \word\colHat{\pr{\arena'_{\action^*}}{\st}{\stratBis_1}{\lnot\F t}\cdot
		\pr{\arena'_{\action^*}}{\st}{\stratBis_1}{\blank \mid \lnot\F t} +
		\sum_{\hist'\in\Hists(\arena, \st, t!)}
		\pr{\arena'_{\action^*}}{\st}{\stratBis_1}{\Cyl{\hist'}}\cdot
		\hist'\prSolo{\arena'_{\action^*}}{t}{\shiftStrat{\stratBis_1}{\hist'}}} \\
		&= \word\colHat{\pr{\arena'_{\action^*}}{\st}{\stratBis_1}{\lnot\F t}\cdot
		\pr{\arena'_{\action^*}}{\st}{\stratBis_1}{\blank \mid \lnot\F t} + \pr{\arena'_{\action^*}}{\st}{\stratBis_1}{\F t}\cdot
		\pr{\arena'_{\action^*}}{\st}{\stratBis_1}{\blank \mid \F t}} \\
		&= \word\colHat{\prSolo{\arena'_{\action^*}}{\st}{\stratBis_1}}\\
		&= \word\prcSolo{\arena'_{\action^*}}{\st}{\stratBis_1}.
	\end{align*}

Now since $\word$ is the sequence of colors corresponding to $\hist\in \Hists(\arena, \initStates)$ and $\strat_1^{\action^*}$ is $\initArena$-$\stratType$-SP in $(\arena'_{\action^*}, \initStates)$, we have $\word\prcSolo{\arena'_{\action^*}}{\st}{\stratBis_1} \pref \word\prcSolo{\arena'_{\action^*}}{\st}{\strat_1^{\action^*}}$, which ends the proof.
\end{proof}

We provide an application of Theorem~\ref{thm:monSel} in Section~\ref{sec:app:wp}, proving that a preference relation admits pure AIFM optimal strategies in its one-player games.
The literature provides some \emph{sufficient} conditions for preference relations to admit pure memoryless optimal strategies in one-player stochastic games (for instance, in~\cite{Gim07}).
Here, we obtain a full characterization when mixing is useless for $\pref$ (in particular, this is a full characterization for Borel real payoff functions), which can deal not only with memoryless strategies, but also with the more general AIFM strategies.
It therefore provides a more fundamental understanding of preference relations for which AIFM strategies suffice or do not suffice.
In particular, there are examples in which the known sufficient conditions are not verified even though pure memoryless strategies suffice (one such example is provided in~\cite{BBE10}), and that is for instance where our characterization can help.

\section{Examples}\label{sec:app}

We study two examples in more detail, proving claims from Section~\ref{sec:intro}.
The first one is the \emph{weak parity} objective, to which we can apply our results both for deterministic and stochastic games, obtaining different AIFM requirements.
The second one is a variant of the \emph{discounted sum} objective, which we use to show that even when AIFM strategies suffice for deterministic games, this may not be the case in stochastic games.

\subsection{Weak parity}\label{sec:app:wp}
Let $\colors = \IN$.
We illustrate the use of our two main theorems (Theorems~\ref{thm:monSel} and~\ref{thm:1to2}) to study the memory requirements of the \emph{weak parity}~\cite{Tho08} winning condition
\[
\simPar = \{\clr_1\clr_2\ldots\in\colors^\omega \mid
\max_{j\ge 1} \clr_j\text{ exists and is even}\},
\]
which was introduced in Example~\ref{ex:objExamples}, both in deterministic and in stochastic games.
In this example, we abusively use $\simPar$ for $\pref_{\simPar}$.
We say that a word $\word\in\colors^\omega$ is \emph{winning} if $\word\in\simPar$, and \emph{losing} if $\word\notin\simPar$.
As this preference relation can be encoded as a payoff function (namely, the indicator function of $\simPar$), we have that mixing is useless for $\simPar$.

\paragraph{Deterministic games.}
We first focus on deterministic games with pure strategies: we show that pure memoryless strategies are sufficient.
To do so, we first consider one-player games --- notice that reasoning about one-player games of $\Pone$ and of $\Ptwo$ is very similar, as the objective of $\Ptwo$ can be rephrased as the objective of $\Pone$ just by replacing all colors $\clr$ by $\clr+1$.
We thus only show arguments from the point of view of $\Pone$.
We prove that the class $\PoneDetArenas$ of all initialized one-player deterministic arenas of $\Pone$ admits pure memoryless $\pure$-optimal strategies (i.e., pure $\pure$-optimal strategies based on the trivial memory skeleton $\memSkelTriv$ with a single state) by proving that $\simPar$ is $\pure$-$\deter$-$\memSkelTriv$-monotone and $\pure$-$\deter$-$\memSkelTriv$-selective.

We start with $\pure$-$\deter$-$\memSkelTriv$-monotony.
Let $(\arena_1,\st_1), (\arena_2,\st_2) \in\PoneDetArenas$.
Notice that as we are restricted to pure strategies in deterministic arenas, notation $\cl{\arena_i}{\st_i}{\pure}$ refers to a set of (Dirac distributions on) infinite words.
For $i\in\{1, 2\}$ let
\[
e_i = \max \{\max_{j\ge 1} \clr_j \mid \clr_1\clr_2\ldots\in\cl{\arena_i}{\st_i}{\pure}\cap\simPar\}
\]
be the greatest even color reachable in $\arena_i$ without reaching any greater color (or $-\infty$ if it is not possible to have an even maximal color).

We first deal with the case $e_1\neq-\infty$ or $e_2\neq-\infty$.
Assume w.l.o.g.\ that $e_1 \le e_2$.
Let $\strat_1^2\in\stratsType{1}{\pure}{\arena_2, \st_2}$ be a pure strategy achieving a maximal color exactly $e_2$.
We prove that for any word $\word\in\colors^*$ (that is, for any word $\word\in\memLang{\memInit}{\memInit}$, as the memory skeleton has only one state), we have
\begin{equation}\label{eq:monDet}
	\word\cl{\arena_1}{\st_1}{\pure} \pref \word\prcSolo{\arena_2}{\st_2}{\strat_1^2}.
\end{equation}
Let $\word = \clr_1\clr_2\ldots\clr_n\in\colors^*$ and $n_\word = \max_{1\le j\le n} \clr_j$.
If $n_\word$ is even or $n_\word \le e_2$, then $\word\prcSolo{\arena_2}{\st_2}{\strat_1^2}$ is a winning word and~\eqref{eq:monDet} holds.
If not, it means that $n_\word$ is odd and $n_\word > e_2 \ge e_1$, in which case all words in $\word\cl{\arena_1}{\st_1}{\pure}$ are necessarily losing, and~\eqref{eq:monDet} also holds.

We now deal with the case $e_1 = e_2 = -\infty$, in which there is no way to obtain an even maximal color both in $\arena_1$ and in $\arena_2$.
For $i\in\{1, 2\}$, let
\[
o_i = \min\{\max_{j\ge 1} \clr_j \mid \clr_1\clr_2\ldots \in \cl{\arena_i}{\st_i}{\pure} \}
\]
be the minimal greatest color appearing along a play (which is necessarily odd, as an even greatest color is not possible).
Assume w.l.o.g.\ that $o_1 \ge o_2$, and let $\strat_1^2\in\stratsType{1}{\pure}{\arena_2, \st_2}$ be a pure strategy achieving a maximal color exactly $o_2$.
We show again that for all $\word\in\colors^*$,
\begin{equation}\label{eq:monDet2}
	\word\cl{\arena_1}{\st_1}{\pure} \pref \word\prcSolo{\arena_2}{\st_2}{\strat_1^2}.
\end{equation}
Let $\word = \clr_1\clr_2\ldots\clr_n\in\colors^*$ and $n_\word = \max_{1\le j\le n} \clr_j$.
If all words in $\word\cl{\arena_1}{\st_1}{\pure}$ are losing then~\eqref{eq:monDet2} is true.
If there is a winning word in $\word\cl{\arena_1}{\st_1}{\pure}$, it means that $n_\word$ is even and $n_\word > o_1 \ge o_2$.
Hence, $\word\prcSolo{\arena_2}{\st_2}{\strat_1^2}$ is also winning and~\eqref{eq:monDet2} also holds.

In both cases, we have $\word\cl{\arena_1}{\st_1}{\pure} \pref \word\prcSolo{\arena_2}{\st_2}{\strat_1^2}\pref \word\cl{\arena_2}{\st_2}{\pure}$.
This proves $\pure$-$\deter$-$\memSkelTriv$-monotony.

We now turn to $\pure$-$\deter$-$\memSkelTriv$-selectivity.
Let $(\arena_1,\st_1), (\arena_2,\st_2) \in\PoneDetArenas$.
Note that the requirement that $\colHat{\Hists(\arena_i, \st_i, \st_i)} \subseteq \memLang{\memInit}{\memInit}$ does not bring information as with this particular memory skeleton, all words are in $\memLang{\memInit}{\memInit}$.
Let $\word = \clr_1\clr_2\ldots\clr_n\in\colors^*$ and $n_\word = \max_{1\le j\le n} \clr_j$.
We prove that
\begin{equation}\label{eq:selDet}
	\word\cl{(\arena_1, \st_1)\arenaJoin(\arena_2, \st_2)}{\joinState}{\pure}
	\pref
	\word\cl{\arena_1}{\st_1}{\pure} \cup \word\cl{\arena_2}{\st_2}{\pure}.
\end{equation}
If all words of $\word\cl{(\arena_1, \st_1)\arenaJoin(\arena_2, \st_2)}{\joinState}{\pure}$ are losing, then~\eqref{eq:selDet} is true.
If there is a winning word $\word\word'\in\word\cl{(\arena_1, \st_1)\arenaJoin(\arena_2, \st_2)}{\joinState}{\pure}$, then we show that we can also find a winning word in $\word\cl{\arena_1}{\st_1}{\pure} \cup \word\cl{\arena_2}{\st_2}{\pure}$.

Assume word $\word\word'$ sees its maximal color $n$ in $\word'$.
If the play corresponding to $\word'$ in $(\arena_1, \st_1)\arenaJoin(\arena_2, \st_2)$ comes back to $t$ after seeing $n$ for the first time, then $n$ is the greatest color on some cycle on $t$.
This means that the strategy repeatedly playing this cycle only takes actions in $\arena_1$ or in $\arena_2$ and also wins.
If the play corresponding to $\word'$ does not come back to $t$ after seeing $n$ for the first time, then there is a suffix of the play fully in $\arena_1$ or in $\arena_2$ --- this suffix can be played after the first visit to $t$, and this generates a winning play.

Now, assume $\word\word'$ sees its maximal color $n$ in $\word$.
If the play corresponding to $\word'$ comes back to $t$, the strategy repeatedly playing this cycle on $t$ is winning, as no color greater than $n$ is seen in this cycle.
If there is no cycle on $t$, it means that $\word'$ is already an infinite word in $\cl{\arena_1}{\st_1}{\pure} \cup \word\cl{\arena_2}{\st_2}{\pure}$.

We have therefore shown~\eqref{eq:selDet} in every case; this shows $\pure$-$\deter$-$\memSkelTriv$-selectivity.

We have proven that $\simPar$ is $\pure$-$\deter$-$\memSkelTriv$-monotone and $\pure$-$\deter$-$\memSkelTriv$-selective; by Theorem~\ref{thm:monSel}, this implies that pure memoryless strategies are sufficient to play $\pure$-optimally in one-player deterministic arenas of $\Pone$.
The same arguments holds from the point of view of $\Ptwo$.
As we have shown that both players' one-player arenas admit pure memoryless $\pure$-optimal strategies, by Theorem~\ref{thm:1to2}, we conclude that both players have pure memoryless $\pure$-optimal (even $\pure$-SP) strategies in all two-player deterministic arenas.

\paragraph{Stochastic games.}
Interestingly, memory requirements of $\simPar$ are larger in stochastic games (which was already noticed in~\cite[Section~4.4]{GZ09}) but pure AIFM strategies still suffice and we can therefore apply our results.
An example of a one-player stochastic arena that requires memory is provided in Figure~\ref{fig:simPar}.
Intuitively, in this case, memory is necessary for correct risk assessment: it may sometimes be needed to attempt to get a greater color with a smaller probability, and that depends on the current maximal color.
In this example, keeping in memory the greatest color seen is sufficient to play optimally.

\begin{figure}[tbh]
	\centering
	\begin{tikzpicture}[every node/.style={font=\small,inner sep=1pt}]
		\draw (0,0) node[rond] (s1) {$\st_1$};
		\draw ($(s1)+(1,0)$) node[rond,fill=black,minimum size=3pt] (s1b) {};
		\draw ($(s1b)+(1,-1)$) node[rond,minimum size=12pt] (s1c) {};
		\draw ($(s1b)+(1.5,0)$) node[rond] (s2) {$\st_2$};
		\draw ($(s2)+(1.5,-0.5)$) node[rond, fill=black,minimum size=3pt] (s22) {};
		\draw ($(s22)+(1,-1)$) node[rond,minimum size=12pt] (s2c) {};
		\draw ($(s22)+(1.5,1.2)$) node[rond,minimum size=12pt] (s31) {};
		\draw ($(s22)+(1.5,0)$) node[rond,minimum size=12pt] (s32) {};
		\draw ($(s1)-(0.75,0)$) edge[-latex'] (s1);
		\draw (s1) edge[-latex'] node[above=2pt] {$0$} (s1b);
		\draw (s1b) edge[-latex'] node[above=2pt] {$\frac{1}{2}$} (s2);
		\draw (s1b) edge[-latex'] node[below left=1pt] {$\frac{1}{2}$} (s1c);
		\draw (s1c) edge[-latex'] node[below right=0.5pt,yshift=1pt] {$1$} (s2);
		\draw (s2) edge[-latex'] node[above=2pt,xshift=-1pt] {$a\mid 0$} (s31);
		\draw (s2) edge[-latex'] node[below=2pt,xshift=-4pt] {$b\mid 1$} (s22);
		\draw (s22) edge[-latex'] node[above=2pt] {$\frac{1}{2}$} (s32);
		\draw (s22) edge[-latex'] node[below left=1pt] {$\frac{1}{2}$} (s2c);
		\draw (s2c) edge[-latex'] node[below right=1pt] {$2$} (s32);
		\draw (s31) edge[-latex',out=30,in=-30,distance=10pt] node[right=2pt] {$0$} (s31);
		\draw (s32) edge[-latex',out=30,in=-30,distance=10pt] node[right=2pt] {$0$} (s32);
	\end{tikzpicture}
	\caption{Initialized one-player stochastic arena that requires memory (even if randomized strategies are allowed) for winning condition $\simPar$.
	$\Pone$ can win with probability $\frac{3}{4}$ by taking the risk to play action $b$ only if $1$ has been seen before.}%
	\label{fig:simPar}
\end{figure}
We generalize this idea and prove that memory skeleton $\Mmax = (\IN, 0, (\memState, n) \mapsto \max\{\memState, n\})$ suffices to play optimally in all stochastic arenas for both players (as argued earlier, although this skeleton is infinite, it is finite as soon as we restrict it to a finite set of colors).

We prove that the class $\PoneStochArenas$ of all initialized one-player stochastic arenas of $\Pone$ admits pure memoryless $\rand$-optimal strategies based on $\Mmax$ by proving $\rand$-$\stoch$-$\Mmax$-monotony and $\rand$-$\stoch$-$\Mmax$-selectivity.

The weak parity winning condition is not prefix-independent but using the definition of $\Mmax$, we prove the following related property: for all $\memState\in\IN$, for all finite words $\word_1, \word_2\in\memLang{0}{\memState}$, for all infinite words $\word\in\colors^\omega$,
\begin{equation}\label{eq:MPrefInd}
	\word_1\word \in \simPar \Longleftrightarrow \word_2\word\in\simPar.
\end{equation}
That is, similar prefixes (in the sense that they reach the same state of the memory skeleton) have the same influence on the outcome; the winning continuations are the same.

Let $\memState\in\IN$, $\word_1, \word_2\in\memLang{0}{\memState}$, and $\word = \clr_1\clr_2\ldots\in\colors^\omega$.
Assume $\word_1\word$ is winning.
Let $n_\word = \max_{j\ge 1} \clr_j$.
If $\memState \ge n_\word$, then it means $\memState$ is even and $\word_2\word$ is therefore also winning.
If $\memState < n_\word$, then it means that $n_\word$ is even and $\word_2\word$ is also winning.

Property~\eqref{eq:MPrefInd} implies the following for distributions: for $\dist\in\Dist(\colors^\omega, \oalg)$, for all $\word_1, \word_2\in\memLang{0}{\memState}$, $\word_1\dist(\simPar) = \word_2\dist(\simPar)$.
This implies $\rand$-$\stoch$-$\Mmax$-monotony: let $\memState\in\IN$ and $(\arena_1,\st_1), (\arena_2,\st_2) \in\PoneStochArenas$;
assume that for some $\word\in\memLang{0}{\memState}$, we have w.l.o.g.\
$
\word\cl{\arena_1}{\st_1}{\rand} \pref \word\cl{\arena_2}{\st_2}{\rand}.
$
Then we automatically have that for all $\word'\in\memLang{0}{\memState}$, we have
$
\word'\cl{\arena_1}{\st_1}{\rand} \pref \word'\cl{\arena_2}{\st_2}{\rand},
$
which proves $\rand$-$\stoch$-$\Mmax$-monotony.

We now turn to $\rand$-$\stoch$-$\Mmax$-selectivity.
Let $\memState\in\IN$ and $(\arena_1, \st_1), (\arena_2, \st_2)\in\PoneStochArenas$ such that for $i\in\{1, 2\}$, $\colHat{\Hists(\arena_i, \st_i, \st_i)} \subseteq \memLang{\memState}{\memState}$.
Let $\word\in\memLang{0}{\memState}$.
Let $\arena$ be the arena $\arenaPref{((\arena_1, \st_1) \arenaJoin (\arena_2, \st_2))}{\word}{t}$ with merged state $t$.

Thanks to the structure of the memory skeleton, we can make the following key observation: any play in $\Plays(\arena, \st_0^\word)$ that visits $t$ infinitely many times has a maximal color exactly $\memState$; indeed, $\memState$ is a color appearing in $\word$, and if a color greater that $\memState$ is seen, the memory state cannot go back down to $\memState$, so $t$ cannot be visited again (it would contradict that every history from $t$ to $t$ is in $\memLang{\memState}{\memState}$).

Let $\strat_1\in\stratsType{1}{\rand}{\arena, t}$.
Our goal is to show that it is possible to do at least as well as $\word\prc{\arena}{t}{\strat_1}{\simPar}$ without the need to use actions both in $\arena_1$ and in $\arena_2$ at $t$.

We first assume that $\memState$ is even: visiting $t$ infinitely often is therefore winning for $\Pone$.
If there is a strategy that, from $t$, comes back to $t$ with probability $1$, then $\Pone$ can achieve the objective with probability $1$ by repeatedly going back to $t$.
The use of randomization at $t$ is not necessary for this strategy: since it goes back to $t$ with probability $1$, every action it may play allows going back to $t$ with probability $1$.
Thus, such a strategy does not need to use actions both in $\arena_1$ and in $\arena_2$, as every time it leaves $t$, it can play the same action and repeat the strategy until it reaches $t$ again.
The $\rand$-$\stoch$-$\Mmax$-selectivity is therefore satisfied, as $\word\cl{\arena_1}{\st_1}{\rand}$ or $\word\cl{\arena_2}{\st_2}{\rand}$ contains a strategy that wins with probability $1$, which is at least as good as $\word\prc{\arena}{t}{\strat_1}{\simPar}$.

Assume now that $\memState$ is odd or that there is no strategy that comes back to $t$ with probability $1$.
In the latter case, the probability to go back to $t$ from $t$ has a probability less than $1 - \varepsilon$ for some $\varepsilon>0$ for all strategies; therefore, visiting $t$ infinitely often necessarily has probability $0$ for all strategies.
We condition $\word\prc{\arena}{t}{\strat_1}{\simPar}$ over which part of the arena the play ends in: either it visits $t$ infinitely often (event $\GG\F t$), or it sticks to $\arena_1$ or $\arena_2$ without visiting $t$ from some point on (events $\F\GG \arena_1\setminus t$ and $\F\GG \arena_2\setminus t$).

We have
\[
\word\prc{\arena}{t}{\strat_1}{\simPar}
= \pr{\arena}{t}{\strat_1}{\GG\F t}\cdot \word\prc{\arena}{t}{\strat_1}{\simPar \mid \GG\F t} +
\sum_{i\in\{1, 2\}} \pr{\arena}{t}{\strat_1}{\F\GG \arena_i\setminus t}\cdot \word\prc{\arena}{t}{\strat_1}{\simPar \mid \F\GG \arena_i\setminus t}.
\]
If $\memState$ is odd, all infinite plays in $\GG\F t$ are losing; if all strategies visit $t$ infinitely often with probability $0$, then $\pr{\arena}{t}{\strat_1}{\GG\F t} = 0$: in any case, the first term is $0$.

We focus on the last two terms.
For $i\in\{1, 2\}$, if the play stays in $\arena_i\setminus t$ from some point onward, as the value is independent from the actual prefix before the last visit to $t$ (by property~\eqref{eq:MPrefInd}), it means that it is possible to reach the same value while never going to $\arena_{3 - i}$.
That is, there exists a strategy $\strat_1^i\in\stratsType{1}{\rand}{\arena_i, t}$ such that
\[
\word\prc{\arena}{t}{\strat_1}{\simPar \mid \F\GG \arena_i\setminus t} \le \word\prc{\arena_i}{t}{\strat^i_1}{\simPar}.
\]
We do not prove it formally; a very similar argument can be found in the proof of~\cite[Theorem~4]{Gim07}: intuitively, it builds a strategy $\strat^i_1$ that induces a distribution on the projection of the plays of $(\arena, t)$ obtained by $\strat_1$ to plays of $(\arena_i, \st_i)$ (by removing the cycles on $t$ in $(\arena_{3 -i}, \st_{3-i})$).
Thus, if we play the best strategy among $\strat_1^1$, which obtains a value at least as good as the part that ends in $\arena_1\setminus t$, and $\strat_1^2$, which obtains a value at least as good as the part that ends in $\arena_2\setminus t$, what we obtain is something at least as good as the value obtained by $\strat_1$, without needing to consider actions both in $\arena_1$ and in $\arena_2$.

We have proven that $\simPar$ is $\rand$-$\stoch$-$\Mmax$-monotone and $\rand$-$\stoch$-$\Mmax$-selective; by Theorem~\ref{thm:monSel}, this implies that pure strategies based on $\Mmax$ are sufficient to play $\rand$-optimally in one-player stochastic arenas of $\Pone$.
The same arguments with the same memory skeleton holds from the point of view of $\Ptwo$.
As we have shown that both players' one-player arenas admit pure $\rand$-optimal strategies based on $\Mmax$, by Theorem~\ref{thm:1to2}, we conclude that both players have pure $\rand$-optimal (even $\rand$-SP) strategies based on $\Mmax\memProduct\Mmax$ (which corresponds to $\Mmax$) in all two-player stochastic arenas.

\subsection{Discounted sum with threshold}
Let $\colors = \IR$.
We consider the \emph{threshold problem for discounted sum}.
For $\discFac\in\intervaloo{0, 1}$, we take the payoff function $\disc$ as in Example~\ref{ex:objExamples}, but make it into a winning condition by setting a threshold at $0$, i.e., we define an event
\[
\wc = \{\word\in\IR^\omega \mid \disc(\word) \ge 0\}
\]
whose probability must be maximized by $\Pone$.

Memoryless strategies suffice in \emph{deterministic} arenas for $\wc$.
Indeed, it is sufficient to play the strategy that maximizes the value of the $\disc$ function: if this value is non-negative, then it means that it is possible to win for objective $\wc$, and if not, it is simply not possible to win.
Since a strategy that maximizes $\disc$ can be chosen to be memoryless~\cite{Sha53}, memoryless strategies also suffice for $\wc$.

That is not the case in \emph{stochastic} arenas: although pure memoryless strategies suffice to maximize the expected value of $\disc$~\cite{Sha53}, maximizing the probability of achieving $\wc$ may require some memory: intuitively, memory is necessary to assess how much risk should be taken.
We provide a formal proof that AIFM strategies are \emph{not} sufficient, even in one-player stochastic arenas of $\Pone$.
Let $\memSkel = \memSkelFull$ be any memory skeleton.
We are going to build an arena in which $\memSkel$ is not sufficient to play optimally.
Since $\memSkel$ is a memory skeleton, there are in particular only finitely many states reachable if we only read color $1$.
Therefore, there exist $n, m \ge 1$ with $n < m$ such that
\[
\memUpdHat(\memInit, \underbrace{1\ldots 1}_{\text{$n$ times}})
= \memUpdHat(\memInit, \underbrace{1\ldots 1}_{\text{$m$ times}}).
\]
Now, consider the arena in Figure~\ref{fig:disc}.
In this arena, $\Pone$ has just one choice to make in $\st_2$ among two actions $a$ and $b$, after reading either $n$ or $m$ times the color $1$.
If $\Pone$ has seen $n$ times the color $1$, then playing $a$ wins with probability $\frac{1}{2}$, but playing $b$ is a sure way to lose.
If $\Pone$ has seen $m$ times the color $1$, then playing $a$ still has probability $\frac{1}{2}$ to win, whereas playing $b$ is a sure way to win (the discounted sum ends up being exactly $0$).
It is therefore possible for $\Pone$ to win with probability $\frac{3}{4}$ by playing $a$ if color $1$ has been seen exactly $n$ times, and $b$ otherwise.

\begin{figure}[tbh]
	\centering
	\begin{tikzpicture}[every node/.style={font=\small,inner sep=1pt}]
		\draw (0,0) node[rond] (s1) {$\st_1$};
		\draw ($(s1)+(1,0)$) node[rond, fill=black,minimum size=3pt] (s1b) {};
		\draw ($(s1b)+(3,0)$) node[rond] (s2) {$\st_2$};
		\draw ($(s2)+(1.5,0.6)$) node[rond, fill=black,minimum size=3pt] (s21) {};
		\draw ($(s21)+(1,1)$) node[rond,minimum size=12pt] (s2c) {};
		\draw ($(s1b)+(1,0.5)$) node[rond,minimum size=12pt] (s1c) {};
		\draw ($(s1b)+(1,-.5)$) node[rond,minimum size=12pt] (s1d) {};
		\draw ($(s21)+(3,0)$) node[rond,minimum size=12pt] (s31) {};
		\draw ($(s2)+(4.5,-.6)$) node[rond,minimum size=12pt] (s32) {};
		\draw ($(s1)-(0.75,0)$) edge[-latex'] (s1);
		\draw (s1) edge[-latex'] node[above=2pt] {$1$}(s1b);
		\draw (s1b) edge[-latex'] node[above=3pt] {$\frac{1}{2}$} (s1c);
		\draw (s1b) edge[-latex'] node[below=3pt] {$\frac{1}{2}$} (s1d);
		\draw (s1c) edge[decorate,-latex',bend left=12] node[above=5pt] {$\overbrace{1\ldots1}^{\text{$n-1$ times}}$} (s2);
		\draw (s1d) edge[decorate,-latex',bend left=-12] node[below=5pt] {$\underbrace{1\ldots1}_{\text{$m-1$ times}}$} (s2);
		\draw (s2) edge[-latex'] node[above=2pt,xshift=-2pt] {$a\mid 0$} (s21);
		\draw (s2) edge[-latex',bend left=-12] node[below=2pt] {$b\mid-\discFac^{-m}\cdot(\sum_{i=0}^{m-1} \discFac^i)$} (s32);
		\draw (s21) edge[-latex'] node[below=2pt] {$\frac{1}{2}$} (s31);
		\draw (s21) edge[-latex'] node[above left=1pt] {$\frac{1}{2}$} (s2c);
		\draw (s2c) edge[-latex'] node[above right,xshift=-2pt] {$-\discFac^{-(m+1)}\cdot(1 + \sum_{i=0}^{m-1} \discFac^i)$} (s31);
		\draw (s31) edge[-latex',out=30,in=-30,distance=10pt] node[right=2pt] {$0$} (s31);
		\draw (s32) edge[-latex',out=30,in=-30,distance=10pt] node[right=2pt] {$0$} (s32);
	\end{tikzpicture}
	\caption{Strategies based on $\memSkel$ do not suffice to play optimally (even with randomization).
		Squiggly arrows indicate a sequence of transitions.}\label{fig:disc}
\end{figure}

This example might seem surprising, as if we can observe which transition has been taken at the first step, then two memory states seem sufficient.
What the whole reasoning shows is that there is no way to define a memory skeleton (which means that it can only be based on colors, and not on actual transitions) that suffices to play optimally on all arenas, while requiring only finitely many states for each individual arena.

\section{Conclusion}
We have studied stochastic games and gave an overview of desirable properties of preference relations that admit \emph{pure arena-independent finite-memory optimal strategies}.
Our analysis provides general tools to help study memory requirements in stochastic games, both with one player (Markov decision processes) and two players, and links both problems.
It generalizes both work on deterministic games~\cite{GZ05,BLORV22} and work on stochastic games~\cite{GZ09}.

We finally highlight a few remaining research directions to which our work does not yet give answers.
\begin{itemize}
	\item A natural question that remains unsolved is the link between memory requirements of a preference relation in deterministic and in stochastic games; our results can be called independently to study both problems, but do not describe a bridge to go from one to the other yet.
	\item Our results can only be used to show the optimality of \emph{pure} strategies with some fixed memory.
	For objectives expressible with a real payoff function, it is known that pure strategies always suffice for \emph{$\varepsilon$-optimality}~\cite[Theorem~4]{CDGH10} (the problem appears open for optimality if optimal strategies exist).
	This means that pure strategies suffice for many reasonable objectives.
	Still, in some cases, using \emph{randomized} strategies allows for lesser memory requirements~\cite{CAH04,Hor09,MPR20}.
	Investigating whether extensions to our results dealing with randomized strategies hold would therefore be valuable, but a first limit to such extensions is given by the example of Section~\ref{sec:noLiftRandomized}.
	\item Our main results in Sections~\ref{sec:1to2} and~\ref{sec:1p} deal with preference relations in which \emph{both} players have pure AIFM optimal strategies in two-player games, but not when a \emph{single} player has such.
	Sufficient conditions for pure memoryless optimal strategies for a single player were given in~\cite{GK14}, and an elegant characterization for memoryless optimal strategies in \emph{deterministic} games was given in~\cite{Ohl23}.
	Whether an interesting characterization can be obtained in stochastic games (already for memoryless strategies, but also for AIFM strategies) remains open.
	\item Even though we give ways to prove that a memory skeleton suffices in one-player or two-player games, our work does not provide a way to infer a sufficient memory skeleton (minimal or not).
	In deterministic games, there are works giving ways to compute minimal memory requirements, but usually for specific classes of objectives~\cite{DJW97,Hor09,CFH14}.
	For instance, even though we believe that the notions of monotony and selectivity from Section~\ref{sec:1p} bring insight, we leave as future work the question of whether they are ``decidable'' for reasonable classes of objectives.
	This could be a first step in providing a way to compute minimal memory skeletons.
\end{itemize}

\bibliography{fm}
\bibliographystyle{alphaurl}

\newpage
\appendix

\section{Proof of Lemma~\ref{lem:memProd}}\label{sec:memProdProof}
We restate and prove Lemma~\ref{lem:memProd} about the links between strategies with memory on an arena and memoryless strategies on a product arena.

\memProd*

\begin{proof}
	This proof goes through multiple steps, which all rely on establishing a correspondence between properties of $(\arena, \initStates)$ and $\prodAS{(\arena, \initStates)}{\memSkel}$.
	We first establish a bijection $\pathFun$ between their finite histories and then a bijection $\stratFun$ between their strategies.
	This is sufficient to show that the $\UCol_\pref^\stratType$ operators are preserved through $\stratFun$, which shows that $\stratType$-optimality is preserved through $\stratFun$.
	It is then left to show that $\stratFun(\strat_1)$ corresponds to $\memNxt$.

	We first define a bijection \[\pathFun\colon \Hists(\arena, \initStates)\to \Hists(\prodAS{(\arena, \initStates)}{\memSkel}).\]
	Let $\hist = \usualHist\in\Hists(\arena, \initStates)$.
	We set $\memState_0 = \memInit$, and for $1\le j\le n$, $\memState_j = \memUpd(\memState_{j-1}, \col{\st_{j-1},\action_j})$.
	We define $\pathFun(\hist) = (\st_0, \memState_0)\action_1(\st_1, \memState_1)\ldots\action_n(\st_n, \memState_n)$.
	Notice that $\colHat{\pathFun(\hist)} = \colHat{\hist}$.
	Furthermore, $\pathFun$ is bijective; as the initial state of the memory $\memInit$ is fixed and the memory skeleton is deterministic, the memory states added to $\hist$ to obtain $\pathFun(\hist)$ are uniquely determined.

	We now show that there is a correspondence between strategies of $\stratsType{i}{\stratType}{\arena, \initStates}$ and strategies of $\stratsType{i}{\stratType}{\prodAS{(\arena, \initStates)}{\memSkel}}$:
	intuitively, augmenting the arena with the skeleton allows some strategies to be played using less memory, but does not fundamentally change each player's possibilities.
	We define a function $\stratFun\colon \stratsType{i}{\stratType}{\arena, \initStates} \to \stratsType{i}{\stratType}{\prodAS{(\arena, \initStates)}{\memSkel}}$.
	For $\stratBis_i \in \stratsType{i}{\stratType}{\arena, \initStates}$, $\hist'\in\Hists_i(\prodAS{(\arena, \initStates)}{\memSkel})$, we define $\stratFun(\stratBis_i)(\hist') = \stratBis_i(\inverse{\pathFun}(\hist'))$ (exploiting that actions are the same in $(\arena, \initStates)$ as in $\prodAS{(\arena, \initStates)}{\memSkel}$).
	Function $\stratFun$ is bijective (for $\stratBis_i'\in\stratsType{i}{\stratType}{\prodAS{(\arena, \initStates)}{\memSkel}}$, its inverse $\inverse{\stratFun}$ can be specified as $\inverse{\stratFun}(\stratBis_i') = \stratBis_i' \circ \pathFun$).
	Moreover, it preserves the pure/randomized and the finite-memory/infinite-memory features of the strategies.

	We observe the following fact%
	\footnote{Remember that our preference relation is defined over distributions over \emph{sequences of colors}, and Equation~\eqref{eq:prcPreserved} compares two such distributions.} 
    about $\stratFun$: for all $\st\in\initStates$, for all $\stratBis_1\in \stratsType{1}{\stratType}{\arena, \initStates}, \stratBis_2 \in\stratsType{2}{\stratType}{\arena,\initStates}$,
	\begin{equation}\label{eq:prcPreserved}
		\prcSolo{\arena}{\st}{\stratBis_1, \stratBis_2} =
		\prcSoloLocalFix{\prodAS{(\arena, \st)}{\memSkel}}{\stratFun(\stratBis_1), \stratFun(\stratBis_2)}.
	\end{equation}
	It can easily be proven by induction that these probability distributions match on all cylinders (as they always induce the same distributions on the actions and on the colors after corresponding histories $\hist$ and $\pathFun(\hist)$), hence they are equal.

	Now let $\stratBis_1\in\stratsType{1}{\stratType}{\arena, \initStates}$.
	We notice that for all $\st\in\initStates$,
	\begin{align*}
		\UCol_\pref^\stratType(\arena, \st, \stratBis_1)
		&= \{ \dist \in \Dist(\colors^\omega, \oalg) \mid \exists \stratBis_2 \in \stratsType{2}{\stratType}{\arena, \initStates},\,
		\prcSolo{\arena}{\st}{\stratBis_1, \stratBis_2} \pref \dist\} \\
		&= \{ \dist \in \Dist(\colors^\omega, \oalg) \mid \exists \stratBis_2 \in \stratsType{2}{\stratType}{\arena, \initStates},\,
		\prcSolo{\prodAS{\arena}{\memSkel}}{(\st,\memInit)}{\stratFun(\stratBis_1), \stratFun(\stratBis_2)} \pref \dist\} \quad\text{by~\eqref{eq:prcPreserved}}\\
		&= \{ \dist \in \Dist(\colors^\omega, \oalg) \mid \exists \stratBis_2' \in \stratsType{2}{\stratType}{\prodAS{(\arena, \initStates)}{\memSkel}},\,
		\prcSolo{\prodAS{\arena}{\memSkel}}{(\st,\memInit)}{\stratFun(\stratBis_1), \stratBis_2'} \pref \dist\} \\
		&= \UCol_\pref^\stratType(\prodAS{(\arena, \initStates)}{\memSkel}, (\st, \memInit), \stratFun(\stratBis_1)),
	\end{align*}
	where the penultimate line holds by bijectivity of $\stratFun$.
	The property holds symmetrically for a strategy $\stratBis_2\in\stratsType{2}{\stratType}{\arena, \initStates}$.
	Thus $\stratType$-optimality of strategies is preserved through $\stratFun$.

	Now remember that $\strat_i \in \stratsType{i}{\stratType}{\arena, \initStates}$ is a strategy encoded by a Mealy machine $(\memSkel, \memNxt)$.
	We notice that $\stratFun(\strat_i)$ corresponds to $\memNxt$ interpreted over the product initialized arena and is thus memoryless.
	By the previous property, we have that $\strat_i$ is $\stratType$-optimal in $\game$ if and only if $\memNxt$ is $\stratType$-optimal in $\game'$.
\end{proof}

\section{Results on splits}\label{sec:split}
We recall here technical results about \emph{split arenas} (Definition~\ref{def:split}) that are already present in~\cite{GZ09} with the slight difference that we consider \emph{initialized} arenas.

Let $(\arena = \arenaFull, \initStates)$ be an initialized arena, $t\in\states_1$ be a state controlled by $\Pone$, and $(\splitArena{\arena}{t}, \splitInitStates{t})$ be the split of $\initArena$ on $t$.

For all $\action\in\actions{t}$, we build a natural bijection between plays of the arena and plays of its split: it is a function
\[
\bijPaths{\action}\colon\Hists(\arena, \initStates) \to \Hists(\splitArena{\arena}{t}, \initStates^\action).
\]
This function simply labels the different states appearing along the play with the right action to make it a play of the split, starting arbitrarily with $\action$: if $t$ has never been visited, it picks action $\action$ by default; if $t$ has been visited, it picks the last action played in $t$.
Formally, let $\hist = \usualHist\in\Hists(\arena, \initStates)$.
We define $\bijPaths{\action}(\hist) = \st_0^{\action'_0}\action_1\st_1^{\action'_1}\ldots\action_n\st_n^{\action'_n}$ where we assume as usual that $t^{\action'_i} = t$, and for $i$, with $0\le i\le n$ such that $\st_i \neq t$,
\[
	\action_i' = \begin{cases}
		\action &\text{if for all $j < i$, $\st_j \neq t$}, \\
		\action_{k+1} &\text{if $k$ is the index of the visit to $t$ preceding $\st_i$ in $\hist$}.
	\end{cases}
\]
The history $\bijPaths{\action}(\hist)$ is a history of the split by construction.
Function $\bijPaths{\action}$ has an inverse $(\bijPaths{\action})^{-1}$ which associates to any history of the split starting in $\initStates^\action$ the same history in which all the action labels have been removed.
We can extend function $\bijPaths{\action}$ to a bijection on plays: let
\[
	\bijPlays{\action}\colon \Plays(\arena, \initStates) \to \Plays(\splitArena{\arena}{t}, \initStates^\action)
\]
be the function such that for $\play = \st_0\action_1\st_1\action_2\st_2\ldots\in\Plays(\arena, \initStates)$, if $\hist_n = \st_0\action_1\st_1\ldots\action_n\st_n$ is a prefix of $\play$, then $\bijPlays{\action}(\play) = \lim_{n\to\infty} \bijPaths{\action}(\hist_n)$.

For $i\in\{1, 2\}$, for all $\action\in\actions{t}$, we build a natural bijection between strategies of the arena and strategies of its split: it is a function
\[
\bijSplit{i}{\action} \colon \stratsType{i}{\rand}{\arena, \initStates} \to \stratsType{i}{\rand}{\splitArena{\arena}{t}, \initStates^\action}.
\]
For $\strat_i\in\stratsType{i}{\rand}{\arena, \initStates}$ a strategy of $\player{i}$, we define
$
	\bijSplit{i}{\action}(\strat_i) = \strat_i \circ (\bijPaths{\action})^{-1}.
$
Similarly, this function has an inverse: for $\strat_i\in\stratsType{i}{\rand}{\splitArena{\arena}{t}, \initStates^\action}$, we define $(\bijSplit{i}{\action})^{-1}(\strat_i) = \strat_i \circ \bijPaths{\action}$.

We prove a few results about these bijections, which correspond to~\cite[Proposition~10 and Lemma~12]{GZ09}.

\begin{lem}\label{lem:bijSplit}
	Let $\action\in\actions{t}$, $\st\in\initStates$, $\strat_1\in\stratsType{1}{\rand}{\arena, \initStates}$, and $\strat_2\in\stratsType{2}{\rand}{\arena, \initStates}$.
	We have
	\[
		\prcSolo{\arena}{\st}{\strat_1, \strat_2} =
		\prcSolo{\splitArena{\arena}{t}}{\st^\action}{\bijSplit{1}{\action}(\strat_1), \bijSplit{2}{\action}(\strat_2)}.
	\]
	Let $\pref$ a preference relation and $\stratType\in\allStrats$ be a type of strategies.
	For $i\in\{1, 2\}$, strategy $\strat_i$ is pure (resp.\ finite-memory) if and only if strategy $\bijSplit{i}{\action}(\strat_i)$ is pure (resp.\ finite-memory).
	For $i\in\{1, 2\}$, $\strat_i$ is $\stratType$-optimal in $(\arena, \initStates, \pref)$ if and only if $\bijSplit{i}{\action}(\strat_i)$ is $\stratType$-optimal in $(\splitArena{\arena}{t}, \initStates^\action, \pref)$.
	Moreover, $(\strat_1, \strat_2)$ is an $\stratType$-NE in $(\arena, \initStates, \pref)$ if and only if $(\bijSplit{1}{\action}(\strat_1), \bijSplit{2}{\action}(\strat_2))$ is an $\stratType$-NE in $(\splitArena{\arena}{t}, \initStates^\action, \pref)$.
\end{lem}
\begin{proof}
	We first show the equality of two distributions in $\Dist(\Plays(\arena, \st), \oalg_{(\arena, \st)})$:
	\begin{equation}\label{eq:bijSplit}
		\pr{\arena}{\st}{\strat_1, \strat_2}{\cdot} =
		\pr{\splitArena{\arena}{t}}{\st^\action}{\bijSplit{1}{\action}(\strat_1), \bijSplit{2}{\action}(\strat_2)}{\bijPlays{\action}(\cdot)}.
	\end{equation}
	We prove the equality for cylinders $\Cyl{\hist}$ with $\hist\in\Hists(\arena, \initStates)$.
	Notice that $\bijPlays{\action}(\Cyl{\hist}) = \Cyl{\bijPaths{\action}(\hist)}$.
	Thanks to our construction of functions $\bijSplit{1}{\action}$ and $\bijSplit{2}{\action}$, an easy induction on the length of $\hist$ shows that
	\[
	\pr{\arena}{\st}{\strat_1, \strat_2}{\Cyl{\hist}} =
	\pr{\splitArena{\arena}{t}}{\st^\action}{\bijSplit{1}{\action}(\strat_1), \bijSplit{2}{\action}(\strat_2)}{\Cyl{\bijPaths{\action}(\hist)}}.
	\]
	Since cylinders generate the $\sigma$-algebra, this proves~\eqref{eq:bijSplit}.

	Now notice that the bijection on plays $\bijPlays{\action}$ preserves the sequence of colors seen, i.e., for all $\play\in\Plays(\arena, \initStates)$, $\colHat{\play} = \colHat{\bijPlays{\action}(\play)}$.
	Using the definition of $\prcOperator$, we can therefore conclude that
	\[
	\prcSolo{\arena}{\st}{\strat_1, \strat_2} =
	\prcSolo{\splitArena{\arena}{t}}{\st^\action}{\bijSplit{1}{\action}(\strat_1), \bijSplit{2}{\action}(\strat_2)}.
	\]

	The claims about $\stratType$-optimality and $\stratType$-NE follow from the first one: as $\bijSplit{i}{\action}$ is a bijection that preserves the induced distributions on colors, it also preserves $\stratType$-optimality and $\stratType$-NE\@.

	Bijection $\bijSplit{i}{\action}$ clearly preserves the ``pure'' feature of strategies by construction, in both directions.
	Now if $\strat_i$ is finite-memory, then $\bijSplit{i}{\action}(\strat_i)$ does not need any more memory to play as it has access to the same information and the last action played in $t$.
	In the other direction, if $\bijSplit{i}{\action}(\strat_i)$ is finite-memory, then $\strat_i$ can play with the same memory \emph{plus} extra information about the last action that was played in $t$.
	Therefore, $\strat_i$ might need more memory than $\bijSplit{i}{\action}(\strat_i)$, but that memory stays finite.%
	\footnote{Formally, our memory model is based on colors, and not on actions. However, we can easily enrich the game graph with a new color for each action available at $t$, and a memory skeleton (which reads colors) can then remember the last action played at $t$.} 
\end{proof}

\begin{lem}\label{lem:splitProjection}
	Let $\pref$ be a preference relation and $\stratType\in\allStrats$ be a type of strategies.
	Let $(\arena = \arenaFull, \initStates)$ be an initialized arena, and $(\splitArena{\arena}{t}, \splitInitStates{t})$ be its split on $t$ for some $t\in\states_1$.
	Assume $(\strat_1, \strat_2)$ is an $\stratType$-NE in $(\splitArena{\arena}{t}, \initStates^{\action^*})$ for some $\action^*\in\actions{t}$.
	If $\strat_1$ is pure memoryless, and $\strat_1(t) = \action^*$, then there exists an $\stratType$-NE $(\strat_1', \strat_2')$ in $(\arena, \initStates)$ such that $\strat_1'$ is pure memoryless.
\end{lem}
\begin{proof}
	We set $(\strat_1', \strat_2') = ((\bijSplit{1}{\action^*})^{-1}(\strat_1), (\bijSplit{2}{\action^*})^{-1}(\strat_2))$, which is an $\stratType$-NE in $\initArena$ by Lemma~\ref{lem:bijSplit}.
	Moreover, Lemma~\ref{lem:bijSplit} shows that $\strat_1'$ is pure.
	It is left to prove that in this particular case, $\strat_1'$ is memoryless.
	Let $\hist\in\Hists_1(\arena, \initStates)$ be a history consistent with $\strat_1'$.
	We know that if $\histOut{\hist} = t$, then $\strat_1'(\hist) = \action^*$ because this is the only possible action played in $t$ by $\strat_1$ in $\splitArena{\arena}{t}$.
	Now assume $\histOut{\hist} \neq t$.
	Notice that since any action taken at $t$ is necessarily $\action^*$ (and since we start in $\initStates^{\action^*}$), any state appearing along $\bijPaths{\action^*}(\hist)$ (except $t$) is necessarily labeled by $\action^*$.
	Therefore, we have $\strat_1'(\hist) = \strat_1(\bijPaths{\action^*}(\hist)) = \strat_1(\st^{\action^*})$, which only depends on $\st$.
\end{proof}

\begin{lem}\label{lem:stratSplit}
	Let $\pref$ be a preference relation.
	Let $(\arena = \arenaFull, \initStates)$ be an initialized arena, and $(\splitArena{\arena}{t}, \splitInitStates{t})$ be its split on $t$ for some $t\in\states_1$.
	Let $\strat_1\in\stratsType{1}{\PFM}{\splitArena{\arena}{t}, \splitInitStates{t}}$ and $\strat_2\in\stratsType{2}{\rand}{\splitArena{\arena}{t}, \splitInitStates{t}}$.
	Assume $\strat_1$ is pure memoryless, and let $\action^* = \strat_1(t)$.
	Let $(\arena_{\action^*}, \initStates^{\action^*})$ be the initialized subarena of $\initArena$ in which only action $\action^*$ is available in $t$ and states are renamed $\st\mapsto\st^{\action^*}$.
	Let $\strat_1^{\action^*}$ and $\strat_2^{\action^*}$ be the restrictions of $\strat_1$ and $\strat_2$ to $\Hists(\arena_{\action^*}, \initStates^{\action^*})$, which are strategies on $(\arena_{\action^*}, \initStates^{\action^*})$.
	Then for all $\st^{\action^*} \in \initStates^{\action^*}$,
	\[
	\prcSolo{\splitArena{\arena}{t}}{\st^{\action^*}}{\strat_1, \strat_2} = \prcSolo{\arena_{\action^*}}{\st^{\action^*}}{\strat_1^{\action^*}, \strat_2^{\action^*}}.
	\]
\end{lem}

\begin{proof}
	Let $\st^{\action^*}\in\initStates^{\action^*}$.
	Notice that any play in $\splitArena{\arena}{t}$ consistent with $\strat_1$ starting in $\st^{\action^*}$ only visits states among $\states^{\action^*}$, since it starts there and every action played in $t$ is $\action^*$.
	This shows that $\strat_1^{\action^*}$ and $\strat_2^{\action^*}$ are indeed well-defined strategies on $(\arena_{\action^*}, \initStates^{\action^*})$.
	Moreover, it shows that
	\[
		\prcSolo{\splitArena{\arena}{t}}{\st^{\action^*}}{\strat_1, \strat_2} = \prcSolo{\arena_{\action^*}}{\st^{\action^*}}{\strat_1^{\action^*}, \strat_2^{\action^*}}
	\]
	since every infinite play stays in the subarena of $\splitArena{\arena}{t}$ corresponding to $\arena_{\action^*}$.
\end{proof}

\section{Missing proofs of Section~\ref{sec:1p}}\label{sec:proofMonSel}
We prove the two technical lemmas from Section~\ref{sec:1p} that were stated without proof.

\colHatDist*
\begin{proof}
	Let $E\in \oalg$ be an event about infinite sequences of colors.
	Then,
	\begin{align*}
		&\colHat{\hist\dist}(E) \\
		&= \hist\dist(\colHatSolo^{-1}(E)) &&\text{by definition of $\colHatSolo$ on distributions}\\
		&= \dist(\{\play\in\Plays(\arena, \histOut{\hist}) \mid \hist\play\in\colHatSolo^{-1}(E)\}) &&\text{by definition of $\hist\dist$}\\
		&= \dist(\{\play\in\Plays(\arena, \histOut{\hist}) \mid \colHat{\hist\play}\in E\}) \\
		&= \dist(\{\play\in\Plays(\arena, \histOut{\hist}) \mid \colHat{\hist}\colHat{\play}\in E\}) \\
		&= (\dist \circ \colHatSolo^{-1})(\{\word'\in\colors^\omega \mid \colHat{\hist}\word'\in E\}) \\
		&= \colHat{\dist}(\{\word'\in\colors^\omega \mid \colHat{\hist}\word'\in E\}) &&\text{by definition of $\colHatSolo$ on distributions}\\
		&= (\colHat{\hist}\colHat{\dist})(E) &&\text{by definition of $\colHat{\hist}\colHat{\dist}$.}\qedhere
	\end{align*}
\end{proof}

\coincides*
\begin{proof}
	Assume $\arena = \arenaFull$.
	Using the definition of $\pr{\arena}{\st}{\blank}{\Cyl{\hist}}$ and the hypothesis, we have
	\[
	\pr{\arena}{\st}{\strat_1}{\Cyl{\hist}}
	= \prod_{i = 1}^{n} \strat_1(\hist_{i-1})(\action_i)\cdot\transProb(\st_{i-1}, \action_i, \st_i)
	= \prod_{i = 1}^{n} \stratBis_1(\hist_{i-1})(\action_i)\cdot\transProb(\st_{i-1}, \action_i, \st_i)
	= \pr{\arena}{\st}{\stratBis_1}{\Cyl{\hist}}
	\]
	which proves the first claim.

	Let $H = \{\Cyl{\hist}\mid \hist\in\Hists(\arena, \st)\text{ and $\strat_1$ coincides with $\stratBis_1$ on $\hist$}\}$.
	By the first claim, $\prSolo{\arena}{\st}{\strat_1}$ and $\prSolo{\arena}{\st}{\stratBis_1}$ are equal on $H$.
	As this class is closed by intersection, by the monotone class lemma, $\prSolo{\arena}{\st}{\strat_1}$ and $\prSolo{\arena}{\st}{\stratBis_1}$ are also equal on the smallest $\sigma$-algebra generated by $H$.

	For $E\in \oalg_{(\arena, \st)}$ an event, we denote by $\complement{E} = \Plays(\arena, \st) \setminus E$ its complement;
	notice that
	\begin{align*}
		\lnot\F t
		&= \complement{\Big( \bigcup_{n\in\IN} \Cyl{\underbrace{\states\setminus\{t\}, \act, \ldots, \states\setminus\{t\}, \act}_{\text{$n$ times ``$\states\setminus\{t\}, \act$''}}, t} \Big)} \\
		&= \bigcap_{n\in\IN} \complement{\Cyl{\underbrace{\states\setminus\{t\}, \act, \ldots, \states\setminus\{t\}, \act}_{\text{$n$ times ``$\states\setminus\{t\}, \act$''}}, t}},
	\end{align*}
	where $\Cyl{\states\setminus\{t\}, \act, \ldots, \states\setminus\{t\}, \act, t}$ refers to the union over all cylinders of histories that go through a certain number of states that are not $t$, and then end in $t$.
	Event $\lnot \F t$ can therefore be expressed with complements and countable intersection of cylinders in $H$, so $\pr{\arena}{\st}{\strat_1}{\lnot\F t} = \pr{\arena}{\st}{\stratBis_1}{\lnot\F t}$, which proves the second claim.

	We now assume that $\pr{\arena}{\st}{\strat_1}{\lnot \F t} > 0$.
	Using the definition of conditional probabilities, to prove that $\pr{\arena}{\st}{\strat_1}{\blank \mid \lnot\F t} = \pr{\arena}{\st}{\stratBis_1}{\blank \mid \lnot\F t}$, it is left to prove that for all events $E\in\oalg_{(\arena, \st)}$, $\pr{\arena}{\st}{\strat_1}{E\cap \lnot \F t} = \pr{\arena}{\st}{\stratBis_1}{E\cap \lnot \F t}$.
	We first prove this equality if $E$ is a cylinder $\Cyl{\hist}$ for some $\hist\in\Hists(\arena, \st)$.
	If $\hist$ visits $t$, then $E\cap \lnot\F t = \emptyset$ and $\pr{\arena}{\st}{\strat_1}{E\cap \lnot \F t} = \pr{\arena}{\st}{\stratBis_1}{E\cap \lnot \F t} = 0$.
	If $\hist$ does not visit $t$, then $E$ is an element of~$H$.
	Therefore, $E\cap \lnot\F t$ can be expressed with complements and countable intersection of elements of~$H$, so $\pr{\arena}{\st}{\strat_1}{E\cap \lnot \F t} = \pr{\arena}{\st}{\stratBis_1}{E\cap \lnot \F t}$.
	Since distributions $\pr{\arena}{\st}{\strat_1}{\blank\cap \lnot \F t}$ and $\pr{\arena}{\st}{\stratBis_1}{\blank\cap \lnot \F t}$ are equal on all cylinders, they are also equal on all the events in $\oalg_{(\arena, \st)}$.
\end{proof}

\end{document}